\documentclass[final,11pt]{article}
\pdfoutput=1

\usepackage[thmstyles]{andrews}
\usepackage{andrews_bib}
\usepackage[inline]{showlabels}

\usepackage[vcentermath]{youngtab}

\usepackage[style=alphabetic, 
backref, 
url=false,
doi=false,
isbn=false,
useprefix=true,
maxbibnames=99,
maxcitenames=5,
mincitenames=3,
maxalphanames=5,
minalphanames=3]{biblatex}
\usepackage{changepage}

\usepackage[displaymath, mathlines]{lineno}

\newcommand{\supp}{\mathrm{supp}}
\newcommand{\lex}{\mathrm{lex}}

\newcommand{\IMM}{\mathrm{IMM}}
\newcommand{\sub}[2]{\mathrm{Sub}_{#1 \to #2}}
\newcommand{\detideal}[3]{I^{\mathrm{det}}_{#1,#2,#3}}
\newcommand{\pfaffideal}[2]{I^{\pfaff}_{#1,#2}}
\newcommand{\pfaff}{\mathrm{pfaff}}
\DeclareMathOperator{\multideg}{multideg}
\DeclareMathOperator{\LM}{LM}
\DeclareMathOperator{\LC}{LC}

\DeclareMathOperator{\Coeff}{Coeff}

\DeclareMathOperator{\relrk}{relrk}
\DeclareMathOperator{\orb}{orb}

\DeclareMathOperator{\Pf}{Pf}

\newbibmacro{string+doi}[1]{%
  \iffieldundef{doi}{#1}{\href{http://dx.doi.org/\thefield{doi}}{#1}}}
\DeclareFieldFormat{title}{\usebibmacro{string+doi}{\mkbibemph{#1}}}
\DeclareFieldFormat*{title}{\usebibmacro{string+doi}{\mkbibquote{#1}}}

\title{Ideals, Determinants, and Straightening: Proving and Using Lower Bounds for Polynomial Ideals}

\author{Robert Andrews\thanks{Department of Computer Science, University of Illinois Urbana-Champaign. Email: \texttt{rgandre2@illinois.edu}. Supported by NSF grants CCF-1755921 and CCF-1814788.} \and Michael A.~Forbes\thanks{Department of Computer Science, University of Illinois Urbana-Champaign. Email: \texttt{miforbes@illinois.edu}. Supported by NSF grants CCF-1755921, CCF-1814788, and CAREER award 2047310.}}
\date{October 27, 2022}

\begin{document}

\maketitle

\thispagestyle{empty}
\pagestyle{empty}

\begin{abstract}
	We show that any nonzero polynomial in the ideal generated by the $r \times r$ minors of an $n \times n$ matrix $X$ can be used to efficiently approximate the determinant.
	Specifically, for any nonzero polynomial $f$ in this ideal, we construct a small depth-three $f$-oracle circuit that approximates the $\Theta(r^{1/3}) \times \Theta(r^{1/3})$ determinant in the sense of border complexity.
	For many classes of algebraic circuits, this implies that every nonzero polynomial in the ideal generated by $r \times r$ minors is at least as hard to approximately compute as the $\Theta(r^{1/3}) \times \Theta(r^{1/3})$ determinant.
	We also prove an analogous result for the Pfaffian of a $2n \times 2n$ skew-symmetric matrix and the ideal generated by Pfaffians of $2r \times 2r$ principal submatrices.

	This answers a recent question of \textcite[Conjecture 6.3]{Grochow20} about complexity in polynomial ideals in the setting of border complexity.
	Leveraging connections between the complexity of polynomial ideals and other questions in algebraic complexity, our results provide a generic recipe that allows lower bounds for the determinant to be applied to other problems in algebraic complexity.
	We give several such applications, two of which are highlighted below.
	\begin{itemize}
		\item
			We prove new lower bounds for the Ideal Proof System of Grochow and Pitassi.
			Specifically, we give super-polynomial lower bounds for refutations computed by low-depth circuits.
			This extends the recent breakthrough low-depth circuit lower bounds of \textcite{LST21a} to the setting of proof complexity.
			Moreover, we show that for many natural circuit classes, the approximative proof complexity of our hard instance is governed by the approximative circuit complexity of the determinant.
		\item
			We construct new hitting set generators for the closure of low-depth circuits.
			For any $\eps > 0$, we construct generators with seed length $O(n^\eps)$ that hit $n$-variate low-depth circuits.
			Our generators attain a near-optimal tradeoff between their seed length and degree, and are computable by low-depth circuits of near-linear size (with respect to the size of their output).
			This matches the seed length of the generators recently obtained by \textcite{LST21a}, but improves on the degree and circuit complexity of the generator.
	\end{itemize}
\end{abstract}

\newpage

\setcounter{tocdepth}{2}
\tableofcontents

\newpage

\pagestyle{plain}
\setcounter{page}{1}

\section{Introduction}

A central goal of algebraic complexity theory is to understand the resources needed to compute multivariate polynomials in algebraic models of computation.
Typically, one attempts to determine the complexity of a single family of polynomials $\set{f_n(\vec{x}) : n \in \naturals}$, such as the $n \times n$ determinant or permanent.
A generalization of this task is to examine the complexity of a family of \emph{ideals} $\set{I_n \subseteq \F[\vec{x}] : n \in \naturals}$ of polynomials.
Recall that in a commutative ring $R$, an ideal $I \subseteq R$ is a subset of $R$ such that (1) if $a, b \in I$, then $a + b \in I$, and (2) if $a \in I$ and $r \in R$, then $ar \in I$.
Ideals naturally arise in commutative algebra and algebraic geometry; for example, the set of polynomials that vanish on a subset $V \subseteq \F^n$ is an ideal.
Closer to computer science and algebraic complexity, ideals appear in the study of polynomial identity testing, polynomial factorization, and algebraic proof complexity, though these appearances are not always made explicit.
Due to the prominence of ideals in algebra and algebraic complexity, it is both natural and worthwhile to study them from a complexity-theoretic perspective.

Every nonzero ideal contains polynomials of arbitrarily large circuit complexity.
This is a straightforward consequence of the fact that ideals are closed under multiplication by arbitrary polynomials.
A more interesting task, then, is to determine the minimum possible complexity of a nonzero polynomial in an ideal.

Unfortunately, little is known about the complexity of ideals aside from what is implicit in their connection to other problems of algebraic complexity.
A recent column by \textcite{Grochow20} surveyed these connections and posed some open questions, both general and concrete, about the complexity of ideals.
In particular, he raised the following question regarding an explicit family of ideals.

\begin{conjecture*}[{\cite[Conjecture 6.3]{Grochow20}}]
	Let $X$ be a $n \times n$ matrix of variables and let $I_{n}$ be the ideal generated by the $n/2 \times n/2$ minors of $X$.
	For every nonzero polynomial $f(X) \in I_n$, there is a small algebraic circuit with $f$-oracle gates that computes the $m \times m$ determinant for some $m = n^{\Theta(1)}$.
\end{conjecture*}

Due to the close relationship between the non-vanishing of minors and matrix rank, it is natural to conjecture that such a circuit exists.
If the oracle circuit is not restricted in any manner, then the desired circuit exists simply because the determinant can be computed efficiently by algebraic circuits.
However, if the oracle circuit is required to be, for example, a formula, then this question becomes nontrivial, as the determinant is not known to be computable by small formulas.

The main contribution of our work is to resolve this conjecture in the setting of approximate algebraic computation.

\begin{theorem*}
	Grochow's conjecture is true (with respect to border complexity).
\end{theorem*}

Specifically, we show that for any nonzero polynomial $f \in I_n$, the $\Theta(n^{1/3}) \times \Theta(n^{1/3})$ determinant can be approximately computed by a small depth-three $f$-oracle circuit with a single oracle gate.
A direct consequence of this is that for many circuit classes $\mathcal{C}$, if the determinant cannot be approximated by polynomial-size $\mathcal{C}$-circuits, then neither can any polynomial in the ideal $I_n$.
Naturally, this has applications to polynomial identity testing and algebraic proof complexity by employing the supporting role played by the complexity of ideals in those areas.

Before describing our results in more detail, we briefly survey what is known about the complexity of ideals and its connections to polynomial identity testing and algebraic proof complexity.

\subsection{The Complexity of Ideals}

Most of what is known about the complexity of ideals is limited to ideals generated by a single polynomial.
The ideal $\abr{f}$ generated by a polynomial $f(\vec{x})$ consists of all multiples of $f$, so questions about the complexity of this ideal become questions about the complexity of $f$ and its multiples.
Determining the minimum complexity of a polynomial in $\abr{f}$ amounts to determining whether there is a multiple of $f$ that is significantly easier to compute than $f$ itself.
This leads to the question of factoring algebraic circuits: given a small circuit computing a polynomial $g(\vec{x})$, can the factors of $g(\vec{x})$ be computed by small circuits?

This question was addressed in a celebrated result of \textcite{Kaltofen87} (with alternate proofs by \textcite[Theorem 2.21]{Burgisser2000} and \textcite{CKS19b}), who showed that factors (of low multiplicity) of small circuits can be computed by small circuits.
Taking the contrapositive, if $f(\vec{x})$ cannot be computed by small circuits, then neither can any polynomial $g \in \abr{f}$ which has $f$ as a factor of low multiplicity.
Polynomial factorization has since been studied in restricted algebraic circuit classes, including low-depth circuits \cite{DSY09,CKS19a}, formulas \cite{Oliveira16,DSS18}, algebraic branching programs \cite{DSS18,ST20}, and sparse polynomials \cite{BSV20}.
This is motivated in part by the use of Kaltofen's theorem to establish hardness-to-pseudorandomness results for polynomial identity testing, as done in the work of \textcite{KI04}.

Kaltofen's result gives us a strong understanding of the complexity of the low-degree polynomials in a principal ideal.
Because algebraic complexity theory is primarily interested in the computation of low-degree polynomials, this suffices for most applications.
However, the situation would be cleaner if lower bounds on the complexity of a polynomial $f$ implied comparable lower bounds on the complexity of all polynomials in the ideal $\abr{f}$, not just for those polynomials $g \in \abr{f}$ for which $f$ is a factor of low multiplicity.
\textcite{Kaltofen87} asked in the language of factorization whether this is the case; this question remains open and is now known as the Factor Conjecture.
In the setting of approximative algebraic computation, the analogue of the Factor Conjecture was proved by \textcite{Burgisser04}.
It is interesting to note that, coincidentally, we also make essential use of approximative computation in our work.

For non-principal ideals, much less is known.
\textcite{KW21} studied ideals generated by minors of a generic matrix, showing that every nonzero polynomial in the ideal generated by minors of size $r$ must have sparsity at least $r!/2$.
Later work by \textcite{DKW21} improved this sparsity lower bound to $r!$, which is optimal as witnessed by any $r \times r$ minor.
Under the assumption that $\VP \neq \VNP$, \textcite{KRST22} proved lower bounds on the complexity of any polynomial that vanishes on the coefficient vector of all polynomials in $\VNP$.
The remainder of what we know about the complexity of non-principal ideals stems from connections to polynomial identity testing and the Ideal Proof System.
We defer our explanation of these connections to \autoref{subsec:pit} and \autoref{subsec:ips}, respectively.

Approximate algebraic computation will play a key role in our work, so we briefly discuss it here.
For simplicity, we will focus on circuits and polynomials defined over the complex numbers; for more details, including a field-independent definition of approximate computation, see \autoref{subsec:border}.
We say that a polynomial $f(\vec{x})$ can be approximately computed by small algebraic circuits if there is a collection of polynomials $\set{f_\eps : \eps > 0}$ such that (1) for all $\eps > 0$, the polynomial $f_\eps$ can be computed by a small circuit, and (2) we have $\lim_{\eps \to 0}f_\eps = f$, where convergence is coefficient-wise.
Over the complex numbers, this can be interpreted as saying that $f$ lies in the closure (with respect to the Euclidean topology) of the set of polynomials computable by small circuits.
If $f$ can be approximated well by polynomials from a circuit class $\mathcal{C}$, then we say that $f$ is in $\overline{\mathcal{C}}$, the closure of $\mathcal{C}$.
The circuit complexity of the approximating polynomials $f_\eps$ is referred to as the \emph{border complexity} of $f$.
Naturally, one can also consider border complexity with respect to other classes of algebraic circuits, such as formulas or branching programs.

Border complexity appeared as early as the late 1970s, when \textcite{BCRL79,Bini80} improved upon the state-of-the-art algorithms for matrix multiplication by considering an approximative version of the problem.
The notion of border complexity also plays a prominent role in the geometric complexity theory program of \textcite{MS01}.
Roughly speaking, the goal of that program is to prove super-polynomial lower bounds on the border complexity of the permanent using techniques from algebraic geometry and representation theory.

In general, the relationship between exact and border complexity is not well-understood. 
\textcite{Forbes16} (see also \textcite{BDI21}) observed that exact and border complexity are equivalent for read-one oblivious algebraic branching programs.
\textcite{DDS21b} recently showed that polynomials in the border of depth-three circuits of bounded top fan-in can be computed exactly by small algebraic branching programs.
However, for classes like $\VP$ and $\VNP$ (the algebraic analogues of $\P$ and $\NP$), it is not clear how they relate to their closure.

Returning to the complexity of ideals, if we are content to operate in the setting of border complexity, then the work of \textcite{Burgisser04} shows that up to polynomial factors, the complexity of a principal ideal $\abr{f}$ is governed by the border complexity of its generator $f$.
Unfortunately, this seems to be where our understanding of the complexity of ideals stops.
Even ideals generated by two polynomials are not well-understood structurally from the viewpoint of complexity theory.
There are examples of explicit ideals, coming from polynomial identity testing, that are not principal and for which we can prove lower bounds; see \autoref{subsec:pit} below for more.

\subsection{Polynomial Identity Testing} \label{subsec:pit}

Polynomial identity testing (which we abbreviate as PIT) is the algorithmic problem of testing whether an algebraic circuit computes the zero polynomial.
Typically, one assumes that the circuit computes a polynomial of degree at most $n^{O(1)}$, where $n$ is the number of input variables.
A simple $\coRP$ algorithm for this problem follows from the Schwartz--Zippel lemma \cite{Zippel79,Schwartz80}.
When the input is allowed to be an algebraic circuit without further structural restrictions, no deterministic algorithm is known that improves on the na\"ive derandomization of this randomized algorithm.
In fact, even obtaining a nondeterministic algorithm running in subexponential time is known to imply circuit lower bounds that lie beyond the reach of current techniques \cite{KI04}.

More is known for many restricted classes of circuits, including sparse polynomials \cite{KS01}, depth-three \cite{DS07,KS07,KS09,KS11,SS11,SS12,SS13} and depth-four \cite{Shpilka19,PS20,PS21,DDS21a} circuits of bounded top fan-in, read-once formulas \cite{SV15,MV18}, read-once oblivious algebraic branching programs \cite{FS13,FSS14,AGKS15,GKS17,GKST17,AFSSV18,GG20,BS21}, low-depth multilinear circuits \cite{KMSV13,AMV15,OSV16,SV18}, and low-depth circuits \cite{LST21a}.
In general, algorithms for PIT are designed by giving an efficient construction of a \emph{hitting set generator}.
That is, we construct a low-degree polynomial map $\mathcal{G} : \F^\ell \to \F^n$ with $\ell \ll n$ such that if $f(\vec{x})$ is a nonzero polynomial computable by a small circuit, then $f(\mathcal{G}(\vec{y})) \neq 0$.
This reduces the number of variables in the circuit without increasing the degree too much.
We then obtain a faster deterministic algorithm by using the brute-force derandomization of the Schwartz--Zippel lemma to test $f(\mathcal{G}(\vec{y}))$.

In fact, constructing such a generator $\mathcal{G}$ corresponds to proving lower bounds against a polynomial ideal.
Fix a circuit class $\mathcal{C}$ (for example, the class of $n^2$-size circuits) and let $\mathcal{G}$ be a hitting set generator for $\mathcal{C}$.
Let $\mathcal{G}(\vec{y}) = (\mathcal{G}_1(\vec{y}),\ldots,\mathcal{G}_n(\vec{y}))$ and consider the ideal of polynomials $f(\vec{x})$ that vanish on $\mathcal{G}(\vec{y})$, i.e., polynomials such that $f(\mathcal{G}(\vec{y})) = 0$.
This ideal can be written as the intersection
\[
	I_{\mathcal{G}} \coloneqq \abr{x_i - \mathcal{G}_i(\vec{y}) : i \in [n]} \cap \F[\vec{x}],
\]
and in general is not generated by a single polynomial.
Suppose $f$ is a nonzero polynomial in the ideal $I_{\mathcal{G}}$.
Because we assumed $\mathcal{G}$ to be a hitting set generator for the circuit class $\mathcal{C}$, this means that $f$ cannot be computed by circuits from $\mathcal{C}$.
That is, proving that $\mathcal{G}$ is a generator for $\mathcal{C}$ is equivalent to proving that no element of $I_{\mathcal{G}}$ can be computed by a circuit from $\mathcal{C}$.
To the best of our knowledge, this connection accounts for all known examples of lower bounds for non-principal ideals.
We remark that this approach can prove lower bounds against ``natural'' non-principal ideals.
For example, \cite[Corollary 6.7]{FSTW16} easily generalizes to prove lower bounds against determinantal ideals for weak circuit classes.
However, this approach does not necessarily allow one to choose an ideal and subsequently prove a lower bound against that particular ideal.

One can also construct hitting set generators using lower bounds for ideals.
\textcite{KI04} used Kaltofen's factorization result to show that circuit lower bounds for explicit families of polynomials can be used to derandomize PIT.
In the analysis of the Kabanets--Impagliazzo generator, what is really needed is a lower bound for all low-degree multiples of a polynomial $f$, which is exactly what Kaltofen's theorem provides if $f$ is assumed to be hard to compute.
Further work on the algebraic hardness-randomness paradigm in the setting of low-depth circuits \cite{DSY09,CKS19a} followed the approach of \textcite{KI04}, proving analogues of Kaltofen's factoring result for bounded-depth circuits.

One can also consider PIT for polynomials of small border complexity.
Even in the randomized setting, the complexity of this problem is unclear, as it is not obvious how to evaluate a polynomial $f(\vec{x})$ given only a circuit that approximates $f(\vec{x})$, nor is it clear that such an approximating circuit even has a succinct description.
However, one can still try to construct hitting set generators for polynomials of small border complexity.
\textcite{FS18,GSS19} gave $\PSPACE$ constructions of hitting set generators for polynomials with small border circuit complexity.
One of the primary conceptual contributions of \textcite{FS18} was the definition of a \emph{robust} hitting set generator.
Roughly, a generator $\mathcal{G}$ for a class $\mathcal{C}$ is robust if for every nonzero polynomial $f \in \mathcal{C}$, the composition $f(\mathcal{G}(\vec{y}))$ is ``far'' from the zero polynomial (after $f$ has been suitably normalized).
It is not hard to show that, over a field of characteristic zero, a generator $\mathcal{G}$ for $\mathcal{C}$ is robust if and only if $\mathcal{G}$ hits the closure $\overline{\mathcal{C}}$ of $\mathcal{C}$.
Over an arbitrary field, one can likewise consider the problem of constructing hitting set generators for the closures of circuit classes, although the notion of $f(\mathcal{G}(\vec{y}))$ being far from the zero polynomial is not as clear. 
In this setting we drop the adjective ``robust'' and focus simply on hitting sets for the closure of a circuit class.
The preceding discussion on the relationship between PIT and the complexity of ideals extends to border complexity.

Designing hitting sets for the closures of circuit classes has been explored as a possible avenue towards resolving grand challenges in polynomial identity testing.
Recent work by \textcite{MS21,ST21b} studied PIT for \emph{orbits} of various classes $\mathcal{C}$.
The orbit $\orb(\mathcal{C})$ of a class $\mathcal{C}$ corresponds to polynomials of the form $f(A\vec{x} + \vec{b})$, where $f(\vec{x}) \in \mathcal{C}$ and $A$ is an invertible $n \times n$ matrix.
Studying PIT for orbits is motivated by the fact that for many simple classes $\mathcal{C}$, there is a far richer class $\mathcal{D}$ such that $\overline{\orb(\mathcal{C})} = \overline{\mathcal{D}}$.
That is, in order to derandomize PIT for a powerful class $\mathcal{D}$, it suffices to construct hitting set generators for the closure of the much simpler class $\orb(\mathcal{C})$.
Unfortunately, this is not always feasible; for example, \textcite{MS21} showed that at least one instantiation of their hitting sets does not extend to the closure of the circuit class it hits.

\subsection{The Ideal Proof System} \label{subsec:ips}

A central question of proof complexity is the following: given an unsatisfiable CNF formula $\varphi$, what is the length of the shortest proof of the unsatisfiability of $\varphi$?
This question can be instantiated with a myriad of different proof systems rooted in logic, algebra, and geometry.
Our focus in this work will be on a proof system based in algebra, namely the Ideal Proof System of \textcite{GP18}.
For a more comprehensive treatment of other proof systems (and proof complexity in general), see the recent book of \textcite{Krajicek19}.

Let $\varphi$ be an unsatisfiable 3CNF formula.
One way to prove that $\varphi$ is unsatisfiable is to translate $\varphi$ into a system of polynomial equations, swapping the roles of 0 and 1, as follows.
The literals $x$ and $\neg x$ are translated into the polynomials $1 - x$ and $x$, respectively.
A clause $\ell_1 \lor \ell_2 \lor \ell_3$ becomes the polynomial $p_{\ell_1} p_{\ell_2} p_{\ell_3}$, where $p_{\ell_i}$ is the polynomial corresponding to the literal $\ell_i$.
Let $f_1,\ldots,f_m$ be the polynomials obtained from the clauses of $\varphi$. 
It is not hard to see that $\varphi$ is satisfiable if and only if there is a $\set{0,1}$-valued solution to the system of equations $f_1 = \cdots = f_m = 0$; equivalently, $\varphi$ is satisfiable if and only if there is a solution to the system $f_1 = \cdots = f_m = x_1^2-x_1 = \cdots = x_n^2 - x_n = 0$.

Thus, to show that $\varphi$ is unsatisfiable, it suffices to prove that a system of polynomial equations is unsatisfiable.
This can be done by finding polynomials $g_1(\vec{x}),\ldots,g_m(\vec{x})$ and $h_1(\vec{x}),\ldots,h_n(\vec{x})$ such that $\sum_{i=1}^m g_i(\vec{x}) f_i(\vec{x}) + \sum_{i=1}^n h_i(\vec{x}) (x_i^2-x_i) = 1$, or more succinctly, by showing that $1$ is in the ideal generated by $\set{f_1,\ldots,f_m,x_1^2-x_1,\ldots,x_n^2-x_n}$.
As a consequence of Hilbert's Nullstellensatz, such a refutation always exists, provided the system is unsatisfiable.
These refutations and various notions of their complexity give rise to the Nullstellensatz \cite{BIKPP96} and Polynomial Calculus \cite{CEI96} proof systems, both of which are well-studied and for which lower bounds are known \cite{BIKPP96,BIKPRS96,Razborov98,IPS99}.  

The recent Ideal Proof System (abbreviated as IPS) of \textcite{GP18} measures the complexity of a refutation by the algebraic circuit complexity of the certificate $\sum_i g_i f_i + \sum_i h_i (x_i^2-x_i)$ when the $f_i$ and $x_i^2-x_i$ are provided as part of the input to the circuit.
Because a refutation in the IPS is written as an algebraic circuit, there are connections between algebraic circuit lower bounds and lower bounds for the IPS.
\textcite{GP18} proved that super-polynomial lower bounds on the size of IPS refutations of a family of CNF formulas imply $\VP \neq \VNP$.
As a proof system, the IPS is very powerful: \textcite{GP18} showed that the IPS polynomially simulates Extended Frege, itself a strong logic-based proof system.
This simulation also behaves nicely if we consider IPS refutations coming from a restricted circuit class $\mathcal{C}$.
For example, over a field of characteristic $p > 0$, the constant-depth version of the IPS polynomially simulates $\AC^0[p]$-Frege, a proof system notorious for its current lack of super-polynomial lower bounds.

Lower bounds, both conditional and unconditional, are known for the IPS.
Conditionally, \textcite{AGHT20} showed that the Shub--Smale hypothesis implies super-polynomial lower bounds on the size of IPS refutations of a particular instance of subset sum.
Later work by \textcite{ST21a} showed that over finite fields, if there is an explicit family of polynomials that cannot be computed by polynomial-size algebraic circuits, then a particular family of CNF formulas cannot be refuted by polynomial-size IPS refutations.
Combined with earlier work by \textcite{GP18}, this establishes that over finite fields, proving super-polynomial lower bounds for the IPS is equivalent to proving super-polynomial lower bounds for algebraic circuits.
\textcite{FSTW16} used techniques from algebraic circuit complexity to prove unconditional lower bounds for restricted subsystems of the IPS, including those computed by depth-three powering formulas, read-once algebraic branching programs, and multilinear formulas.

The Ideal Proof System is defined in terms of algebraic circuits, so it is natural to expect progress on IPS lower bounds to mirror progress on lower bounds for algebraic circuits.
Empirically, this has been the case, although additional effort is required to translate circuit lower bounds into IPS lower bounds.
To prove circuit lower bounds, one only needs to show that a single polynomial cannot be computed by small circuits.
In contrast, to prove lower bounds on the circuit size of IPS refutations of a system of polynomials, it is necessary to show that small circuits cannot compute any valid refutation.

Luckily, the set of IPS refutations of a fixed system of equations exhibits some algebraic structure: all refutations of a fixed system of polynomials lie in a coset of a particular ideal, as observed by \textcite[Section 6]{GP18}.
Thus, one can try to prove lower bounds for the IPS by proving circuit lower bounds for nonzero cosets of ideals.
To the best of our knowledge, the only known lower bounds for nonzero cosets of ideals are those that follow from previously-mentioned lower bounds on the IPS.
Notably, these proofs do not directly establish lower bounds for cosets of ideal, but rather reduce the task of proving IPS lower bounds to the more-tractable task of proving algebraic circuit lower bounds.
One could hope that by better understanding the complexity of (cosets of) ideals, this progress could be used to prove lower bounds for IPS and restricted variants thereof.
We refer the interested reader to \textcite{GP18,Grochow20} for further details.

For more on the Ideal Proof System, see the recent survey of \textcite{PT16}.

\subsection{Our Results}

We now describe our results in more detail.
Throughout this subsection, we let $X$ denote an $n \times m$ matrix of variables and $\detideal{n}{m}{r} \subseteq \F[X]$ the ideal generated by the $r \times r$ minors of $X$.
For simplicity, we state our results over fields of characteristic zero (such as the rational or complex numbers).

\subsubsection{Complexity of Determinantal Ideals}

Our main theorem constructs, for any nonzero polynomial $f(X) \in \detideal{n}{m}{r}$, a small $f$-oracle circuit that approximately computes the $s \times s$ determinant for $s = \Theta(r^{1/3})$.
This answers a question of \textcite[Conjecture 6.3]{Grochow20} in the setting of border complexity.

\begin{theorem}[Informal version of \autoref{thm:proj to small abp} and \autoref{cor:proj to det}] \label{thm:informal main}
	Let $\F$ be a field of characteristic zero.
	Let $X$ be an $n \times m$ matrix of variables and let $\detideal{n}{m}{r} \subseteq \F[X]$ be the ideal generated by the $r \times r$ minors of $X$.
	Let $f(X) \in \detideal{n}{m}{r}$ be a nonzero polynomial.
	Then there is a depth-three $f$-oracle circuit of size $O(n^2 m^2)$ that approximately computes the $s \times s$ determinant for $s = \Theta(r^{1/3})$.
\end{theorem}

More generally, the conclusion of \autoref{thm:informal main} holds if the determinant is replaced by any polynomial $g$ that can be approximately computed by an algebraic branching program with $r$ vertices.
The conclusion of \autoref{thm:informal main} also holds if we have oracle gates that approximately compute $f$ instead of oracles that compute $f$ exactly.

An immediate consequence of \autoref{thm:informal main} is that for formulas and low-depth circuits, the border complexity of any nonzero polynomial in $\detideal{n}{m}{r}$ is at least as large as the border complexity of the $\Theta(r^{1/3}) \times \Theta(r^{1/3})$ determinant, up to polynomial factors.
To the best of our knowledge, the only complexity lower bounds for the ideal $\detideal{n}{m}{r}$ known prior to this work are due to \textcite{KW21,DKW21} and \textcite[Corollary 6.7]{FSTW16}, who showed that every nonzero polynomial in $\detideal{n}{m}{r}$ is $\exp(\Omega(r))$-hard for several weak circuit classes.

To prove \autoref{thm:informal main}, we have to reason about arbitrary polynomials in $\detideal{n}{m}{r}$.
That is, if $\set{g_1,\ldots,g_N}$ are the $r \times r$ minors of $X$, we have to consider all nonzero polynomials of the form $\sum_{i=1}^N f_i g_i$, where the $f_i$ are arbitrary polynomials.
This is difficult in part because if we apply a linear change of variables $X \mapsto L(X)$, it is not clear how to control the behavior of the $f_i$.
To circumvent this, we use an alternate basis for $\F[X]$ instead of the monomial basis.
This alternate basis consists of products of minors (of possibly different sizes) of $X$ that satisfy a particular combinatorial condition; these products are known as \emph{standard bideterminants}.
Working in this basis, we gain a better understanding of how the multiplicands $f_i$ behave under a change of variables.

The proof of \autoref{thm:informal main} then proceeds in two steps.
First, we find a change of variables that takes a polynomial $f \in \detideal{n}{m}{r}$ to an approximation (in the border complexity sense) of a standard bideterminant $h(X)$ in the support of $f$.
The analysis of this step crucially relies on the use of the standard bideterminant basis and its properties, which we describe in \autoref{subsec:bidet}.
Because $f$ lies in the ideal $\detideal{n}{m}{r}$, one can show that $h(X)$ is divisible by a $t \times t$ minor of $X$ for some $t \ge r$.
The second step is to find a projection of $h(X)$ to the $\Theta(r^{1/3}) \times \Theta(r^{1/3})$ determinant.
Since $h$ may be a product of minors of varying sizes, we need to find a projection that (1) behaves nicely on small minors of $X$ and (2) allows us to deal with the possibility that $h$ may be a large power of a minor.
We accomplish this by modifying an argument of \textcite{Valiant79}.

\subsubsection{Complexity of Pfaffian Ideals}

Let $Y$ be a $2n \times 2n$ skew-symmetric matrix.
It is well-known that the determinant of $Y$ is the square of another polynomial, the \emph{Pfaffian} $\Pf(Y)$ of $Y$.
Let $\pfaffideal{2n}{2n} \subseteq \F[Y]$ be the ideal generated by the Pfaffians of the $2r \times 2r$ principal submatrices of $Y$.
Our next result is an analogue of \autoref{thm:informal main} for the ideal $\pfaffideal{2n}{2r}$.

\begin{theorem}[Informal version of \autoref{thm:pfaff to small abp} and \autoref{cor:proj to pfaff}] \label{thm:informal pfaff}
	Let $\F$ be a field of characteristic zero.
	Let $Y$ be a $2n \times 2n$ skew-symmetric matrix of variables and let $\pfaffideal{2n}{2r} \subseteq \F[Y]$ be the ideal generated by the Pfaffians of the $2r \times 2r$ principal submatrices of $Y$.
	Let $f(Y) \in \pfaffideal{2n}{2r}$ be a nonzero polynomial.
	Then there is a depth-three $f$-oracle circuit of size $O(n^4)$ that approximately computes the $s \times s$ Pfaffian for $s = \Theta(r^{1/3})$.
\end{theorem}

The proof of \autoref{thm:informal pfaff} is similar to that of \autoref{thm:informal main}.
The primary difference is that we now express polynomials in $\pfaffideal{2n}{2r}$ in an alternate basis consisting of products of Pfaffians of principal submatrices of $Y$.
Along the way, we modify some of the technical details of the construction to accommodate for Pfaffians instead of determinants.

We remark that because the Pfaffian is the square root of the skew-symmetric determinant (in the sense that $\Pf(Y)^2 = \det(Y)$), it is natural to attempt proving \autoref{thm:informal pfaff} using \autoref{thm:informal main}.
For any polynomial $f(\vec{x})$, one can use the Taylor series expansion of $\sqrt{1 + x^2}$ to construct a small $f(\vec{x})^2$-oracle circuit that computes $f(\vec{x})$.
Combining this with \autoref{thm:informal main}, one obtains an analogue of \autoref{thm:informal main} for the ideal generated by the squares of sub-Pfaffians of $Y$, which is weaker than \autoref{thm:informal pfaff} above.

\subsubsection{The Space of Partial Derivatives in Determinantal Ideals}

The remainder of our work consists of three applications of \autoref{thm:informal main} and its proof, the first of which is to algebraic circuit complexity.
For a polynomial $f \in \F[X]$, let $\partial_{< \infty}(f)$ denote the span of the partial (Hasse) derivatives of $f$.
The dimension of $\partial_{< \infty}(f)$ and related spaces has been used successfully as a complexity measure in proving lower bounds for restricted classes of algebraic circuits (see the survey of \textcite{Saptharishi19} for more on this).
While \autoref{thm:informal main} shows that computing a polynomial in $\detideal{n}{m}{r}$ is not much harder than computing the $\Theta(r^{1/3}) \times \Theta(r^{1/3})$ determinant, it is natural to ask if there are polynomials in $\detideal{n}{m}{r}$ that are ``simpler'' than the $r \times r$ determinant with respect to complexity measures like $\dim(\partial_{< \infty}(\bullet))$.
Our next result shows that among nonzero polynomials in the ideal $\detideal{n}{m}{r}$, the $r \times r$ determinant in fact minimizes the value of $\dim(\partial_{< \infty}(\bullet))$.

\begin{theorem}[Informal version of \autoref{thm:pd lb}] \label{thm:informal pd}
	For every nonzero $f(X) \in \detideal{n}{m}{r}$, we have $\dim(\partial_{< \infty}(f)) \ge \dim(\partial_{< \infty}(\det_r)) = \binom{2r}{r}$.
\end{theorem}

Using tools developed in the proof of \autoref{thm:informal main}, we can easily reduce the task of proving \autoref{thm:informal pd} to the case where $f(X)$ is a product of minors of $X$.
As $f$ is in the ideal $\detideal{n}{m}{r}$, at least one factor of $f$ must be an $s \times s$ minor of $X$ for some $s \ge r$.
We can then directly bound $\dim(\partial_{< \infty}(f))$ from below by a slight generalization of the argument used to bound $\dim(\partial_{< \infty}(\det_s))$.

We note that one can easily prove a lower bound of $\dim(\partial_{< \infty}(f)) \ge 2^r$ using observations due to \textcite{FSTW16} (see \autoref{sec:partials} for details).
Our result improves on this, obtaining an optimal bound of $\binom{2r}{r} = \Theta(4^r / \sqrt{r})$.

\subsubsection{Polynomial Identity Testing for Low-Depth Circuits and Formulas}

Next, we use \autoref{thm:informal main} to derandomize special cases of polynomial identity testing.
It is a straightforward consequence of \autoref{thm:informal main} that for circuit classes like low-depth circuits and formulas, computing any nonzero element of $\detideal{n}{m}{r}$ is effectively as hard as computing the $\Theta(r^{1/3}) \times \Theta(r^{1/3})$ determinant.
Over an algebraically closed field, the ideal $\detideal{n}{m}{r}$ can be equivalently described as the ideal of polynomials that vanish on matrices of rank less than $r$.
Using this alternate description, we construct hitting set generators that unconditionally hit the closure of small low-depth circuits and conditionally hit the closure of small formulas.

\begin{theorem}[Informal version of \autoref{thm:small depth hsg} and \autoref{thm:formula hardness-randomness}] \label{thm:informal hsg}
	Let $\F$ be a field of characteristic zero.
	For every $k \in \naturals$, there is a hitting set generator $\mathcal{G}_k$ with seed length $n^{1/2^k + o(1)}$ and degree $2^k$ that hits the closure of polynomial-size low-depth algebraic circuits.
	The generator $\mathcal{G}_k$ can be computed by either (1) a circuit of product-depth $k$ and size $n^{1 + o(1)}$, (2) a formula of size $n^{1 + o(1)}$, or (3) a circuit of size $n \log^{O(1)} n$.
	Assuming the border formula complexity of the determinant is super-polynomial, the generator $\mathcal{G}_k$ is also a hitting set generator for the closure of polynomial-size algebraic formulas.
\end{theorem}

Our hitting set generators are very simple to describe.
For $k = 1$, our generator takes as input two matrices of variables $Y$ and $Z$, where $Y$ is a $\sqrt{n} \times n^{o(1)}$ matrix and $Z$ is an $n^{o(1)} \times \sqrt{n}$ matrix, and outputs the product $YZ$.
For $k \ge 2$, we construct the generator $\mathcal{G}_k$ by arranging the input variables of $\mathcal{G}_{k-1}$ into a square matrix and replacing them with the product of an $n^{1/2^k + o(1)} \times n^{o(1)}$ matrix and an $n^{o(1)} \times n^{1/2^k + o(1)}$ matrix.

To prove that our generators correctly hit polynomial-size low-depth circuits, we must show that every small low-depth circuit does not vanish on the output of our generator.
Using the description of $\detideal{n}{m}{r}$ as the ideal of polynomials vanishing on matrices of rank at most $r$, establishing the correctness of our generators equates to proving that no small low-depth circuit can compute a polynomial in the ideal $\detideal{\sqrt{n}}{\sqrt{n}}{n^{o(1)}}$.
Such a lower bound follows in a straightforward manner by combining our \autoref{thm:informal main} with the recent breakthrough lower bounds of \textcite{LST21a}.

In the regime of $n^{\Theta(1)}$ seed length, our generators attain a near-optimal tradeoff between seed length and degree.
It is not hard to show that a generator of seed length $n^{1/2^k + o(1)}$ must be of degree at least $2^k$, and conversely that any generator of degree $2^k$ must have seed length at least $\Omega(n^{1/2^k})$ (see \autoref{lem:seed length lb}).
We also note that the circuit complexity of our generators is near-optimal, as any function with $n$ outputs necessarily requires size $\Omega(n)$ to compute.

Prior to this, the best-known hitting set generator for low-depth circuits was given by \textcite{LST21a}, using the hardness-randomness results of \textcite{CKS19a}.
They obtained, for all fixed $\eps > 0$, a generator with seed length $O(n^\eps)$ and degree $O(\log n / \log \log n)$.
Our construction attains the same seed length, but improves on the degree (as remarked above) and the circuit complexity of the generator.
When instantiated to hit circuits of size $s$, the generator of \textcite{LST21a} necessarily has circuit complexity $\Omega(s)$.
In contrast, our generator can be computed by a constant-depth circuit or formula of size $n^{1 + o(1)}$ or a circuit of size $n \log^{O(1)} n$, even when hitting low-depth circuits of size $O(n^{10^{100}})$.

For formulas, the best-known (conditional) constructions of hitting set generators prior to our work are due to \textcite{DSY09,CKS19a}.
Both works yield generators with parameters similar to the low-depth generator of \textcite{LST21a} mentioned above (although the generator of \cite{DSY09} can only hit formulas of small individual degree).
While our construction has better parameters, we use a stronger hardness assumption than what is needed by prior work.
The constructions of \textcite{DSY09,CKS19a} can be instantiated with any explicit family of polynomials that requires formulas of super-polynomial size.
In contrast, our construction depends crucially on super-polynomial lower bounds on the border formula complexity of the determinant.
This is a stronger assumption, as the determinant is computable by polynomial-size branching programs and circuits, a fact which likely does not hold for all explicit families of polynomials.

\subsubsection{Lower Bounds for the Ideal Proof System}

Finally, we use \autoref{thm:informal main} to prove lower bounds for the Ideal Proof System.
Let $X$ and $Y$ be $n \times n$ matrices of variables and let $I_n$ be the $n \times n$ identity matrix.
Consider the system of polynomial equations given by $\set{\det_n(X) = 0, XY - I_n = 0}$.
This system is unsatisfiable, as $\det_n(X) = 0$ implies that $X$ is non-invertible, while $XY - I_n = 0$ implies that $X$ is invertible with inverse $Y$.

More generally, one can replace the equation $\det_n(X) = 0$ with an encoding of the statement ``$\rank(X) < r$'' for some $r \le n$.
One such encoding is given by requiring that the $r \times r$ minors of $X$ vanish.
While this encoding is natural from a mathematical perspective, it consists of $\binom{n}{r}^2$ equations.
Our lower bounds will be of order $n^{(\log n)^{\Omega(1)}}$, so the number of equations in this natural encoding quickly eclipses the lower bound.
Instead, we use the rank condensers of \textcite{FS12} to encode ``$\rank(X) < r$'' using only $O(nr)$ equations.
In what follows, we abbreviate this succinct encoding as $[\rank(X) < r]$.

We show that the constant-depth version of the Ideal Proof System cannot efficiently refute the system $\set{[\rank(X) < r], XY - I_n = 0}$ when $r \ge n^{\Omega(1)}$.
Assuming lower bounds on the border formula complexity of the determinant, we also show that formula-IPS cannot efficiently refute this system.
We remark that our lower bounds also hold when the boolean axioms $x_{i,j}^2 - x_{i,j} = 0$ are included in the system of equations, but we suppress these here for brevity.

\begin{theorem}[Informal version of \autoref{cor:ips vac0 lb} and \autoref{thm:ips vf lb}] \label{thm:informal ips}
	Let $\F$ be a field of characteristic zero and let $r \ge n^{\Omega(1)}$.
	Let $X$ and $Y$ be $n \times n$ matrices of variables and let $I_n$ be the $n \times n$ identity matrix.
	Then any IPS refutation of the system $\set{[\rank(X) < r], XY - I_n = 0}$ cannot be approximately computed by a constant-depth circuit of polynomial size.
	Assuming the border formula complexity of the determinant is super-polynomial, then any IPS refutation of this system cannot be approximately computed by a formula of polynomial size.
\end{theorem}

The proof of \autoref{thm:informal ips} follows the approach of \textcite{FSTW16}, who showed that lower bounds for the IPS can be derived from circuit lower bounds for multiples of a polynomial.
Our choice of the system $\set{\det_n(X) = 0, XY - I_n = 0}$ is motivated by the fact that, using the techniques of \cite{FSTW16}, the desired IPS lower bounds follow from circuit lower bounds for multiples of the determinant.
By a suitable generalization of this technique, we show that lower bounds on the size of IPS refutations of $\set{[\rank(X) < r], XY - I_n = 0}$ follow from lower bounds on the complexity of nonzero polynomials in the ideal $\detideal{n}{n}{r}$.
We can obtain the necessary lower bounds by combining our \autoref{thm:informal main} with lower bounds against the determinant.
In the case of low-depth circuits, our IPS lower bounds are unconditional thanks to the recent breakthrough circuit lower bounds of \textcite{LST21a}.
For formula-IPS, our lower bounds remain conditional.

We also show that in the case of $r = n$, computing an IPS refutation of the hard instance $\set{\det_n(X) = 0, XY - I_n = 0}$ reduces to computing the determinant.
Namely, we give a small depth-three circuit with $\det_n$-oracle gates that computes an IPS refutation of our hard instance.
Passing to border complexity (using \autoref{lem:exact to approx oracle}), this shows that the approximative complexity of the smallest IPS refutation of $\set{\det_n(X) = 0, XY - I_n = 0}$ is sandwiched between the approximative complexity of the $\Theta(n^{1/3}) \times \Theta(n^{1/3})$ and $n \times n$ determinants.

The strongest unconditional lower bounds for the IPS prior to our work are due to \textcite{FSTW16}, who proved lower bounds for subsystems of the IPS computed by restricted classes of circuits, including read-once oblivious algebraic branching programs and multilinear formulas.
\textcite{IMP20} showed that the constant-depth version of Polynomial Calculus (PC) over finite fields is surprisingly strong.
The size of a constant-depth IPS refutation is essentially the number of lines in a constant-depth PC refutation, so lower bounds for constant-depth IPS over finite fields imply comparable lower bounds for constant-depth PC.
However, our lower bounds do not extend to finite fields, nor do our lower bounds hold for refutations of an unsatisfiable CNF, so we are unable to conclude lower bounds for constant-depth PC and related proof systems.

We also mention a recent work of \textcite{Alekseev21}, who proved lower bounds on the bit-size of refutations in a version of PC augmented with an extension rule.
This is somewhat incomparable to our result: Alekseev's proof system allows for proofs of arbitrary depth, but must pay to use constants of large bit complexity; on the other hand, we work with a low-depth proof system that can use arbitrary rational numbers (or even arbitrary complex numbers) for free.
Our lower bound is on circuit size, which is analogous to the number of lines in PC, whereas Alekseev's lower bound is on the number of bits needed to write down a refutation, which does not necessarily imply a lower bound on the number of proof lines.

\section{Preliminaries}

For a natural number $n \in \naturals$, we write $[n] \coloneqq \set{1, 2, \ldots, n}$.
We use $\vec{x} = (x_1,\ldots,x_n)$ to denote a vector of variables and $X = (x_{i,j})_{i \in [n], j \in [m]}$ to denote a matrix of variables.
For a matrix $A \in \F^{n \times m}$ and sets $R \subseteq [n]$, $C \subseteq [m]$, we denote by $A_{R,C}$ the submatrix of $A$ whose rows and columns are taken from the sets $R$ and $C$, respectively.
A submatrix $A_{R,C}$ is \emph{principal} if $R=C$.
Given a polynomial $f(\vec{x}) \in \F[\vec{x}]$, it will often be useful to view the variables $\vec{x}$ as the entries of a matrix, typically of size $\ceil{\sqrt{n}} \times \ceil{\sqrt{n}}$.
The precise way in which the variables $\vec{x}$ are arranged into a matrix will not matter, so we will perform this rearrangement implicitly without specifying the details.
If $X$ is an $n \times m$ matrix of variables, then for $r \le \min(n,m)$ we denote by $\detideal{n}{m}{r} \subseteq \F[X]$ the ideal of $\F[X]$ generated by the $r \times r$ minors of $X$.

We endow $\F[X]$ with a $(\naturals^n \oplus \naturals^m)$-grading in the following way.
Let $\vec{e}_i \in \naturals^n$ denote the element of $\naturals^n$ with 1 in the $i$\ts{th} position and zeroes elsewhere.
By abuse of notation, we also use $\vec{e}_i$ to denote the corresponding element of $\naturals^m$.
We assign degree $\vec{e}_i \oplus \vec{e}_j$ to the variable $x_{i,j}$ and extend this to $\F[X]$ in the natural way.
The degree of an element $f \in \F[X]$ with respect to this grading is called the \emph{multidegree} of $f$, written $\multideg(f)$.
We say an element of $\F[X]$ is \emph{multihomogeneous} if it is homogeneous with respect to this grading.

Recall that given a field $\F$ and an indeterminate $x$, we write
\begin{itemize}
	\item
		$\F[x]$ for the ring of polynomials in $x$ with coefficients from $\F$,
	\item
		$\F(x)$ for the field of rational functions in $x$ with $\F$-coefficients,
	\item
		$\F\llb x \rrb$ for the ring of formal power series in $x$ over $\F$, and
	\item
		$\F((x))$ for the field of formal Laurent series in $x$ over $\F$ (equivalently, the field of fractions of $\F \llb x \rrb$).
\end{itemize}

We assume familiarity with the basic notion of an algebraic circuit and restricted classes thereof, including formulas, branching programs, and bounded-depth circuits.
The interested reader may consult the surveys of \textcite{SY10,Saptharishi19} or the text of \textcite{BCS97} for more on algebraic circuits.

\subsection{Border Complexity} \label{subsec:border}

We now define border complexity, a modification of the standard notion of algebraic complexity.

\begin{definition}
	Let $\F$ be any field and let $\eps$ be an indeterminate.
	Let $f(\vec{x}) \in \F[\vec{x}]$.
	We say that an algebraic circuit $C$ \emph{border computes} $f$ if $C$ is defined over $\F((\eps))$ and computes a polynomial in $\F\llb \eps \rrb [\vec{x}]$ such that
	\[
		C(\vec{x}) = f(\vec{x}) + \eps g(\vec{x})
	\]
	for some $g(\vec{x}) \in \F\llb\eps\rrb [\vec{x}]$.
	We abbreviate this as $C(\vec{x}) = f(\vec{x}) + O(\eps)$.
	The \emph{border complexity of $f$} is the size of the smallest circuit $C$ that border computes $f$.
\end{definition}

If $\mathcal{C} \subseteq \F[\vec{x}]$ is a set of polynomials computed by some class of circuits, we denote by $\overline{\mathcal{C}} \subseteq \F[\vec{x}]$ the set of polynomials computed by the border of this same set of circuits.
For example, $\VP$ denotes the class of $n$-variate polynomials that have $n^{O(1)}$ degree and can be computed by circuits of $n^{O(1)}$ size, while $\overline{\VP}$ denotes $n$-variate polynomials of degree $n^{O(1)}$ that can be border computed by circuits of $n^{O(1)}$ size.

Over fields of characteristic zero, one can interpret border complexity as a notion of approximate computation.
In this case, if $C(\vec{x}) = f(\vec{x}) + O(\eps)$, then $\lim_{\eps \to 0} C(\vec{x}) = f(\vec{x})$, so $C$ computes a polynomial that coefficient-wise approximates $f$ arbitrarily well as $\eps$ goes to zero.
Since the circuit $C$ is defined over $\F((\eps))$, it may be the case that $C$ is not well-defined when $\eps = 0$, as intermediate computations may involve division by $\eps$.
This prohibits setting $\eps = 0$ in order to obtain a circuit that computes $f$ exactly.

When the underlying field $\F$ has positive characteristic (for example, when $\F$ is finite), this notion of approximation breaks down.
However, we can consider ``approximate'' computation in the symbolic sense defined above, which is still meaningful.

Alternatively, one can define border complexity using only the polynomial ring $\F[\eps]$, avoiding the use of $\F\llb \eps \rrb$ and $\F((\eps))$.
In this modified definition, we say that a circuit $C$ border computes $f(\vec{x})$ if $C$ is defined over $\F[\eps]$ and there is a polynomial $g(\vec{x}) \in \F[\eps][\vec{x}]$ and a natural number $q \in \naturals$ such that
\[
	C(\vec{x}) = \eps^q f(\vec{x}) + \eps^{q+1} g(\vec{x}).
\]
We abbreviate this as $C(\vec{x}) = \eps^q f(\vec{x}) + O(\eps^{q+1})$.
It turns out that these notions are equivalent, as one can translate between them by appropriately modifying the constants appearing in the circuit; see \textcite[Lemma 5.6(1)]{Burgisser04} for a proof.
(Note that the statement of \cite[Lemma 5.6(1)]{Burgisser04} only claims equivalence up to a factor of 2 in complexity.
This arises due to the fact that the model of straight-line programs used in \cite{Burgisser04} charges for scalar multiplications, whereas we allow multiplication by scalars for free.)

Given a set of polynomials $F \coloneqq \set{f_1,\ldots,f_k} \subseteq \F[\vec{x}]$, one can also define the border complexity of $F$ to be the size of the smallest multi-output circuit $C(\vec{x})$ over $\F((\eps))$ such that $C$ outputs $\set{f_1 + O(\eps),\ldots,f_k + O(\eps)}$.
Naturally, one can also consider (single- or multi-output) border complexity with respect to subclasses of algebraic circuits, such as formulas, branching programs, or constant-depth circuits.

It will be useful to make the dependence of a polynomial on the approximation parameter $\eps$ explicit.
In this case, we may write $f(\vec{x},\eps)$ for a polynomial in $\F\llb\eps\rrb [\vec{x}]$ or $\F[\eps][\vec{x}]$, even though $\eps$ is regarded as an element of the underlying ring and is not a variable.
This affords convenient notation for applying the map $\eps \mapsto \eps^N$ for some $N \in \naturals$ or the map $\delta \mapsto \eps^N$ for a second indeterminate $\delta$.
We can use this to compose approximations as in the lemma below.

\begin{lemma}[{\cite[Lemma 2.3(1)]{Burgisser04}}] \label{lem:semicontinuity}
	Let $f(\vec{x}) \in \F[\vec{x}]$.
	Suppose 
	\begin{enumerate}
		\item
			$\Phi$ is a circuit over $\F((\eps))[\vec{x}]$ such that $\Phi(\vec{x},\eps) = f(\vec{x}) + O(\eps) \in \F\llb \eps \rrb [\vec{x}]$, and
		\item
			$\Psi$ is a circuit over $\F((\delta))((\eps))[\vec{x}]$ such that $\Psi(\vec{x},\eps,\delta) = \Phi(\vec{x},\eps) + O(\delta) \in \F\llb \delta \rrb ((\eps))[\vec{x}]$.
	\end{enumerate}
	Then there is some sufficiently large $N \in \naturals$ such that $\Psi(\vec{x},\eps,\eps^N) = f(\vec{x}) + O(\eps) \in \F\llb\eps\rrb [\vec{x}]$.
\end{lemma}

It is tempting to prove the preceding lemma by setting $\delta = \eps$ and concluding that $\Psi(\vec{x},\eps,\eps) = f(\vec{x}) + O(\eps)$.
This is incorrect, as the $O(\delta)$ error term in $\Psi(\vec{x},\eps,\delta)$ may involve division by $\eps$, so setting $\delta = \eps$ may introduce erroneous terms to the output of $\Psi(\vec{x},\eps,\delta)$.
By setting $\delta = \eps^N$ for sufficiently large $N \in \naturals$, this problem is avoided.

Let $f(\vec{x}), g(\vec{x}) \in \F[\vec{x}]$ be polynomials such that $f(\vec{x}) + O(\eps)$ can be computed by a circuit with $g$-oracle gates.
Suppose we want to replace the $g$-oracle gates with oracles that approximately compute $g(\vec{x})$, i.e., oracle gates that compute some $h(\vec{x},\delta) = g(\vec{x}) + O(\delta)$.
As a consequence of the preceding lemma, we can obtain a circuit that computes $f(\vec{x}) + O(\eps)$ by using $h(\vec{x},\eps^N)$-oracles for some sufficiently large $N$.

\begin{lemma} \label{lem:exact to approx oracle}
	Let $f(\vec{x}), g(\vec{x}) \in \F[\vec{x}]$ be polynomials.
	Suppose $f(\vec{x}) + O(\eps)$ can be computed by a circuit of size $s$ with $g$-oracle gates.
	Let $h(\vec{x},\delta) \in \F\llb \delta \rrb [\vec{x}]$ be a polynomial such that $h(\vec{x},\delta) = g(\vec{x}) + O(\delta)$.
	Then there is some $N \in \naturals$ such that $f(\vec{x}) + O(\eps)$ can be computed by a circuit of size $s$ with $h(\vec{x},\eps^N)$-oracle gates.
\end{lemma}

\begin{proof}
	Let $\Phi(\vec{x},\eps)$ be a $g$-oracle circuit that computes $f(\vec{x}) + O(\eps)$ over $\F((\eps))[\vec{x}]$.
	Let $\Psi(\vec{x},\eps,\delta)$ be the circuit over $\F((\delta))((\eps))[\vec{x}]$ obtained by replacing each $g$-oracle gate with an $h(\vec{x},\delta)$ oracle.
	Since $h(\vec{x},\delta) = g(\vec{x}) + O(\delta)$, we have
	\[
		\Psi(\vec{x},\eps,\delta) = \Phi(\vec{x},\eps) + O(\delta) \in \F\llb \delta \rrb ((\eps)) [\vec{x}].
	\]
	Applying \autoref{lem:semicontinuity} yields an $N \in \naturals$ such that $\Psi(\vec{x},\eps,\eps^N) = f(\vec{x}) + O(\eps)$ as desired.
\end{proof}

\subsection{Polynomial Identity Testing}

When designing deterministic algorithms for polynomial identity testing (PIT), our focus will be on the black-box regime, where we are given access to a circuit $\Phi$ through an evaluation oracle.
Derandomizing PIT in this setting is equivalent to giving an explicit construction of a hitting set, defined below, for the set of polynomials computed by small circuits.

\begin{definition}
	Let $\mathcal{C} \subseteq \F[\vec{x}]$ be a set of polynomials.
	A set $\mathcal{H} \subseteq \F^n$ is a \emph{hitting set for $\mathcal{C}$} if for every nonzero $f \in \mathcal{C}$, there is some $\vec{\alpha} \in \mathcal{H}$ such that $f(\vec{\alpha}) \neq 0$.
\end{definition}

Alternatively, one can try to find an explicit, low-degree map $\mathcal{G} : \F^\ell \to \F^n$ with $\ell \ll n$ such that $f(\mathcal{G}(\vec{y})) \neq 0$ if $f$ is a nonzero polynomial computed by a small circuit.

\begin{definition}
	Let $\mathcal{C} \subseteq \F[\vec{x}]$ be a set of polynomials.
	A polynomial map $\mathcal{G} : \F^\ell \to \F^n$ is a \emph{hitting set generator for $\mathcal{C}$} if for every nonzero $f \in \mathcal{C}$, we have $f(\mathcal{G}(\vec{y})) \neq 0$.
	We call $\ell$ the \emph{seed length} of the generator.
	The \emph{degree} of the generator, denoted by $\deg(\mathcal{G})$, is given by $\max_{i \in [n]} \deg(\mathcal{G}_i)$.
\end{definition}

Small hitting sets (and hitting set generators with small seed length and low degree) are known to exist non-constructively.
In derandomizing PIT, one seeks efficient uniform constructions of these objects.
One can show that the notions of hitting sets and generators are essentially equivalent using polynomial interpolation (see, e.g., \textcite[Section 4]{SV15}).
In this work, we will prefer the language of generators, as they are more amenable to composition than are hitting sets.

It is natural to extend the definition of a hitting set to the setting of border complexity.
Over fields of characteristic zero, \textcite{FS18} defined a notion of a \emph{robust hitting set} for a class $\mathcal{C}$.
Using continuity, one can easily show that if $\mathcal{H}$ is a robust hitting set for a class $\mathcal{C}$, then $\mathcal{H}$ is also a hitting set for the closure $\overline{\mathcal{C}}$.
In this work, we will be concerned with hitting sets for the closures of circuit classes, but we will not pay particular attention to the robustness parameter, as some of our constructions take place in characteristic $p > 0$.

We note that a generator cannot simultaneously have very small seed length and very low degree.
In particular, a generator of degree $\Theta(1)$ must have seed length $n^{\Theta(1)}$.

\begin{lemma} \label{lem:seed length lb}
	Let $\mathcal{C} \subseteq \F[\vec{x}]$ be a set of polynomials such that $\mathcal{C}$ contains all linear polynomials.
	Suppose $\mathcal{G} : \F^\ell \to \F^n$ is a hitting set generator for $\mathcal{C}$ of degree $d$.
	Then we must have $\binom{\ell + d}{d} \ge n$.
	In particular, if $d$ is a fixed constant independent of $n$, then $\ell \ge \Omega(n^{1/d})$.
\end{lemma}

\begin{proof}
	For $i \in [n]$, let $\mathcal{G}_i(\vec{y})$ be the $i$\ts{th} coordinate of $\mathcal{G}$.
	Observe that each $\mathcal{G}_i(\vec{y})$ is a polynomial in $\ell$ variables of degree at most $d$.
	The space of $\ell$-variate polynomials of degree at most $d$ is a vector space of dimension $\binom{\ell+d}{d}$.
	Suppose for the sake contradiction that $\binom{\ell + d}{d} < n$.
	Then there is a non-trivial linear relation among the $n$ coordinates of $\mathcal{G}$.
	That is, there is a linear polynomial $L(x_1,\ldots,x_n) \neq 0$ such that
	\[
		L(\mathcal{G}_1(\vec{y}),\ldots,\mathcal{G}_n(\vec{y})) = 0.
	\]
	Since $L(\vec{x})$ is linear, we have $L \in \mathcal{C}$.
	This contradicts the assumption that $\mathcal{G}$ is a hitting set generator for $\mathcal{C}$.
\end{proof}

\subsection{Matrix Rank}

We will frequently make use of the fact that the rank of a matrix can be characterized by the (non-)vanishing of its minors.
This is a straightforward consequence of the fact that the row rank and column rank of a matrix coincide.

\begin{lemma}
	Let $A \in \F^{n \times m}$.
	Then $\rank(A) \ge r$ if and only if some $r \times r$ minor of $A$ does not vanish.
	Equivalently, $\rank(A) < r$ if and only if every $r \times r$ minor of $A$ vanishes.
\end{lemma}

We now define the hitting set generator which will be the focus of our work on PIT.

\begin{construction} \label{cons:matrix generator}
	Let $n, m, r \in \naturals$ with $r \le \min(n,m)$.
	Define the map $\mathcal{G}_{n,m,r} : \F^{n \times r} \times \F^{r \times m} \to \F^{n \times m}$ via
	\[
		\mathcal{G}_{n,m,r}(Y,Z)_{i,j} = (YZ)_{i,j}.
	\]
\end{construction}

The following are immediate consequences of the definition of $\mathcal{G}_{n,m,r}(Y,Z)$.

\begin{lemma} \label{lem:matrix generator}
	Let $\mathcal{G}_{n,m,r} : \F^{n \times r} \times \F^{r \times m} \to \F^{n \times m}$ be defined as in \autoref{cons:matrix generator}.
	\begin{enumerate}
		\item
			The image of $\mathcal{G}_{n,m,r}$ contains all $n \times m$ matrices of rank at most $r$.
		\item
			Each coordinate of $\mathcal{G}_{n,m,r}(Y,Z)$ is a $2r$-sparse degree-2 polynomial in the variables $Y \cup Z$.
		\item
			The map $\mathcal{G}_{n,m,r}(Y,Z)$ can be computed by a multi-output algebraic circuit of size $2nmr$ and product-depth 1.
			Additionally, each coordinate of the output can be computed by a homogeneous formula of size $2r$.
	\end{enumerate}
\end{lemma}

In order to prove that $\mathcal{G}_{n,m,r}$ is a hitting set generator for a class of circuits $\mathcal{C}$, it will be useful to understand which polynomials vanish when composed with $\mathcal{G}_{n,m,r}$.
If $f(X) \in \F[X]$ is a nonzero polynomial such that $f(\mathcal{G}_{n,m,r}(Y,Z)) = 0$, then $f$ necessarily vanishes on all $n \times m$ matrices of rank at most $r$.
The ideal of polynomials which vanish on matrices of rank at most $r$ is well-understood from the viewpoint of mathematics.

Let $\detideal{n}{m}{r}$ be the ideal generated by the $r \times r$ minors of a generic $n \times m$ matrix and let $J_{n,m,r}$ be the ideal of polynomials which vanish on all $n \times m$ matrices of rank at most $r$.
It is clear that $\detideal{n}{m}{r+1} \subseteq J_{n,m,r}$.
When the field $\F$ is algebraically closed, we in fact have the equality $\detideal{n}{m}{r+1} = J_{n,m,r}$.
This follows from Hilbert's Nullstellensatz and the fact that $\detideal{n}{m}{r}$ is radical (see, for example, \cite[Theorem 2.10 and Remark 2.12]{BV88}).
This implies that if $f(X)$ is nonzero and $f(\mathcal{G}_{n,m,r}(Y,Z)) = 0$, then $f \in J_{n,m,r} = \detideal{n}{m}{r+1}$.

In the case where $\F$ is not algebraically closed, we can still conclude that $f \in \detideal{n}{m}{r+1}$ if $f(\mathcal{G}_{n,m,r}(Y,Z)) = 0$.
This follows from the fact that if $f(\mathcal{G}_{n,m,r}(Y,Z)) = 0$, then $f$ vanishes on matrices of rank at most $r$ with entries in any extension $\mathbb{K} \supseteq \F$.
In particular, $f$ vanishes on matrices of rank at most $r$ with entries in $\overline{\F}$, the algebraic closure of $\F$.

We record the preceding observations as a lemma.

\begin{lemma} \label{lem:vanish on matrix generator} 
	Let $\F$ be any field and let $n,m,r \in \naturals$ with $r \le \min(n,m)$.
	Let $\detideal{n}{m}{r}$ denote the ideal of $\F[X]$ generated by the $r \times r$ minors of a generic $n \times m$ matrix and let $f(X) \in \F[X]$.
	Then $f(\mathcal{G}_{n,m,r-1}(Y,Z)) = 0$ if and only if $f(X) \in \detideal{n}{m}{r}$.
\end{lemma}

The characterization of matrix rank by (non-)vanishing of minors is useful mathematically, but does not immediately give rise to an efficient algorithm to compute matrix rank, as this requires checking $\binom{n}{r}^2$ minors of size $r$.
For our applications to the Ideal Proof System, it will be useful to have a small collection of polynomial equations that characterize matrix rank.
One can efficiently compute matrix rank via Gaussian elimination, but doing so requires branching steps that depend on the entries of the matrix.
In particular, Gaussian elimination does not provide a small set of equations characterizing matrix rank.

To obtain such equations for matrix rank, we will make use of rank condensers, which can be thought of as a matrix-oblivious form of Gaussian elimination.
Rank condensers originate in the work of \textcite{GR08}, who used them to design extractors for affine sources.
Since then, rank condensers have found applications to polynomial identity testing \cite{KS11,FS12,FSS14}, derandomization \cite{LMPS18}, and algorithms for linear algebra \cite{CKL13}.
Rank condensers also feature in the theory of linear-algebraic pseudorandomness developed by \textcite{FG15}.

\begin{definition} \label{def:rank condenser}
	Let $\F$ be a field and let $n \ge r \ge 1$.
	A collection of matrices $\mathcal{E} \subseteq \F^{t \times n}$ is a \emph{weak $(r,L)$-lossless rank condenser} if for all matrices $A \in \F^{n \times r}$ with $\rank(A) = r$, we have
	\[
		\Abs{\Set{E : E \in \mathcal{E}, \rank(EA) < \rank(A)}} \le L. \qedhere
	\]
\end{definition}

The following construction of weak lossless rank condensers was given by \textcite{FS12} (with an improved analysis due to \textcite{FG15}).

\begin{lemma}[\cite{FS12,FG15}] \label{lem:rank condenser construction}
	Let $\F$ be a field and let $\omega \in \F$ be an element of multiplicative order at least $n$.
	Define the matrix $W_\omega(x) \in \F[x]^{r \times n}$ by $(W_{\omega}(x))_{i,j} = (\omega^i x)^j$.
	Let $S \subseteq \F \setminus \set{0}$.
	Then the collection of matrices
	\[
		\mathcal{E} = \set{W_{\omega}(\alpha) : \alpha \in S} \subseteq \F^{r \times n}
	\]
	is a weak $(r, r(n-r))$-lossless rank condenser.
\end{lemma}

\subsection{Hasse Derivatives}

In this work, we use Hasse derivatives in place of the standard partial derivative.
Originally defined by \textcite{Hasse36}, Hasse derivatives are a notion of derivative that is more well-behaved over fields of small positive characteristic.
For a more thorough treatment of Hasse derivatives and their properties, see, for example, the thesis of \textcite[Appendix C]{Forbes14}.

\begin{definition}
	Let $\F$ be a field and let $f(\vec{x}) \in \F[\vec{x}]$.
	For $\vec{a} \in \naturals^n$, we define the \emph{$\vec{a}$\ts{th} Hasse derivative} of $f(\vec{x})$ to be
	\[
		\frac{\partial}{\partial \vec{x}^{\vec{a}}}(f) \coloneqq \Coeff_{\vec{y}^{\vec{a}}}(f(\vec{x} + \vec{y})),
	\]
	where $f(\vec{x} + \vec{y})$ is viewed as a polynomial in $\F[\vec{x}][\vec{y}]$.
\end{definition}

Equivalently, one can define Hasse derivatives in terms of their action on monomials.

\begin{lemma}
	Let $\vec{a}, \vec{b} \in \naturals^n$.
	Then 
	\[
		\frac{\partial}{\partial \vec{x}^{\vec{a}}}(\vec{x}^{\vec{b}}) = \prod_{i=1}^n \binom{b_i}{a_i} x_i^{b_i - a_i},
	\]
	where we use the convention that $\binom{b}{a} = 0$ if $b < a$.
\end{lemma}

A straightforward consequence of the preceding lemma is that Hasse derivatives interact nicely with degree.

\begin{lemma}
	Let $f(\vec{x}) \in \F[\vec{x}]$ and let $\vec{a} \in \naturals^n$.
	Then
	\[
		\deg\del{\frac{\partial}{\partial \vec{x}^{\vec{a}}}(f)} \le \deg(f) - \norm{\vec{a}}_1,
	\]
	with equality if $\frac{\partial}{\partial \vec{x}^{\vec{a}}}(f) \neq 0$.
\end{lemma}

Hasse derivatives also respect the multigrading on $\F[X]$.

\begin{lemma} \label{lem:multideg partial derivative}
	Let $f \in \F[X]$ be a multihomogeneous polynomial and let $A \in \naturals^{n \times m}$.
	Write $X^A \coloneqq \prod_{i=1}^n \prod_{j=1}^m x_{i,j}^{a_{i,j}}$ for the monomial with powers given by the matrix $A$.
	If $\frac{\partial f}{\partial X^A} \neq 0$, then 
	\[
		\multideg\del{\frac{\partial f}{\partial X^A}} = \multideg(f) - \multideg(X^A).
	\]
\end{lemma}

Just like standard partial derivatives, Hasse derivatives commute with one another.

\begin{lemma}[{see, e.g., \cite[Lemma C.1.4(5)]{Forbes14}}] \label{lem:pd commute}
	Let $f \in \F[\vec{x}]$ and let $\vec{a}, \vec{b} \in \naturals^n$.
	Then 
	\[
		\frac{\partial}{\partial \vec{x}^{\vec{a}}} \del{\frac{\partial}{\partial \vec{x}^{\vec{b}}}(f)} = \frac{\partial}{\partial \vec{x}^{\vec{b}}} \del{\frac{\partial}{\partial \vec{x}^{\vec{a}}}(f)}.
	\]
\end{lemma}

Hasse derivatives obey a modified form of the product rule.

\begin{lemma}[{see, e.g., \cite[Lemma C.1.7]{Forbes14}}] \label{lem:product rule}
	Let $f_1,\ldots,f_m \in \F[\vec{x}]$.
	For any $i \in [n]$ and $a \in \naturals$, we have
	\[
		\frac{\partial}{\partial x_i^a}(f_1 \cdots f_m) = \sum_{a_1 + \cdots + a_m = a}\frac{\partial}{\partial x_i^{a_1}}(f_1) \cdots \frac{\partial}{\partial x_i^{a_m}}(f_m).
	\]
\end{lemma}

We now define the space of ($d$\ts{th} order) partial derivatives of a polynomial.
The dimension of this space (and related spaces, like the space of shifted partial derivatives \cite{Kayal12}) is a useful complexity measure within algebraic circuit complexity.

\begin{definition}
	Let $f(\vec{x}) \in \F[\vec{x}]$.
	The \emph{space of partial derivatives of $f$}, denoted $\partial_{<\infty}(f)$, is defined as
	\[
		\partial_{< \infty}(f) \coloneqq \opspan_{\F}\Set{\frac{\partial f}{\partial \vec{x}^{\vec{a}}} : \vec{a} \in \naturals^n}.
	\]
	The \emph{space of $d$\ts{th}-order partial derivatives of $f$}, written $\partial_d(f)$, is given by
	\[
		\partial_d(f) \coloneqq \opspan_{\F} \Set{\frac{\partial f}{\partial \vec{x}^{\vec{a}}} : \vec{a} \in \naturals^n, \norm{\vec{a}}_1 = d}.
	\]
	We also write
	\[
		\partial_{\le d}(f) \coloneqq \opspan_{\F} \bigcup_{i=0}^{d} \partial_{i}(f)
	\]
	for the space of partial derivatives of order at most $d$.
\end{definition}

We will need the following lemma relating the dimension of the space of partial derivatives of a polynomial $f(\vec{x})$ and a linear projection $f(A\vec{x})$.

\begin{lemma} \label{lem:pd linear transformation}
	Let $f(\vec{x}) \in \F[\vec{x}]$ and let $A \in \F^{n \times n}$.
	Then for every $d \in \naturals$, we have $\dim(\partial_{\le d}(f(A \vec{x}))) \le \dim(\partial_{\le d}(f(\vec{x})))$.
	In particular, if $A$ is invertible, then $\dim(\partial_{\le d}(f(A \vec{x}))) = \dim(\partial_{\le d}(f(\vec{x})))$.
\end{lemma}

\begin{proof}
	Using the chain rule for Hasse derivatives, one can show (see, e.g., \cite[Corollary C.2.7]{Forbes14}) that for all $\vec{e} \in \naturals^n$ with $\norm{\vec{e}}_1 \le d$, we have
	\[
		\frac{\partial}{\partial \vec{x}^{\vec{e}}}(f(A \vec{x})) \in \opspan_{\F} \Set{g(A\vec{x}) : g(\vec{x}) \in \partial_{\le d}(f(\vec{x}))}.
	\]
	Let $V \coloneqq \opspan_{\F} \Set{g(A\vec{x}) : g(\vec{x}) \in \partial_{\le d}(f(\vec{x}))}$.
	This implies
	\[
		\partial_{\le d}(f(A \vec{x})) \subseteq V,
	\]
	so
	\[
		\dim\partial_{\le d}(f(A \vec{x})) \le \dim V.
	\]
	We now show that $\dim V$ bounded by $\dim \partial_{\le d}(f(\vec{x}))$.
	Let $g_1(\vec{x}),\ldots,g_k(\vec{x}) \in \partial_{\le d}(f(\vec{x}))$ and suppose that $g_1(A\vec{x}),\ldots,g_k(A\vec{x})$ are linearly independent.
	This implies that $g_1(\vec{x}), \ldots, g_k(\vec{x})$ are linearly independent, as any linear relation satisfied by $g_1(\vec{x}),\ldots,g_k(\vec{x})$ will also be satisfied by $g_1(A\vec{x}),\ldots,g_k(A\vec{x})$.
	If we select the $g_i$ such that $\set{g_1(A\vec{x}),\ldots,g_k(A\vec{x})}$ forms a basis of $V$, then we have
	\[
		\dim V = k \le \dim \partial_{\le d}(f(\vec{x})).
	\]
	Combining this with the previous inequality completes the proof.

	In the case where $A$ is invertible, we use the fact that $\vec{x} = A^{-1} A \vec{x}$ to obtain
	\[
		\dim \partial_{\le d}(f(\vec{x})) \le \dim \partial_{\le d}(f(A \vec{x})) \le \dim \partial_{\le d}(f(\vec{x})),
	\]
	so equality holds.
\end{proof}

We note that by taking $d \ge \deg(f)$ in \autoref{lem:pd linear transformation}, one can replace $\partial_{\le d}(\bullet)$ with $\partial_{< \infty}(\bullet)$.

\subsection{Bideterminants and the Straightening Law} \label{subsec:bidet}

The proof of \autoref{thm:proj to small abp} relies on understanding how a polynomial $f \in \detideal{n}{m}{r}$ behaves under the map $X \mapsto A X B$ for invertible matrices $A$ and $B$.
For example, it is easy to see that $f(AXB)$ also lies in $\detideal{n}{m}{r}$.
However, it is not clear if there is other structure we may take advantage of.
By working in a different basis of $\F[X]$, we can better understand how $f(AXB)$ relates to $f(X)$.
Before describing this basis, we recall the notions of a Young diagram and Young tableau.

\begin{definition}
	A \emph{partition} $\sigma = (\sigma_1,\sigma_2,\ldots,\sigma_k)$ is a non-increasing sequence of natural numbers.
	If $\sum_{i=1}^k \sigma_i = n$, we write $\sigma \vdash n$.
	The \emph{transpose} of $\sigma$, denoted $\hat{\sigma}$, is the partition given by $\hat{\sigma}_i = \Abs{\set{j : \sigma_j \ge i}}$.
	Associated with a partition $\sigma$ is its \emph{Young diagram} $D_\sigma \subseteq \naturals \times \naturals$, given by $D_\sigma = \set{(i,j) : j \le \sigma_i}$.
\end{definition}

Note that $\hat{\sigma}_1$ counts the number of rows in the Young diagram of $\sigma$.
We graphically depict the Young diagram of a partition as a collection of boxes.
For example, the Young diagram of the partition $(4,2,2,1)$ is 
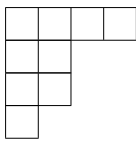
\[
	\yng(4,2,2,1).
\]
This partition has transpose $(4,3,1,1)$, with Young diagram given by
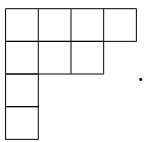
\[
	\yng(4,3,1,1).
\]
The lexicographic ordering on integer sequences induces an ordering on partitions, which we denote by $<_\lex$.

We now define Young tableaux, which can be obtained by writing a number in each cell of the Young diagram of some partition $\sigma$.

\begin{definition}
	Given a partition $\sigma$, a \emph{Young tableau $T$ of shape $\sigma$} is a map $T : D_\sigma \to \naturals$ assigning a natural number to each cell of the Young diagram of $\sigma$.
	We denote the $i$\ts{th} row of $T$ by $T(i,\bullet)$, which we will view as either a set or a one-row Young tableau depending on context.
	A Young tableau is \emph{standard} if its entries are strictly increasing along each column and along each row.
	A Young tableau is \emph{semistandard} if its entries are strictly increasing along each column and are nondecreasing along each row.
	If $T : D_\sigma \to \naturals$ is a Young tableau, its \emph{conjugate tableau} $\hat{T} : D_{\hat{\sigma}} \to \naturals$ is given by $\hat{T}(i,j) = T(j,i)$.
\end{definition}

Continuing the example above, one Young tableau (of many) of shape $(4,2,2,1)$ is given by
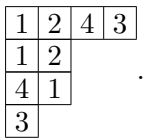
\[
	\young(1243,12,41,3).
\]

Next, we introduce bitableaux and bideterminants.
A bitableau is simply a pair of Young tableau of the same shape, while a bideterminant is a natural polynomial associated to this pair of tableaux.

\begin{definition}
	Let $X = (x_{1,1},\ldots,x_{n,n})$ be an $n \times n$ matrix of variables.
	A \emph{bitableau} $(S,T)$ is a pair of Young tableaux of the same shape $\sigma$.
	If the entries of $S$ and $T$ are from $[n]$, we associate to $(S,T)$ the \emph{bideterminant} $(S|T)(X)$, defined as
	\[
		(S|T)(X) \coloneqq \prod_{i=1}^{\hat{\sigma}_1} \det 
		\begin{pmatrix}
			x_{S(i,1),T(i,1)} & x_{S(i,1), T(i,2)} & \cdots & x_{S(i,1),T(i,\sigma_i)} \\
			x_{S(i,2),T(i,1)} & x_{S(i,2), T(i,2)} & \cdots & x_{S(i,2),T(i,\sigma_i)} \\
			\vdots & \vdots & \ddots & \vdots \\
			x_{S(i,\sigma_i),T(i,1)} & x_{S(i,\sigma_i), T(i,2)} & \cdots & x_{S(i,\sigma_i), T(i, \sigma_i)}
		\end{pmatrix}.
	\]
	The $i$\ts{th} term in this product is the determinant of the submatrix whose rows and columns are listed in the $i$\ts{th} row of the tableaux $S$ and $T$, respectively.
	The \emph{width} of the bideterminant $(S|T)$ is given by $\sigma_1$.
	We say that the bitableau $(S,T)$ and bideterminant $(S|T)$ are \emph{standard} if, as tableaux, both $S$ and $T$ are increasing along each row and nondecreasing along each column (equivalently, that $S$ and $T$ are both the transpose of a semistandard Young tableau).
\end{definition}

For example, associated to the bitableau
\[
	\del{\young(123,13,4),\ \young(134,24,3)}
\]
is the bideterminant
\[
	\det\begin{pmatrix}
		x_{1,1} & x_{1,3} & x_{1,4} \\
		x_{2,1} & x_{2,3} & x_{2,4} \\
		x_{3,1} & x_{3,3} & x_{3,4}
	\end{pmatrix}
	\det\begin{pmatrix}
		x_{1,2} & x_{1,4} \\
		x_{3,2} & x_{3,4} 
	\end{pmatrix}
	\det\begin{pmatrix}
		x_{4,3}
	\end{pmatrix}.
\]
Note that a bideterminant $(S|T)$ is multihomogeneous of degree $(s_1 \vec{e}_1 + \cdots + s_n \vec{e}_n) \oplus (t_1 \vec{e}_1 + \cdots + t_n \vec{e}_n)$, where $s_i$ and $t_i$ count the number of occurrences of $i$ in $S$ and $T$, respectively.

It is easy to see that the bideterminants span $\F[X]$, since a monomial $\prod_{i=1}^d x_{r_i,c_i}$ is the bideterminant corresponding to the bitableau
\[
	\newcommand{\ra}{r_1}
	\newcommand{\rb}{r_2}
	\newcommand{\rd}{r_d}
	\newcommand{\ca}{c_1}
	\newcommand{\cb}{c_2}
	\newcommand{\cd}{c_d}
	\del{\young(\ra,\rb,\cdots,\rd),\ \young(\ca,\cb,\cdots,\cd)}.
\]
Perhaps surprisingly, there is a natural subset of the bideterminants which form a basis of $\F[X]$.

\begin{theorem}[\cite{DRS74}]
	The standard bideterminants form a basis of $\F[X]$.
\end{theorem}

To show $\F[X]$ is spanned by standard bideterminants, it suffices to express non-standard bideterminants as linear combinations of standard bideterminants.
The fact that this can be done, along with some additional structural information, is known as the straightening law.
For more on the straightening law, including its history and its applications to invariant theory, see the introduction of \textcite{DKR78}.

\begin{theorem}[{\cite{DRS74}, see also \cite{DKR78,CEP80}}] \label{thm:straightening}
	Let $(S|T)(X)$ be a bideterminant of shape $\sigma$.
	Then $(S|T)(X)$ can be expressed as a linear combination
	\[
		(S|T)(X) = \sum_{(A,B)} c_{A,B} (A|B)(X),
	\]
	where the $c_{A,B}$ are integers and the sum ranges over all standard bitableaux $(A,B)$ of shape $\tau$ such that $\tau \ge_\lex \sigma$.
\end{theorem}

One immediate corollary of this is a characterization of polynomials in the ideal $\detideal{n}{m}{r}$ by their support in the standard bideterminant basis.

\begin{corollary} \label{cor:width of I_r}
	A polynomial $f \in \F[X]$ is an element of the ideal $\detideal{n}{m}{r}$ if and only if $f$ is supported on bideterminants of width at least $r$.
\end{corollary}

\subsection{Pfaffians}

This subsection departs slightly from the setting of the previous subsections.
Let $X$ be a $2n \times 2n$ skew-symmetric matrix of variables. 
That is, the $(i,j)$ entry of $X$ is the variable $x_{i,j}$ and the variables $x_{i,j}$ and $x_{j,i}$ satisfy the relation $x_{i,j} = - x_{j,i}$.
It is well-known that the determinant of $X$ is the square of a polynomial; this square root of the determinant is the \emph{Pfaffian} of $X$.
Formally, one can define the Pfaffian $\Pf(X)$ as
\[
	\Pf(X) = \frac{1}{2^n n!} \sum_{\sigma \in S_{2n}} \sgn(\sigma) \prod_{i=1}^n x_{\sigma(2i-1),\sigma(2i)},
\]
where $S_{2n}$ is the group of all permutations on $[2n] = \set{1,\ldots,2n}$.
Each monomial in the above sum appears $2^n n!$ times, so every monomial in the support of the Pfaffian has a coefficient of $1$ or $-1$.
In particular, the Pfaffian is well-defined even over fields of small characteristic.

As remarked above, we have $\Pf(X)^2 = \det(X)$ when $X$ is a skew-symmetric matrix.
If $X$ is an $m \times m$ skew-symmetric matrix for odd $m$, then $\det(X) = 0$, so we restrict our attention to matrices of even order.
The equation $\Pf(X)^2 = \det(X)$ relates the Pfaffian and determinant of a skew-symmetric matrix.
For general matrices, we can relate Pfaffians and determinants via the following lemma.

\begin{lemma} \label{lem:pfaff orthogonal transform}
	Let $A$ be a $2n \times 2n$ skew-symmetric matrix and let $B$ be an arbitrary $2n \times 2n$ matrix.
	Then $B A B^\top$ is skew-symmetric and $\Pf(B A B^\top) = \det(B) \Pf(A)$.
\end{lemma}

We will also make use of the symmetries of the Pfaffian as described in the next lemma.

\begin{lemma} \label{lem:pfaff to det}
	Let $A$ be an $n \times n$ matrix.
	Then
	\[
		\Pf\begin{pmatrix} 0 & A \\ -A^{\top} & 0 \end{pmatrix} = (-1)^{\binom{n}{2}} \det(A).
	\]
\end{lemma}

As with determinants, one can consider the ideal generated by sub-Pfaffians of the same size of a skew-symmetric matrix.
To ensure that the Pfaffian of a submatrix of $X$ is well-defined, we restrict our attention to \emph{principal submatrices}.
Recall that a submatrix $X_{R,C}$ of $X$ is principal if $R = C$.
If $X$ is skew-symmetric, then so is any principal submatrix of $X$.
Throughout this work, we will use $\pfaffideal{2n}{2r}$ to denote the ideal of $\F[X]$ generated by the Pfaffians of the $2r \times 2r$ principal submatrices of $X$.

Much like the case with determinants, one can understand the ideal $\pfaffideal{2n}{2r}$ using an analogous straightening law for Pfaffians.
To do this, we begin by defining the analogues of standard bideterminants for Pfaffian ideals.

\begin{definition}
	Let $T$ be a conjugate semistandard Young tableau of shape $\sigma$ such that every row of $T$ has even length.
	We associate to $T$ the \emph{standard monomial} $[T](X)$, which is a polynomial defined as the product of Pfaffians
	\[
		[T](X) \coloneqq \prod_{i=1}^{\hat{\sigma}_1} \Pf
		\begin{pmatrix}
			x_{T(i,1),T(i,1)} & x_{T(i,1),T(i,2)} & \cdots & x_{T(i,1),T(i,\sigma_i)} \\
			x_{T(i,2),T(i,1)} & x_{T(i,2),T(i,2)} & \cdots & x_{T(i,2),T(i,\sigma_i)} \\
			\vdots & \vdots & \ddots & \vdots \\
			x_{T(i,\sigma_i),T(i,1)} & x_{T(i,\sigma_i),T(i,2)} & \cdots & x_{T(i,\sigma_i),T(i,\sigma_i)}
		\end{pmatrix}.
	\]
	That is, the $i$\ts{th} polynomial in the above product is the Pfaffian of the submatrix of $X$ whose rows and columns are listed in the $i$\ts{th} row of the tableau $T$.
	The \emph{width} of $[T](X)$ is $\sigma_1$, the size of the largest Pfaffian in the above product.
\end{definition}

If we were to extend the above definition to all Young tableaux, it is clear that the resulting set of polynomials spans $\F[X]$, since
\[
	\sbr{\,\young(ij)\,}(X) = \Pf \begin{pmatrix} 0 & x_{i,j} \\ -x_{i,j} & 0 \end{pmatrix} = x_{i,j}.
\]
However, we do not lose much by ignoring these non-standard monomials.
In a manner analogous to the determinantal case, \textcite{CP76} proved that the standard monomials form a basis of $\F[X]$.

\begin{theorem}[{\cite[Theorem 6.5]{CP76}}] \label{thm:pfaff straightening}
	For any commutative ring $R$ with unity, the standard monomials form a basis of $R[X]$.
\end{theorem}

To prove this, \textcite{CP76} showed that the standard monomials span $R[X]$ and that any non-standard monomial can be written as a linear combination of standard monomials.
The expression of a non-standard monomial as a linear combination of standard monomials is, as in the determinantal case, known as the \emph{straightening law}.
Using the straightening law of \textcite[Lemmas 6.1 and 6.2]{CP76}, one can show (following \textcite[Section 8]{DRS74}) that a non-standard monomial of width $2r$ is supported only on standard monomials of width at least $2r$.
A straightforward corollary of this is that every polynomial in the ideal generated by the Pfaffians of the principal $2r \times 2r$ submatrices of a matrix $X$ is supported on standard monomials of width at least $2r$.

\begin{corollary} \label{cor:pfaff ideal width}
	Let $X$ be a generic $2n \times 2n$ skew-symmetric matrix.
	Let $\pfaffideal{2n}{2r}$ be the ideal generated by the Pfaffians of the $2r \times 2r$ principal submatrices of $X$.
	Then any $f \in \pfaffideal{2n}{2r}$ is supported on standard monomials of width at least $2r$.
\end{corollary}

\subsection{Monomial Orders}

Our use of border complexity stems from the need to construct circuits that compute only a particular subset of the monomials appearing in the support of a polynomial $f$.
To do this, we make use of monomial orders and leading monomials, which we now define.

\begin{definition}
	A \emph{monomial order} $\prec$ is a total order on the monomials of $\F[\vec{x}]$ which satisfies
	\begin{enumerate}
		\item
			$1 \prec \vec{x}^{\vec{a}}$ for all nonzero $\vec{a} \in \naturals^n$, and
		\item
			if $\vec{x}^{\vec{a}} \prec \vec{x}^{\vec{b}}$, then $\vec{x}^{\vec{a} + \vec{c}} \prec \vec{x}^{\vec{b} + \vec{c}}$ for all $\vec{a}, \vec{b}, \vec{c} \in \naturals^n$. \qedhere
	\end{enumerate}
\end{definition}

\begin{definition}
	Let $\prec$ be a monomial order and let $f(\vec{x}) \in \F[\vec{x}]$ be a nonzero polynomial.
	The \emph{leading monomial of $f$ with respect to $\prec$}, written $\LM_\prec(f)$, is the $\prec$-maximal monomial appearing in the support of $f$.
	The \emph{leading coefficient of $f$ with respect to $\prec$}, denoted $\LC_\prec(f)$, is the coefficient of $\LM_\prec(f)$ when $f$ is written as a sum of monomials.
\end{definition}

We may write $\LM(f)$ and $\LC(f)$ for the leading monomial and coefficient of $f$, respectively, if the order $\prec$ is clear from context.
A useful property of leading monomials is that taking the leading monomial commutes with products of polynomials.

\begin{lemma} \label{lem:LM commutes}
	Let $\prec$ be a monomial order and let $f, g \in \F[\vec{x}]$ be nonzero polynomials.
	Then $\LM_{\prec}(fg) = \LM_\prec(f) \cdot \LM_\prec(g)$.
\end{lemma}

We will primarily be interested in lexicographic orders, which are a special case of weight orders.
To specify a weight order, we are given some weight vector $\vec{u} \in \reals^n$, and we order two monomials $\vec{x}^{\vec{a}}$ and $\vec{x}^{\vec{b}}$ by comparing the inner products $\abr{\vec{u},\vec{a}}$ and $\abr{\vec{u},\vec{b}}$.
To obtain a total order on the set of monomials, ties must be broken.
This is done by choosing another weight vector $\vec{w} \in \reals^n$ and breaking ties by comparing $\abr{\vec{w},\vec{a}}$ and $\abr{\vec{w},\vec{b}}$.
If ties are still possible, we continue choosing new weight vectors until all ties are broken.
It turns out that \emph{every} monomial order can be obtained from such a collection of weight vectors.

\begin{theorem}[{\cite[Theorem 2.5]{Robbiano86}, see \cite{Robbiano85} for a proof}] \label{thm:weight order}
	Let $\prec$ be a monomial ordering on $\F[\vec{x}]$.
	Denote by $\abr{\bullet,\bullet}$ the standard inner product on $\reals^n$.
	There is an integer $s \in [n]$ and vectors $\vec{u}^{(1)}, \ldots, \vec{u}^{(s)} \in \reals^n$ such that $\vec{x}^{\vec{a}} \prec \vec{x}^{\vec{b}}$ if and only if there is some $j \in [s]$ such that 
	\begin{enumerate}
		\item
			$\abr{\vec{a},\vec{u}^{(i)}} = \abr{\vec{b},\vec{u}^{(i)}}$ for all $i < j$, and 
		\item
			$\abr{\vec{a},\vec{u}^{(j)}} < \abr{\vec{b},\vec{u}^{(j)}}$.
	\end{enumerate}
\end{theorem}

Our focus will be on monomial orders specified by integral weight vectors, which includes all lexicographic orders.

\begin{fact}
	Any lexicographic monomial ordering can be specified by a collection of integral weight vectors.
\end{fact}

Let $f(\vec{x}) \in R[\vec{x}]$ be a polynomial over a commutative ring $R$ and let $\prec$ be a monomial order that corresponds to a collection of integral weights.
It will be useful later on to find an assignment $x_i \mapsto \eps^{d_i}$ of the variables to powers of $\eps$ such that $f(\vec{x})$ evaluates to $\eps^m \LC_\prec(f) + O(\eps^{m+1})$ for some integer $m$.
As a first step, we record as a lemma an argument of \textcite[Example 2.2]{Burgisser04} on degenerating a polynomial to a face of its Newton polytope.

\begin{lemma}[{\cite[Example 2.2]{Burgisser04}}] \label{lem:single newton degeneration}
	Let $R$ be a commutative ring and let $f \in R[\vec{x}]$ be given by
	\[
		f(\vec{x}) = \sum_{\vec{a} \in \supp(f)} \alpha_{\vec{a}} \vec{x}^{\vec{a}}.
	\]
	Let $\vec{u} \in \integers^n$, let $\lambda = \max_{\vec{a} \in \supp(f)} \abr{\vec{a},\vec{u}}$, and let $H = \set{\vec{a} \in \supp(f) : \abr{\vec{a},\vec{u}} = \lambda}$.
	Then 
	\[
		\eps^\lambda f(\eps^{-u_1} x_1,\ldots,\eps^{-u_n}x_n) = \sum_{\vec{a} \in \supp(f) \cap H} \alpha_{\vec{a}}\vec{x}^{\vec{a}} + O(\eps).
	\]
\end{lemma}

One can iteratively apply this lemma, further restricting the monomials of $f$ to have exponents that lie in the intersection of multiple hyperplanes.

\begin{lemma} \label{lem:multiple newton degeneration}
	Let $R$ be a commutative ring and let $f \in R[\vec{x}]$ be given by
	\[
		f(\vec{x}) = \sum_{\vec{a} \in \supp(f)} \alpha_{\vec{a}} \vec{x}^{\vec{a}}.
	\]
	Let $\vec{u}^{(1)},\ldots,\vec{u}^{(k)} \in \reals^n$ be vectors.
	For each $i \in [k]$, let
	\begin{align*}
		\lambda_i &\coloneqq \max_{\vec{a} \in \supp(f) \cap H_1 \cap \cdots \cap H_{i-1}} \abr{\vec{a},\vec{u}^{(i)}} \\
		H_i &\coloneqq \Set{\vec{a} \in \supp(f) \cap H_1 \cap \cdots \cap H_{i-1} : \abr{\vec{a},\vec{u}^{(i)}} = \lambda}.
	\end{align*}
	Then there are integers $d_1,\ldots,d_n$ and $m$ such that
	\[
		\eps^m f(\eps^{d_1} x_1, \ldots, \eps^{d_n} x_n) = \sum_{\vec{a} \in \supp(f) \cap H_1 \cap \cdots \cap H_k} \alpha_{\vec{a}} \vec{x}^{\vec{a}} + O(\eps).
	\]
\end{lemma}

\begin{proof}
	We proceed by induction on $k$, noting that the case of $k=1$ exactly corresponds to \autoref{lem:single newton degeneration}.
	When $k \ge 2$, by induction we have integers $d_1',\ldots,d_n'$ and $m'$ such that
	\[
		\eps^{m'} f(\eps^{d'_1} x_1,\ldots, \eps^{d_n'} x_n) = \sum_{\vec{a} \in \supp(f) \cap H_1 \cap \cdots \cap H_{k-1}} \alpha_{\vec{a}} \vec{x}^{\vec{a}} + \eps \cdot g(\vec{x},\eps),
	\]
	where $g(\vec{x},\eps) \in \F[\eps][\vec{x}]$.
	By \autoref{lem:single newton degeneration}, we have
	\begin{multline*}
		\delta^{\lambda_k} \eps^{m'} f(\eps^{d'_1} \delta^{-\vec{u}^{(k)}_1} x_1,\ldots, \eps^{d_n'} \delta^{-\vec{u}^{(k)}_n} x_n) \\ = \sum_{\vec{a} \in \supp(f) \cap H_1 \cap \cdots \cap H_{k}} \alpha_{\vec{a}} \vec{x}^{\vec{a}} + \delta^{\lambda_k} \eps \cdot g(\delta^{-\vec{u}_1^{(k)}} x_1, \ldots, \delta^{-\vec{u}_n^{(k)}} x_n, \eps) + O(\delta).
	\end{multline*}
	The expression $\delta^{\lambda_k} \eps \cdot g(\delta^{-\vec{u}_1^{(k)}} x_1, \ldots, \delta^{-\vec{u}_n^{(k)}} x_n, \eps)$ lies in the ring $\eps\F[\delta,\delta^{-1},\eps][\vec{x}]$ and may have terms whose coefficient involves a negative power of $\delta$.
	Let $M$ be the largest power of $\delta$ appearing in the denominator of the coefficient of a monomial in $\delta^{\lambda_k} \eps \cdot g(\delta^{-\vec{u}_1^{(k)}} x_1, \ldots, \delta^{-\vec{u}_n^{(k)}} x_n, \eps)$.
	Then under the substitution
	\begin{align*}
		\eps &\mapsto \eps^{M+1} \\
		\delta &\mapsto \eps,
	\end{align*}
	every monomial of $\eps^{\lambda_k + M + 1} g(\eps^{-\vec{u}_1^{(k)}} x_1, \ldots, \eps^{-\vec{u}_n^{(k)}} x_n, \eps)$ has a coefficient in $\eps \F[\eps]$.
	In particular, we have
	\[
		\eps^{\lambda_k + (M+1)m'} f(\eps^{d'_1(M+1) - \vec{u}^{(k)}_1} x_1 , \ldots, \eps^{d'_n(M+1) - \vec{u}^{(k)}_n} x_n) = \sum_{\vec{a} \in \supp(f) \cap H_1 \cap \cdots \cap H_k} \alpha_{\vec{a}} \vec{x}^{\vec{a}} + O(\eps).
	\]
	This completes the proof of the inductive step.
\end{proof}

By applying \autoref{lem:multiple newton degeneration} to a polynomial and subsequently setting $x_i \mapsto 1$ for all $i \in [n]$, we can approximate the leading coefficient of $f$ in the sense of border complexity.
If the ring $R$ is a field, then this is not necessarily useful.
However, we will apply this result when the ring $R$ is a polynomial ring in another set of variables, which makes this lemma useful.

\begin{lemma} \label{lem:approximate lc}
	Let $R$ be a commutative ring.
	Let $f(\vec{x}) \in R[\vec{x}]$ and let $\prec$ be a monomial order on $\vec{x}$.
	Suppose that the ordering $\prec$ can be specified by a collection of integral weight vectors $\vec{u}^{(1)},\ldots,\vec{u}^{(s)} \in \naturals^n$.
	Then there is some $m \in \integers$ and a collection of nonzero integers $\set{d_1,\ldots,d_n}$ such that the mapping
	\[
		x_i \mapsto \eps^{d_i}
	\]
	sends $f(\vec{x})$ to
	\[
		\eps^m \cdot \LC(f) + O(\eps^{m+1}).
	\]
\end{lemma}

\begin{proof}
	As in the statement of \autoref{lem:multiple newton degeneration}, for $i \in [k]$ let
	\begin{align*}
		\lambda_i &\coloneqq \max_{\vec{a} \in \supp(f) \cap H_1 \cap \cdots \cap H_{i-1}} \abr{\vec{a},\vec{u}^{(i)}} \\
		H_i &\coloneqq \Set{\vec{a} \in \supp(f) \cap H_1 \cap \cdots \cap H_{i-1} : \abr{\vec{a},\vec{u}^{(i)}} = \lambda}.
	\end{align*}
	Let $\vec{x}^{\vec{e}} = \LM(f)$.
	Since $\vec{u}^{(1)},\ldots,\vec{u}^{(k)}$ are weight vectors specifying a monomial order, it follows from the definition of such an order that $H_k = \set{\vec{e}}$.
	Applying \autoref{lem:multiple newton degeneration} yields integers $d_1,\ldots,d_n$ and $m$ such that
	\[
		f(\eps^{d_1} x_1, \ldots, \eps^{d_n} x_n) = \eps^m \LC(f) \LM(f) + O(\eps^{m+1}).
	\]
	Setting $x_i \mapsto 1$ for all $i \in [n]$ yields
	\[
		f(\eps^{d_1},\ldots,\eps^{d_n}) = \eps^m \LC(f) + O(\eps^{m+1})
	\]
	as claimed.
\end{proof}

\subsection{The Ideal Proof System}

The ideal proof system of \textcite{GP18} is an algebraic proof system used to refute unsatisfiable systems of polynomial equations.
The complexity of a proof in this system is measured by the size of the smallest algebraic circuit representing that proof.

\begin{definition}[\cite{GP18}]
	Let $\F$ be a field and let $f_1(\vec{x}),\ldots,f_m(\vec{x}) \in \F[\vec{x}]$.
	An \emph{ideal proof system (IPS) certificate} that the system $f_1(\vec{x}) = \cdots = f_m(\vec{x}) = 0$ is unsatisfiable over the algebraic closure $\overline{\F}$ is a polynomial $C(\vec{x},\vec{y}) \in \F[\vec{x},\vec{y}]$ such that
	\begin{enumerate}
		\item
			$C(\vec{x},\vec{0}) = 0$, and
		\item
			$C(\vec{x},f_1(\vec{x}),\ldots,f_m(\vec{x})) = 1$. \qedhere
	\end{enumerate}
\end{definition}

The first condition equates to requiring that $C(\vec{x},\vec{y})$ is in the ideal generated by $y_1,\ldots,y_m$.
This, along with the second condition, implies that $C(\vec{x},\vec{y})$ is a certificate for the fact $1 \in \abr{f_1(\vec{x}),\ldots,f_m(\vec{x})}$, hence that $f_1 = \cdots = f_m = 0$ is unsatisfiable.

For a class of algebraic circuits $\mathcal{C}$, one can also consider the $\mathcal{C}$-IPS proof system wherein we require the IPS certificate be computed by a circuit from $\mathcal{C}$.
We will primarily be concerned with IPS certificates computable by formulas or low-depth circuits.

As mentioned in the introduction, proving lower bounds on the complexity of IPS refutations is \textit{a priori} more difficult than proving lower bounds for algebraic circuits.
This is due to the fact that there may be infinitely many IPS certificates for a single system of equations, so we are faced with proving lower bounds for an infinite family of polynomials.
However, these certificates all lie in a coset of an ideal, so one could hope to understand this ideal well enough to prove lower bounds for the relevant coset.
See \textcite[Section 6]{GP18} for more on the difference between lower bounds for algebraic circuits and IPS.

The following lemma establishes a connection between lower bounds for multiples and lower bounds for IPS.
\textcite{FSTW16} originally stated and proved this lemma with $\set{x_i^2 - x_i : i \in [n]}$ as an additional set of axioms, but these are not necessary.
We will make use of this lemma when proving lower bounds for IPS.

\begin{lemma}[{\cite[Lemma 7.1]{FSTW16}}] \label{lem:fstw}
	Let $f(\vec{x}), g_1(\vec{x}), \ldots, g_k(\vec{x}) \in \F[\vec{x}]$ be an unsatisfiable system of equations where $g_1(\vec{x}),\ldots,g_k(\vec{x})$ is satisfiable.
	Let $C \in \F[\vec{x},y,\vec{z}]$ be an IPS refutation of $f,g_1,\ldots,g_k$.
	Then $1 - C(\vec{x},0,g_1(\vec{x}),\ldots,g_k(\vec{x}))$ is a nonzero multiple of $f(\vec{x})$.
\end{lemma}

\section{Hardness of Determinantal Ideals} \label{sec:hardness}

Recall that $X$ denotes an $n \times m$ matrix of variables and $\detideal{n}{m}{r} \subseteq \F[X]$ is the ideal generated by the $r \times r$ minors of $X$.
In this section, we study the minimum possible border complexity of a nonzero polynomial in $\detideal{n}{m}{r}$.
Our main result is that, up to polynomial factors, there is no polynomial $f \in \detideal{n}{m}{r}$ that is easier to compute than the $r \times r$ determinant.
We do this by constructing, for every nonzero $f \in \detideal{n}{m}{r}$, a depth-three $f$-oracle circuit that border computes the $\Theta(r^{1/3}) \times \Theta(r^{1/3})$ determinant.

The argument proceeds in two steps.
First, we show that for every $f(X) \in \detideal{n}{m}{r}$, there is a linear change of variables that takes $f(X)$ to $(S|T)(X) + O(\eps)$ for some bideterminant $(S|T)$ of width at least $r$.
The analysis of this step crucially relies on the straightening law (\autoref{thm:straightening}).
Second, for any $g(\vec{y})$ computed by an ABP of size at most $r$ and any bideterminant $(S|T)(X)$ of width $r$, we construct a depth-three $(S|T)$-oracle circuit computing $g(\vec{y}) + O(\eps)$.
As the determinant can be efficiently computed by ABPs, composing these steps yields an $f$-oracle circuit for $\det_{\Theta(r^{1/3})}(X) + O(\eps)$.

\subsection{Computing a Single Bideterminant}

For $i, j \in [n]$ with $i \neq j$, we define the \emph{substitution operator} $\sub{i}{j}$ acting on a transpose semistandard Young tableau $T$ as follows: for every row in $T$ containing $i$ but not $j$, substitute $i$ with $j$ and re-order the row to be in increasing order.
Let $h_i^j(T)$ denote the number of rows of $T$ changed by applying $\sub{i}{j}$ to $T$.
In general, the map $T \mapsto (\sub{i}{j}(T), h_i^j(T))$ may not be injective.
However, the following lemma shows that mapping is injective when restricted to tableaux satisfying a particular property.

\begin{lemma}[{\cite[Proposition 1.6]{CEP80}}] \label{lem:S operator}
	Let $i, j \in [n]$.
	Suppose $T$ is a conjugate semistandard tableau with entries in $[n]$ with the property that if a row of $T$ contains an integer $k \le i$, then that row contains all integers in $\set{i,i+1,\ldots,j-1}$.
	Then $\sub{i}{j}(T)$ is also a conjugate semistandard tableau and $T$ is determined by $\sub{i}{j}(T)$ and $h_i^j(T)$.
\end{lemma}

While the condition in the above lemma seems strange at first, it arises in a natural way when one repeatedly applies the $\sub{i}{j}$ operators as described by the next claim.
For the sake of completeness, we provide a proof.

\begin{claim}[{implicit in proof of \cite[Corollary 1.7]{CEP80}}] \label{claim:S operator induction}
	Let $T$ be a conjugate semistandard tableau with entries in $[n]$.
	Let 
	\[
		(1,2) \prec (1,3) \prec \cdots \prec (1,n) \prec (2,3) \prec \cdots \prec (n-2,n-1) \prec (n-2,n) \prec (n-1,n)
	\]
	be a partial order on $[n]^2$.
	Let $i, j \in [n]$ be such that $i < j$ and let $(i',j')$ be the immediate predecessor of $(i,j)$ in the $\prec$ order.
	Then the tableau
	\[
		T' \coloneqq \sub{i'}{j'} \circ \cdots \circ \sub{1}{3} \circ \sub{1}{2}(T)
	\]
	satisfies the hypothesis of \autoref{lem:S operator} for $(i,j)$.
	In other words, if a row of $T'$ contains an integer $k \le i$, then that row contains all integers in $\set{i,i+1,\ldots,j-1}$.
\end{claim}

\begin{proof}
	The case of $(i,j) = (1,2)$ is vacuously true.
	Suppose $(i,j) \succ (1,2)$ and that some row $r$ of $T'$ contains an integer $k \le i$.

	\begin{itemize}
		\item
			If $k = i$, then it must be the case that the operator $\sub{i}{j-1} \circ \cdots \circ \sub{i}{i+1}$ did not replace the $i$ in row $r$.
			This implies that the tableau $\sub{i-1}{n} \circ \cdots \circ \sub{1}{2}(T)$ contains every element of $\set{i,i+1,\ldots,j-1}$ in row $r$. 
			Because of this, the operator $\sub{i}{j-1} \circ \cdots \circ \sub{i}{i+1}$ does not modify any of the entries in row $r$ coming from the set $\set{i,i+1,\ldots,j-1}$, so row $r$ of $T'$ contains every element of $\set{i,\ldots,j-1}$.
		\item
			If $k < i$, then the application of the composite operator $\sub{k}{n} \circ \sub{k}{n-1} \circ \cdots \circ \sub{k}{k+1}$ in the definition of $T'$ did not replace the $k$ appearing in row $r$ of $T$.
			This means that every element of $\set{k,\ldots,n}$ appears in row $r$ of the tableau $\sub{k-1}{n} \circ \cdots \circ \sub{1}{2}(T)$.
			Applying the operator $\sub{i'}{j'} \circ \cdots \circ \sub{k}{k+1}$ will not change this, so row $r$ of $T'$ contains every element of $\set{k,\ldots,n}$.
			In particular, every element of $\set{i,\ldots,j-1}$ appears in this row. \qedhere
	\end{itemize}
\end{proof}

For a partition $\sigma$ and natural number $n \in \naturals$, we let $K_\sigma$ and $\overline{K}_\sigma$ denote the conjugate semistandard tableaux whose $i$\ts{th} row has entries $(1,\ldots,\sigma_i)$ and $(n-i+1,n-i+2,\ldots,n)$, respectively.
For example, if $\sigma = (4,3,1)$ and $n = 5$, we have
\[
	K_{(4,3,1)} = \young(1234,123,1) \qquad \overline{K}_{(4,3,1)} = \young(2345,345,5).
\]
The operators $\sub{i}{j}$ provide a convenient way to transform an arbitrary conjugate semistandard tableau into $\overline{K}_\sigma$.

\begin{lemma}[{\cite[Corollary 1.7]{CEP80}}] \label{lem:composite S operator}
	Let $T$ be a conjugate semistandard tableau of shape $\sigma$.
	Then 
	\[
		(\sub{n-1}{n} \circ \sub{n-2}{n} \circ \cdots \circ \sub{2}{3} \circ \sub{1}{n} \circ \cdots \circ \sub{1}{3} \circ \sub{1}{2})(T) = \overline{K}_\sigma.
	\]
	Moreover, if we denote by $h_i^j$ the number of times $i$ is replaced by $j$ in the application of $\sub{i}{j}$ above, then $T$ is determined by $\sigma$ and the $h_i^j$.
\end{lemma}

We are now ready to progress towards the main result of this section.
Namely, for any nonzero $f \in \detideal{n}{m}{r}$, we will find a linear change of variables that sends $f$ to $(K_\sigma | K_\sigma) + O(\eps)$ where $\sigma$ is the shape of some standard bideterminant in the support of $f$ when $f$ is written in the standard bideterminant basis.
For comparison, it is easy to do something similar in the monomial basis: given a polynomial $f(\vec{x})$ of degree $d$, there is some $m \in \naturals$ such that
\[
	\eps^{m} f(\eps^{-(d+1)} x_1, \eps^{-(d+1)^2} x_2, \ldots, \eps^{-(d+1)^n} x_n) = \LC_{\lex}(f) \LM_{\lex}(f) + O(\eps)
\]
where we take the lexicographic monomial order induced by $x_1 \succ x_2 \succ \cdots \succ x_n$.
To some extent, we are constructing an analogous change of variables in the bideterminant basis.

The main difficulty lies in finding a useful change of variables.
In the monomial basis, individual terms can be distinguished by their degree, so it suffices to use a change of variables that only involves multiplying each $x_i$ by some power of $\eps$.
However, in the bideterminant basis, multidegree is too coarse a notion to distinguish between bideterminants, so it seems that finding a clever substitution $x_{i,j} \mapsto \eps^{d_{i,j}} x_{i,j}$ will not be enough.

We start by working in a larger polynomial ring $\F[X,\Lambda,\Xi]$.
We will give two changes of variables: one that enforces structure on the tableaux encoding the rows of the bideterminants in the support of a polynomial $f$, and another that handles the tableaux encoding the columns of the bideterminants.
The proof of this lemma is inspired by and borrows ideas from the proof of \cite[Theorem 3.3]{CEP80}.

\begin{lemma} \label{lem:row/col transform}
	Let $\Lambda = (\lambda_{i,j})$ be an $n \times n$ matrix of variables and let $\prec_{\Lambda}$ be the lexicographic monomial order on $\F[\Lambda]$ induced by the order $\lambda_{i,j} \succ \lambda_{k,\ell}$ if $i < k$ or $i = k$ and $j < \ell$.
	Likewise, let $\Xi = (\xi_{i,j})$ be an $m \times m$ matrix of variables and let $\prec_{\Xi}$ be the corresponding lexicographic monomial order on $\F[\Xi]$.
	Then there are matrices $M \in \F[\Lambda]^{n \times n}$ and $N \in \F[\Xi]^{m \times m}$ with $\det(M) = \pm 1$ and $\det(N) = \pm 1$ such that the following holds.

	Let $f(X) \in \detideal{n}{m}{r}$ be a nonzero polynomial and let $f(X) = \sum_{k \in [s]} \alpha_k (S_k | T_k)(X)$ be the expansion of $f$ in the standard bideterminant basis.
	For $k \in [s]$, let $\sigma_k$ be the shape of the bideterminant $(S_k | T_k)$.
	Then there are nonempty sets $A, B \subseteq [s]$ such that
	\begin{align*}
		\LC_{\prec_{\Lambda}}(f(M X)) &= \sum_{k \in A} \alpha_k (K_{\sigma_k} | T_k)(X) \\
		\LC_{\prec_{\Xi}}(f(X N)) &= \sum_{k \in B} \alpha_k (S_k | K_{\sigma_k})(X),
	\end{align*}
	where we take leading coefficients in the rings $\F[X][\Lambda]$ and $\F[X][\Xi]$, respectively.
\end{lemma}

\begin{proof}
	We first construct the matrix $M$ and prove the corresponding claim.
	For $i, j \in [n]$ with $i \neq j$, let $E_{i,j}(z)$ be the $n \times n$ matrix with ones on the diagonal and $z$ in the $(i,j)$ entry.
	Let $J_n$ be the $n \times n$ matrix whose $(i,j)$ entry is $1$ if $i + j = n+1$ and zero otherwise.
	We define the matrix $M$ as
	\[
		M \coloneqq E_{1,2}(\lambda_{1,2}) E_{1,3}(\lambda_{1,3}) \cdots E_{1,n}(\lambda_{1,n}) E_{2,3}(\lambda_{2,3}) \cdots E_{n-1,n}(\lambda_{n-1,n}) J_n.
	\]
	Since $\det(J_n) = \pm 1$ and $\det(E_{i,j}(z)) = 1$ for $i \neq j$, it follows that $\det(M) = \pm 1$.

	We now analyze the polynomial $f(M X)$.
	Recall that for a tableau $S$, we denote by $h_i^j(S)$ the number of entries changed from $i$ to $j$ when we apply the operator $\sub{i}{j}$ to $S$.
	Observe that for a bideterminant $(S|T)$, it follows from properties of the determinant that
	\[
		(S|T)(E_{i,j}(z) X) = z^{h_i^j(S)} (\sub{i}{j}(S) | T)(X) + O(z^{h_i^j(S) - 1}),
	\]
	where $O(z^{h_i^j(S)-1})$ denotes a polynomial in $\F[X][z]$ of degree at most $h_i^j(S)-1$.
	For $i,j \in [n]$ with $i \neq j$, define
	\[
		f_{i,j}(X,\Lambda) \coloneqq f(E_{1,2}(\lambda_{1,2}) E_{1,3}(\lambda_{1,3}) \cdots E_{1,n}(\lambda_{1,n}) E_{2,3}(\lambda_{2,3}) \cdots E_{i,j}(\lambda_{i,j}) X).
	\]
	Note that $f(M X) = f_{n-1,n}(J_n X, \Lambda)$.

	We claim that for every $i, j \in [n]$ with $i < j$, there is a non-empty set $A_{i,j} \subseteq [s]$ such that
	\[
		\LC_{\prec_{\Lambda}}(f_{i,j}(X,\Lambda)) = \sum_{k \in A_{i,j}} \alpha_k (\sub{i}{j} \circ \cdots \circ \sub{2}{3} \circ \sub{1}{n} \circ \cdots \circ \sub{1}{3} \circ \sub{1}{2}(S_k) | T_k)(X).
	\]
	By \autoref{lem:composite S operator}, this implies
	\[
		\LC_{\prec_{\Lambda}}(f_{n-1,n}(X,\Lambda)) = \sum_{k \in A_{n-1,n}} \alpha_k (\overline{K}_{\sigma_k} | T_k)(X).
	\]
	Using the fact that $(\overline{K}_{\sigma_k} | T)(J_n X) = (K_{\sigma_k} | T_k)(X)$, this yields
	\[
		\LC_{\prec_{\Lambda}}(f(M X)) = \LC_{\prec_{\Lambda}}(f_{n-1,n}(J_n X, \Lambda)) = \sum_{k \in A_{n-1,n}} \alpha_k (K_{\sigma_k} | T_k)(X)
	\]
	as claimed.

	We now prove the claim by induction on $(i,j)$ in the order $(1,2) \prec (1,3) \prec \cdots \prec (1,n) \prec (2,3) \prec \cdots \prec (n-1,n)$.
	Let $(i',j')$ be the predecessor of $(i,j)$ in the $\prec$ order.
	In the case that $(i,j) = (1,2)$, we abuse notation and set $f_{i',j'} \coloneqq f$ and $A_{i',j'} \coloneqq [s]$.
	Let 
	\[
		H_i^j \coloneqq \max_{k \in A_{i',j'}} h_i^j(\sub{i'}{j'} \circ \cdots \circ \sub{1}{2}(S_k))
	\]
	and
	\[
		A_{i,j} = \set{k \in A_{i',j'} : h_i^j(\sub{i'}{j'} \circ \cdots \circ \sub{1}{2}(S_k)) = H_i^j}.
	\]
	Note that $A_{i,j}$ is necessarily non-empty, as $H_i^j$ is a maximum over a finite nonempty set.
	By induction, there is some $\vec{e} \in \naturals^{n \times n}$ such that
	\[
		f_{i',j'}(X,\Lambda) = \Lambda^{\vec{e}} \sum_{k \in A_{i',j'}} \alpha_{k} (\sub{i'}{j'} \circ \cdots \circ \sub{1}{2}(S_k) | T_k)(X) + g(X,\Lambda),
	\]
	where $g(X,\Lambda) \in \F[X][\Lambda]$ is a polynomial in which every monomial is smaller than $\Lambda^{\vec{e}}$ in the $\prec_{\Lambda}$ order.
	Since $f_{i',j'}$ only depends on $\lambda_{1,2},\ldots,\lambda_{i',j'}$, it follows that $\Lambda^{\vec{e}}$ is a monomial in only these variables.
	We then apply the definition of $f_{i,j}$ to obtain
	\begin{align*}
		f_{i,j}(X,\Lambda) &= f_{i',j'}(E_{i,j}(\lambda_{i,j}) X,\Lambda) \\
		&= \Lambda^{\vec{e}} \sum_{k \in A_{i',j'}} \alpha_k (\sub{i'}{j'} \circ \cdots \circ \sub{1}{2}(S_k) | T_k)(E_{i,j}(\lambda_{i,j}) X) + g(E_{i,j}(\lambda_{i,j}) X, \Lambda) \\
		&= \Lambda^{\vec{e}} \lambda_{i,j}^{H_i^j} \sum_{k \in A_{i,j}} \alpha_k (\sub{i}{j} \circ \cdots \circ \sub{1}{2}(S_k) | T_k)(X) + \Lambda^{\vec{e}} p(X, \lambda_{i,j}) + g(E_{i,j}(\lambda_{i,j}) X, \Lambda),
	\end{align*}
	where $p(X,\lambda_{i,j}) \in \F[X][\Lambda]$ is a polynomial of degree at most $H_i^j-1$ in $\lambda_{i,j}$.
	This implies that every monomial of $\Lambda^{\vec{e}} p(X,\Lambda)$ is smaller than $\Lambda^{\vec{e}} \lambda_{i,j}^{H_i^j}$ in the $\prec_{\Lambda}$ order.
	Observe that the substitution $X \mapsto E_{i,j}(\lambda_{i,j}) X$ only changes the $\lambda_{i,j}$-degree of any $\Lambda$-monomial in $g(X,\Lambda)$.
	In particular, because every monomial of $g(X,\Lambda)$ is smaller than $\Lambda^{\vec{e}}$ in the $\prec_{\Lambda}$ order, the same holds true for every $\Lambda$-monomial of $g(E_{i,j}(\lambda_{i,j}) X, \Lambda)$.
	This implies that
	\[
		\LC_{\prec_{\Lambda}}(f_{i,j}) = \sum_{k \in A_{i,j}} \alpha_k (\sub{i}{j} \circ \cdots \circ \sub{1}{2}(S_k) | T_k)(X)
	\]
	as claimed.
	This establishes the claimed properties of $M$.

	To construct the matrix $N$, we overload notation and let $E_{i,j}(z)$ be the $m \times m$ matrix with ones on the diagonal and $z$ in the $(i,j)$ entry.
	Just as the matrix $M$ consisted of a sequence of row operations, the matrix $N$ will be composed of a sequence of column operations.
	We define $N$ as
	\[
		N \coloneqq J_m E_{m-1,m}(\xi_{m-1,m}) \cdots E_{2,3}(\xi_{2,3}) E_{1,m}(\xi_{1,m}) \cdots E_{1,3}(\xi_{1,3}) E_{1,2}(\xi_{1,2}).
	\]
	Since $\det(J_m) = \pm 1$ and $\det(E_{i,j}(z)) = 1$ for $i < j$, we get that $\det(N) = \pm 1$.

	As in the previous case, it follows from properties of the determinant that for a bideterminant $(S|T)$, we have
	\[
		(S|T)(X E_{i,j}(z)) = z^{h_i^j(T)} (S | \sub{i}{j}(T))(X) + O(z^{h_i^j(T)-1}).
	\]
	Using this, the analysis of the leading coefficient of $f(X N) \in \F[X][\Xi]$ proceeds in a manner analogous to the case of $f(MX)$, so we omit the details.
\end{proof}

We now come to the main result of this subsection: a change of variables that sends a polynomial $f(X)$ to $(K_\sigma | K_\sigma)(X) + O(\eps)$ where $\sigma$ is the shape of some standard bideterminant in the support of $f$.

\begin{proposition} \label{prop:reduction to single bideterminant}
	Let $f(X) \in \detideal{n}{m}{r}$ be nonzero.
	There is a collection of $nm$ linearly independent linear functions $\ell_{i,j}(X,\eps) \in \F(\eps)[X]$ indexed by $(i,j) \in [n] \times [m]$, an integer $q \in \integers$, a nonzero $\alpha \in \F$, and a partition $\sigma$ with $\sigma_1 \ge r$ such that
	\[
		f(\ell_{1,1}(X,\eps),\ldots,\ell_{n,m}(X,\eps)) = \eps^q \alpha (K_\sigma | K_\sigma)(X) + O(\eps^{q+1}).
	\]
\end{proposition}

\begin{proof}
	Let $f = \sum_{k \in [s]} \alpha_k (S_k | T_k)$ be the expansion of $f$ in the standard bideterminant basis.
	Let $M$ and $N$ be the matrices constructed in \autoref{lem:row/col transform}.
	Let $\prec$ denote the lexicographic order on $\F[X][\Lambda,\Xi]$ induced by $\lambda_{1,2} \succ \lambda_{1,3} \succ \cdots \succ \lambda_{n-1,n} \succ \xi_{1,2} \succ \cdots \succ \xi_{m-1,m}$. 
	\autoref{lem:row/col transform} implies that there is a non-empty set $A \subseteq [s]$ such that
	\[
		g(X) \coloneqq \LC_{\prec}(f(M X)) = \sum_{k \in A} \alpha_k (K_{\sigma_k} | T_k)(X),
	\]
	and likewise that there is a non-empty set $B \subseteq A$ such that
	\[
		\LC_{\prec}(g(X N)) = \sum_{k \in B} \alpha_k (K_{\sigma_k} | K_{\sigma_k})(X).
	\]
	This implies that
	\[
		\LC_{\prec}(f(M X N)) = \sum_{k \in B} \alpha_k (K_{\sigma_k} | K_{\sigma_k})(X),
	\]
	where $\sigma_k$ denotes the shape of the bideterminant $(S_k | T_k)$.
	By \autoref{cor:width of I_r}, each bideterminant in the above sum has width at least $r$, so $(\sigma_k)_1 \ge r$ for all $k \in A$.

	Let $y$ and $z$ be new indeterminates and let $D \coloneqq \deg(f(X))$.
	Consider the change of variables
	\[
		x_{i,j} \mapsto y^{(D+1)^i} z^{(D+1)^j} x_{i,j}.
	\]
	Let $h(X,\Lambda,\Xi,y,z)$ be the image of $f(M X N)$ under this map.
	By construction, an $X$-monomial of multidegree $(\sum_i a_i \vec{e}_i) \oplus (\sum_i b_i \vec{e}_i)$ is multiplied by a factor of $y^{\sum_i a_i (D+1)^i} z^{\sum_j b_j (D+1)^j}$.
	In particular, since $\max_i a_i \le D$ and $\max_i b_i \le D$, $X$-monomials of distinct multidegree have distinct $(y,z)$-degree under this mapping.
	Observe that $\multideg((K_\sigma | K_\sigma)(X)) \neq \multideg((K_\tau | K_\tau)(X))$ for distinct partitions $\sigma \neq \tau$.
	Since each bideterminant $(K_\sigma | K_\sigma)(X)$ is mapped to a unique $(y,z)$-degree under this substitution, we get that the polynomial
	\[
		p(X) = \LC_{(y,z)}(\LC_{(\Lambda,\Xi)}(h(X,\Lambda,\Xi,y,z)))
	\]
	is a nonzero multiple of the bideterminant $(K_{\sigma_k} | K_{\sigma_k})(X)$ for some $k \in B$.
	If we augment the monomial order $\prec$ by setting $\Lambda \succ \Xi \succ y \succ z$ and taking the corresponding lexicographic order, we then have 
	\[
		\LC_{\prec}(h(X,\Lambda,\Xi,y,z)) = \alpha_k (K_{\sigma_k} | K_{\sigma_k})(X)
	\]
	for some $k \in B$.
	
	Applying \autoref{lem:approximate lc} to $h(X,\Lambda,\Xi,y,z)$ viewed as an element of $\F[X][\Lambda,\Xi,y,z]$, we get a map $\varphi : (\Lambda \cup \Xi \cup \set{y,z}) \to \set{\eps^d : d \in \integers}$ such that
	\[
		\varphi(h(X,\Lambda,\Xi,y,z)) = \eps^q \alpha_k (K_{\sigma_k} | K_{\sigma_k})(X) + O(\eps^{q+1})
	\]
	for some integer $q$.

	Note that $h(X,\Lambda,\Xi,y,z)$ was obtained from $f(X)$ by an invertible linear transformation of the $X$ variables.
	That is, there are $nm$ linearly independent linear polynomials $\ell'_{1,1}(X),\ldots,\ell'_{n,m}(X) \in \F[\Lambda,\Xi,y,z][X]$ such that
	\[
		h(X,\Lambda,\Xi,y,z) = f(\ell_{1,1}'(X),\ldots,\ell_{n,m}'(X)).
	\]
	Set $\ell_{i,j}(X,\eps) \coloneqq \varphi(\ell_{i,j}'(X)) \in \F(\eps)[X]$ for each $(i,j) \in [n] \times [m]$.
	Since the transformation $x_{i,j} \mapsto \ell_{i,j}'(X)$ is invertible as long as $y \neq 0$ and $z \neq 0$, the transformation $x_{i,j} \mapsto \ell_{i,j}(X,\eps)$ remains invertible under $\varphi$.
	Finally, it follows from the definition of $\varphi$ that
	\begin{align*}
		f(\ell_{1,1}(X,\eps),\ldots,\ell_{n,m}(X,\eps)) &= f(\varphi(\ell'_{1,1}(X)),\ldots,\varphi(\ell'_{n,m}(X))) \\
		&= \varphi(f(\ell'_{1,1}(X),\ldots,\ell'_{n,m}(X))) \\
		&= \varphi(h(X,\Lambda,\Xi,y,z)) \\
		&= \eps^q \alpha_k (K_{\sigma_k} | K_{\sigma_k})(X) + O(\eps^{q+1}). \qedhere
	\end{align*}
\end{proof}

\subsection{Projecting to the Determinant}

So far, we have constructed a linear change of variables taking a polynomial $f \in \detideal{n}{m}{r}$ to $(K_\sigma | K_\sigma) + O(\eps)$ for a bideterminant $(K_\sigma | K_\sigma)$ of width at least $r$.
Next, we show that a $(K_\sigma | K_\sigma)$-oracle can be used to compute $g(\vec{y}) + O(\eps)$, where $g$ is any polynomial computable by an algebraic branching program on $r$ vertices.
Ideally, one would like to appeal to the $\VBP$-completeness of the determinant, which gives a projection from $\det_r(X)$ to $g(\vec{y})$, to prove such a result.
The difficulty lies in the fact that a bideterminant may be a product of multiple determinants of varying sizes.
Because of this, we need a projection that behaves well on proper minors of $X$ and also allows us to deal with the possibility that we may be projecting from a power of the determinant as opposed to the determinant itself.
We almost construct such a projection, but we will need some post-processing in the form of an extra addition gate in order to handle powers of the determinant.

Let $g(\vec{y})$ be computable by a small algebraic branching program.
We begin by describing a projection $\varphi : X \to \vec{y} \cup \F$ of a generic matrix $X$ such that $\det(\varphi(X)) = 1 + g(\vec{y})$ and the leading principal minors of $\varphi(X)$ have determinant 1.
This is a small modification of an argument due to \textcite[Theorem 1]{Valiant79}; we include a proof for the sake of completeness.

\begin{lemma} \label{lem:abp to det}
	Let $g(\vec{y}) \in \F[\vec{y}]$ and suppose $g$ can be computed by a layered algebraic branching program on $m$ vertices.
	Then there is an $m \times m$ matrix $A \in \F[\vec{y}]^{m \times m}$ whose entries are linear polynomials in $\vec{y}$ such that
	\begin{enumerate}
		\item
			$\det(A) = 1 + g(\vec{y})$, and
		\item
			for every $k \in [m-1]$, we have $\det(A_{[k],[k]}) = 1$.
	\end{enumerate}
\end{lemma}

\begin{proof}
	We first recall the correspondence between cycle covers in graphs and the determinant.
	Let $G$ be a weighted directed graph on $m$ vertices and denote the weight of the edge $(i,j)$ by $w(i,j)$.
	Let $A(G) = (a_{i,j})$ be the $m \times m$ matrix given by
	\[
		a_{i,j} = \begin{cases}
			w(i,j) & (i,j) \in E(G) \\
			0 & (i,j) \notin E(G).
		\end{cases}
	\]
	Recall that a \emph{cycle cover} $C$ of $G$ is a collection of vertex-disjoint cycles in $G$ which span the vertices of $G$.
	Let $CC(G)$ denote the collection of all cycle covers of $G$.
	Given a cycle cover $C$ of $G$, let $\pi(C)$ denote the product of the edge weights in $C$.
	If every cycle cover of $G$ consists of odd-length cycles, then the definitions of $A(G)$ and the determinant imply that
	\[
		\det(A(G)) = \sum_{C \in CC(G)} \pi(C).
	\]

	We now proceed with the proof of \autoref{lem:abp to det}.
	Suppose $g(\vec{y})$ can be computed by a layered algebraic branching program on $m$ nodes.
	Let $s$ and $t$ be the start and end nodes of this branching program, respectively.
	Since the program is layered, every $s$-$t$ path has the same length.
	If the length of each $s$-$t$ path is even, we add an edge of weight $1$ from $t$ to $s$ and a self-loop of weight 1 to every vertex (including $s$ and $t$); if the length of each $s$-$t$ path is odd, we identify the vertices $s$ and $t$ with one another (resulting in a graph on $m-1$ nodes), add an isolated vertex $r$, and then add a self-loop to every vertex.
	Denote the resulting graph by $G$.
	In both cases, $G$ has one cycle cover for every $s$-$t$ path in the branching program, as well as a single cycle cover corresponding to the set of self-loops in the graph.
	Moreover, every cycle cover in $G$ consists solely of odd-length cycles.

	For a cycle cover $C$ corresponding to an $s$-$t$ path $P$ in the branching program, it follows from the definition of $G$ that $\pi(C) = \pi(P)$, where $\pi(P)$ is the product of the weights on the edges of $P$.
	If $C$ is the all-self-loops cycle cover, then $\pi(C) = 1$.
	Since every cycle cover in $G$ consists of odd-length cycles, we have
	\[
		\det(A(G)) = \sum_{C \in CC(G)} \pi(C) = 1 + \sum_{P} \pi(P) = 1 + g(\vec{y}),
	\]
	where the second summation is over all $s$-$t$ paths $P$ in the branching program.
	This proves the first part of the lemma.

	To prove the second part, let $v_1, \ldots, v_m$ be a topological ordering of the vertices in the algebraic branching program.
	Note that $v_1 = s$ and $v_m = t$.
	If every $s$-$t$ path in the branching program has even length, we order the rows and columns of $A(G)$ such that
	\[
		A(G)_{i,j} = w(v_i,v_j).
	\]
	If instead every $s$-$t$ path in the branching program has odd length, we set
	\[
		A(G)_{i,j} = \begin{cases}
			w(r, v_j) & i = 1 \\
			w(v_i, r) & j = 1 \\
			w(v_i, v_j) & \text{otherwise,}
		\end{cases}
	\]
	where $r$ is the isolated vertex with a self-loop.
	In either case, note that if $i > j$ and $A(G)_{i,j} \neq 0$, then we must have $i = m$.
	This implies that for every $k \in [m-1]$, the matrix $A(G)_{[k],[k]}$ is upper-triangular with ones along the diagonal.
	Thus $\det(A(G)_{[k],[k]}) = 1$ as desired.
\end{proof}

Although we want to construct an $(K_\sigma | K_\sigma)$-oracle circuit that computes any polynomial $g(\vec{y})$ that is computable by a small layered algebraic branching program, it will be convenient for us to assume that $g$ is homogeneous.
This is not restrictive, as one can always introduce a new variable $z$ and consider the homogeneous polynomial $\hat{g}(\vec{y},z) \coloneqq z^{\deg(g)} g(y_1/z,\ldots,y_n/z)$, which specializes to $g(\vec{y})$ under the map $z \mapsto 1$.
One needs to show that $\hat{g}(\vec{y},z)$ is as easy to compute as $g(\vec{y})$.
Below, we provide a proof that this can be done for layered ABPs, although we technically show that this is the case for $z^d g(y_1/z,\ldots,y_n/z)$ for some $d \ge \deg(g)$.

\begin{lemma} \label{lem:abp to hom abp}
	Let $g(\vec{y}) \in \F[\vec{y}]$ be a polynomial and suppose that $g$ can be computed by a layered algebraic branching program on $m$ vertices.
	Let $z$ be a new variable.
	Then there is a homogeneous polynomial $\hat{g}(\vec{y},z) \in \F[\vec{y},z]$ such that $\hat{g}$ can be computed by a layered algebraic branching program on $m$ vertices and that $\hat{g}(\vec{y},1) = g(\vec{y})$.
\end{lemma}

\begin{proof}
	Let $G = (V = V_0 \sqcup V_1 \sqcup \cdots \sqcup V_k, E)$ be an $m$-vertex ABP that computes $g(\vec{y})$, where the $V_i$ are the layers of the ABP.
	Without loss of generality, we assume that no vertex of $G$ computes the zero polynomial; if this is the case, we simply remove such a vertex.
	We relabel the edges of $G$ as follows: if an edge $e \in E$ is labeled by the polynomial $\ell_e(\vec{y}) = \alpha_0 + \sum_{i=1}^n \alpha_i y_i$, we relabel the edge $e$ with $\hat{\ell}_e(\vec{y},z) = \alpha_0 z + \sum_{i=1}^n \alpha_i y_i$.
	Let $\hat{G}$ denote the relabeled ABP.

	It is clear that $\hat{G}$ is an $m$-vertex layered ABP.
	For each vertex $v \in V$, let $g_v(\vec{y})$ be the polynomial computed by $v$ in $G$, and let $\hat{g}_v(\vec{y},z)$ be the polynomial computed at $v$ in $\hat{G}$.
	We claim that for each $i \in \set{0,1,\ldots,k}$ and $v \in V_i$, the polynomial $\hat{g}_v(\vec{y},z)$ is homogeneous of degree $i$ and that $\hat{g}_v(\vec{y},1) = g_v(\vec{y})$.
	We prove this by induction on the depth of the vertex $v$ in $G$, i.e., the layer of $V$ containing $v$.

	If $v \in V_0$, then $\hat{g}_v(\vec{y},z) = g_v(\vec{y}) = 1$ and we are done.
	Otherwise, we have $v \in V_i$ for some $i \ge 1$.
	By definition, we have
	\[
		\hat{g}_v(\vec{y},z) = \sum_{u \in V_{i-1}} \hat{\ell}_{u \to v}(\vec{y},z) \cdot \hat{g}_u(\vec{y},z).
	\]
	By induction, for every $u \in V_{i-1}$, the polynomial $\hat{g}_u(\vec{y},z)$ is a homogeneous degree-$(i-1)$ polynomial that satisfies $\hat{g}_u(\vec{y},1) = g_u(\vec{y})$.
	Furthermore, each nonzero $\hat{\ell}_{u \to v}(\vec{y},z)$ is a homogeneous degree-1 polynomial, so it follows that $\hat{g}_v(\vec{y},z)$ is a homogeneous degree-$i$ polynomial.
	Setting $z \mapsto 1$, we have
	\begin{align*}
		\hat{g}_v(\vec{y},1) &= \sum_{u \in V_{i-1}} \hat{\ell}_{u \to v}(\vec{y},1) \cdot \hat{g}_u(\vec{y},1) \\
		&= \sum_{u \in V_{i-1}} \ell_{u \to v}(\vec{y}) \cdot g_u(\vec{y}) \\
		&= g_v(\vec{y}).
	\end{align*}
	Thus, the polynomial $\hat{g}_v(\vec{y},z)$ is as claimed.

	To finish the proof of the lemma, observe that if $v$ is the output vertex of $G$, then $\hat{g}_v(\vec{y},z)$ is the desired polynomial.
\end{proof}

Given a nonzero $f(X) \in \detideal{n}{m}{r}$, we will use the preceding lemmas together with \autoref{prop:reduction to single bideterminant} to construct a depth-three $f$-oracle circuit computing $\det_{\Theta(r^{1/3})}(X) + O(\eps)$.
In fact, for any polynomial $g(\vec{y})$ computable by a layered algebraic branching program on $r$ vertices, we can construct an $f$-oracle circuit computing $g$.

\begin{theorem} \label{thm:proj to small abp} 
	Let $f(X) \in \detideal{n}{m}{r}$ be a nonzero polynomial and let $h(X,\eps) \in \F\llb \eps \rrb[X]$ be any polynomial such that $h(X,\eps) = f(X) + O(\eps)$.
	Let $g(\vec{y}) \in \F[\vec{y}]$ be a polynomial in the border of layered algebraic branching programs with at most $r$ vertices. 
	Then there is a depth-three $h$-oracle circuit $\Phi$ defined over $\F(\eps)$ such that the following hold.
	\begin{enumerate}
		\item
			$\Phi$ has $nm$ addition gates at the bottom layer, a single $h$-oracle gate in the middle layer, and a single addition gate at the top layer.
		\item
			If $\ch(\F) = 0$, then $\Phi$ computes $g(\vec{y}) + O(\eps)$.
		\item
			If $\ch(\F) = p > 0$, then $\Phi$ computes $g(\vec{y})^{p^k} + O(\eps)$ for some $k \in \naturals$.
	\end{enumerate}
\end{theorem}

\begin{proof}
	By \autoref{lem:exact to approx oracle}, it suffices to prove the theorem in the case where the oracle gates compute $f$ exactly.
	By assumption, there is a polynomial $\tilde{g}(\vec{y},\eps) \in \F[\eps][\vec{y}]$ such that $\tilde{g}(\vec{y},\eps) = g(\vec{y}) + O(\eps)$ and $\tilde{g}(\vec{y},\eps)$ can be computed by a layered algebraic branching program on at most $r$ vertices.
	\autoref{lem:abp to hom abp} implies that there is a homogeneous polynomial $\hat{g}(\vec{y},\eps,z) \in \F[\eps][\vec{y},z]$ computable by a layered algebraic branching program on at most $r$ vertices such that $\hat{g}(\vec{y},\eps,1) = \tilde{g}(\vec{y},\eps)$.

	Applying \autoref{prop:reduction to single bideterminant} to $f(X)$, we obtain linear functions $\ell_{1,1}(X,\eps),\ldots,\ell_{n,m}(X,\eps)$, a nonzero $\alpha \in \F$, and some $q \in \integers$ such that
	\[
		f(\ell_{1,1}(X,\eps),\ldots,\ell_{n,m}(X,\eps)) = \eps^q \alpha (K_\sigma | K_\sigma)(X) + O(\eps^{q+1})
	\]
	for some partition $\sigma$ of width at least $r$.
	Since $\hat{g}(\vec{y},\eps,z)$ can be computed by a layered algebraic branching program on at most $r$ vertices, we can obtain a layered ABP on exactly $r$ vertices computing $\hat{g}(\vec{y},\eps,z)$ by adding isolated vertices.
	Let $A(\vec{y},z) \in \F[\eps][\vec{y},z]^{r \times r}$ be the matrix obtained by applying \autoref{lem:abp to det} to $\hat{g}(\vec{y},\eps,z)$.
	Extend $A(\vec{y},z)$ to an $n \times m$ matrix by adding ones along the main diagonal and zeroes elsewhere.
	Then we have
	\begin{align*}
		f(\ell_{1,1}(A(\vec{y},z),\eps),\ldots,\ell_{n,m}(A(\vec{y},z),\eps)) &= \eps^q \alpha (K_\sigma | K_\sigma)(A(\vec{y},z)) + O(\eps^{q+1}) \\
		&= \eps^q \alpha \prod_{i=1}^{\hat{\sigma}_1} \det_{\sigma_i}(A(\vec{y},z)_{[\sigma_i],[\sigma_i]}) + O(\eps^{q+1}) \\
		&= \eps^q \alpha \prod_{i : \sigma_i \ge r} \det_{\sigma_i}(A(\vec{y},z)_{[\sigma_i],[\sigma_i]}) \cdot \prod_{i : \sigma_i < r} \det_{\sigma_i}(A(\vec{y},z)_{[\sigma_i], [\sigma_i]}) + O(\eps^{q+1}) \\
		&= \eps^q \alpha \prod_{i : \sigma_i \ge r} (1 + \hat{g}(\vec{y},\eps,z)) + O(\eps^{q+1}).
	\end{align*}
	Let $h(\vec{y},\eps,z) \coloneqq f(\ell_{1,1}(A(\vec{y},z),\eps),\ldots,\ell_{n,m}(A(\vec{y},z),\eps))$ and let $t = \Abs{\set{i : \sigma_i \ge r}}$.
	The above establishes $h(\vec{y},\eps,z) = \eps^q \alpha (1 + \hat{g}(\vec{y},\eps,z))^t + O(\eps^{q+1})$.

	Suppose $\ch(\F) = 0$.
	Under the substitution $y_i \mapsto \delta \cdot y_i$ and $z \mapsto \delta$, we have
	\begin{align*}
		h(\delta \cdot \vec{y}, \eps, \delta) &= \eps^q \alpha (1 + \hat{g}(\delta \cdot \vec{y},\eps,\delta))^t + O(\eps^{q+1}) \\
		&= \eps^q \alpha (1 + \delta^{\deg(\hat{g})} \hat{g}(\vec{y},\eps,1))^t + O(\eps^{q+1}) \\
		&= \eps^q \alpha (1 + \delta^{\deg(\hat{g})} g(\vec{y}) + O(\eps))^t + O(\eps^{q+1}) \\
		&= \eps^q \alpha \sum_{i=0}^t \binom{t}{i} \delta^{i \cdot \deg(\hat{g})} g(\vec{y})^i + O(\eps^{q+1}) \\
		&= \eps^q \alpha + \eps^q \delta^{\deg(\hat{g})} \alpha t g(\vec{y}) + O(\eps^q \delta^{2 \deg(\hat{g})}) + O(\eps^{q+1}).
	\end{align*}
	Performing the substitution
	\begin{align*}
		\eps &\mapsto \eps^N \\
		\delta &\mapsto \eps
	\end{align*}
	for $N$ sufficiently large yields
	\[
		h(\eps \cdot \vec{y}, \eps^N, \eps) = \eps^{qN} \alpha + \eps^{qN + \deg(\hat{g})} \alpha t g(\vec{y}) + O(\eps^{qN + \deg(\hat{g}) + 1}).
	\]
	The desired $f$-oracle circuit for $g$ is then given by
	\[
		\Phi(\vec{y}) \coloneqq \frac{h(\eps \cdot \vec{y}, \eps^N, \eps) - \eps^{qN}\alpha}{\eps^{qN + \deg(\hat{g})} \alpha t} = g(\vec{y}) + O(\eps).
	\]

	If instead $\ch(\F) = p > 0$, the above proof only needs to be modified in the case that $p$ divides $t$.
	Let $k \in \naturals$ be the largest natural number such that $p^k$ divides $t$ and write $t = p^k b$.
	In this case, we instead get
	\[
		h(\delta \cdot \vec{y}, \eps, \delta) = \eps^q \alpha + \eps^q \delta^{\deg(\hat{g}) p^k} \alpha b g(\vec{y})^{p^k} + O(\eps^q \delta^{2\deg(\hat{g}) p^k}) + O(\eps^{q+1}).
	\]
	Again, for $N$ sufficiently large, we obtain an $f$-oracle circuit for $g$ via
	\[
		\Phi(\vec{y}) \coloneqq \frac{h(\eps \cdot \vec{y}, \eps^N, \eps) - \eps^{qN} \alpha}{\eps^{qN + \deg(\hat{g}) p^k} \alpha b} = g(\vec{y})^{p^k} + O(\eps). \qedhere
	\]
\end{proof}

We now instantiate \autoref{thm:proj to small abp} with the determinant and iterated matrix multiplication polynomials.
These corollaries are essentially obvious, but seem interesting in their own right and will be of use in later sections.

\begin{corollary} \label{cor:proj to det}
	Let $f(X) \in \detideal{n}{m}{r}$ be a nonzero polynomial and let $h(X,\eps) \in \F\llb \eps \rrb [X]$ be any polynomial such that $h(X,\eps) = f(X) + O(\eps)$.
	Let $t \le O(r^{1/3})$.
	Then there is a depth-three $h$-oracle circuit $\Phi$ defined over $\F(\eps)$ with the following properties.
	\begin{enumerate}
		\item
			The bottom layer of $\Phi$ consists of $nm$ addition gates, the middle layer has a single $h$-oracle gate, and the top layer has a single addition gate.
		\item
			If $\ch(\F) = 0$, then $\Phi$ computes $\det_t(Y) + O(\eps)$.
		\item
			If $\ch(\F) = p > 0$, then $\Phi$ computes $\det_t(Y)^{p^k} + O(\eps)$ for some $k \in \naturals$.
	\end{enumerate}
\end{corollary}

\begin{proof}
	\textcite[Theorem 2]{MV97} constructed a layered ABP on $O(t^3) \le r$ vertices that computes $\det_t(Y)$.
	The corollary then follows from \autoref{thm:proj to small abp}.
\end{proof}

\begin{corollary} \label{cor:proj to imm}
	Let $f(X) \in \detideal{n}{m}{r}$ be a nonzero polynomial and let $h(X,\eps) \in \F\llb \eps \rrb[X]$ be any polynomial such that $h(X,\eps) = f(X) + O(\eps)$.
	Let $w, d \in \naturals$ satisfy $w(d-1) + 2 \le r$.
	Then there is a depth-three $h$-oracle circuit $\Phi$ defined over $\F(\eps)$ with the following properties.
	\begin{enumerate}
		\item
			The bottom layer of $\Phi$ consists of $nm$ addition gates, the middle layer has a single $h$-oracle gate, and the top layer has a single addition gate.
		\item
			If $\ch(\F) = 0$, then $\Phi$ computes $\IMM_{w,d}(\vec{y}) + O(\eps)$.
		\item
			If $\ch(\F) = p > 0$, then $\Phi$ computes $\IMM_{w,d}(\vec{y})^{p^k} + O(\eps)$ for some $k \in \naturals$.
	\end{enumerate}
\end{corollary}

\begin{proof}
	It is clear that $\IMM_{w,d}(\vec{y})$ is computable by a layered algebraic branching program on $w(d-1) + 2 \le r$ vertices.
	\autoref{thm:proj to small abp} completes the proof.
\end{proof}

We conclude this section with a remark on the fact that in characteristic $p > 0$, we only obtain an oracle circuit for a $p$\ts{th} power of the target polynomial $g(\vec{y})$.

\begin{remark}
	Let $\F$ be a field of characteristic $p > 0$.
	If we interpret \autoref{thm:proj to small abp} as a result on ``factoring'' a polynomial $\detideal{n}{m}{r}$, then the appearance of $p$\ts{th} powers in the ``factors'' is not too surprising.
	Most results on polynomial factorization \cite{Kaltofen87,DSY09,KSS15,CKS19a} only guarantee a circuit that computes a $p$\ts{th} power of a factor if the multiplicity of this factor is a multiple of $p^k$ for some $k > 0$.
	In fact, if $f(\vec{x})^p$ can be computed by a size $s$ circuit, it is open whether $f(\vec{x})$ can be computed by a circuit of size $\poly(n,\deg(f),s)$, although some results are known when $n$ is small compared to $s$ \cite{Andrews20}.
\end{remark}

\section{Hardness of Pfaffian Ideals} \label{sec:pfaff}

This section proves an analogue of \autoref{thm:proj to small abp} for ideals generated by sub-Pfaffians of a skew-symmetric matrix.
The outline of the proof is similar to that of \autoref{thm:proj to small abp}, but some technical details must be modified to accommodate the change to Pfaffians.

\subsection{Computing a Standard Monomial}

In this subsection, we construct, for any nonzero $f \in \pfaffideal{2n}{2r}$, a change of variables that takes $f$ to $[K_\sigma](X) + O(\eps)$ for some partition $\sigma$ with $\sigma_1 \ge 2r$.
The outline of the proof is the same as the proof of \autoref{lem:row/col transform}, replacing the straightening law for bideterminants with the corresponding straightening law for Pfaffians.

The following lemma finds a change of variables that takes $f$ to a sum of standard monomials of the form $[K_\sigma](X)$.
This is the Pfaffian analogue of \autoref{lem:row/col transform} and borrows ideas from the proof of \textcite[Lemmas 2.1 and 2.2]{AF80} in a manner analogous to the use of \cite[Theorem 3.3]{CEP80} in proving \autoref{lem:row/col transform}.

\begin{lemma} \label{lem:pfaffian row/col transform}
	Let $\Lambda = (\lambda_{i,j})$ be a $2n \times 2n$ matrix of variables and let $\prec_{\Lambda}$ be the lexicographic monomial order on $\F[\Lambda]$ induced by the order $\lambda_{i,j} \succ \lambda_{k,\ell}$ if $i < k$ or $i = k$ and $j < \ell$.
	Then there is a matrix $M \in \F[\Lambda]^{2n \times 2n}$ with $\det(M) = \pm 1$ such that the following holds.

	Let $f(X) \in \pfaffideal{2n}{2r}$ be a nonzero polynomial and let $f(X) = \sum_{k \in [s]} \alpha_k [S_k](X)$ be the expansion of $f$ as a sum of standard monomials.
	For $k \in [s]$, let $\sigma_k$ be the shape of the tableau $S_k$.
	Then there is a nonempty set $A \subseteq [s]$ such that
	\[
		\LC_{\prec_{\Lambda}}(f(M X M^\top)) = \sum_{k \in A} \alpha_k [K_{\sigma_k}](X)
	\]
	where we take the leading coefficient in the ring $\F[X][\Lambda]$.
\end{lemma}

\begin{proof}
	We begin with the construction of the matrix $M$.
	For $i, j \in [2n]$ with $i < j$, let $E_{i,j}(z)$ denote the matrix which has ones on the diagonal and $z$ in the $(i,j)$ entry.
	We then let $M_{i,j}(\Lambda) \in \F[\Lambda]^{2n \times 2n}$ be the matrix
	\[
		M_{i,j}(\Lambda) \coloneqq E_{1,2}(\lambda_{1,2}) E_{1,3}(\lambda_{1,3}) \cdots E_{1,n}(\lambda_{1,n}) E_{2,3}(\lambda_{2,3}) \cdots E_{i,j}(\lambda_{i,j}).
	\]
	Letting $J_{2n}$ denote the $2n \times 2n$ matrix with ones on the anti-diagonal and zeroes elsewhere, we then define $M = M_{n-1,n}(\Lambda) J_n$.
	It is clear from the definition of $M$ that $\det(M) = \pm 1$.

	We now show that the polynomial $f(M X M^\top)$ behaves as claimed.
	Recall that if $S$ is a Young tableau, we let $h_i^j(S)$ denote the number of entries changed from $i$ to $j$ when the operator $\sub{i}{j}$ is applied to $S$.
	Observe that if $S$ is a one-row tableau, then the multilinearity of the Pfaffian and \autoref{lem:pfaff orthogonal transform} imply
	\[
		[S](E_{i,j}(z) X E_{i,j}(z)^\top) = \begin{cases}
			[S](X) + z [\sub{i}{j}(S)](X) & \text{if $i$ appears in $S$ but $j$ does not} \\
			[S](X) & \text{otherwise.}
		\end{cases}
	\]
	Note that if both $i$ and $j$ appear in $S$ or if neither appear in $S$, then $S = \sub{i}{j}(S)$.
	Thus, viewing the above as a polynomial in $\F[X][z]$, we see that the leading term is $z^{h_i^j(S)} [\sub{i}{j}(S)](X)$.
	This extends to a multi-row tableau $S$ via
	\[
		[S](E_{i,j}(z) X E_{i,j}(z)^\top) = z^{h_i^j(S)} [\sub{i}{j}(S)](X) + O(z^{h_i^j(S) - 1}),
	\]
	where $O(z^{h_i^j(S)-1})$ denotes a polynomial in $\F[X][z]$ of degree at most $h_i^j(S) - 1$.

	For $i, j \in [2n]$ with $i < j$, let 
	\[
		f_{i,j}(X,\Lambda) \coloneqq f(M_{i,j}(\Lambda) X M_{i,j}(\Lambda)^\top).
	\]
	Note that $f(M X M^\top) = f_{n-1,n}(J_n X J_n^\top)$.
	We claim that for every $i, j \in [2n]$ with $i < j$, there is a nonempty set $A_{i,j} \subseteq [s]$ such that
	\[
		\LC_{\prec_{\Lambda}}(f_{i,j}(X,\Lambda)) = \sum_{k \in A_{i,j}} \alpha_k [\sub{i}{j} \circ \cdots \sub{2}{3} \circ \sub{1}{n} \circ \cdots \circ \sub{1}{2} (S_k)](X).
	\]
	Assuming this, \autoref{lem:composite S operator} implies
	\[
		\LC_{\prec_\Lambda}(f_{n-1,n}(X,\Lambda)) = \sum_{k \in A_{n-1,n}} \alpha_k [\overline{K}_{\sigma_k}](X).
	\]
	From this, we use the fact that $[\overline{K}_{\sigma}](J_n X J_n^\top) = [K_\sigma](X)$ to obtain
	\begin{align*}
		\LC_{\prec_\Lambda}(f(M X M^\top)) &= \LC_{\prec_\Lambda}(f_{n-1,n}(J_n X J_n^\top)) \\
		&= \sum_{k \in A_{n-1,n}} \alpha_k [\overline{K}_{\sigma_k}](J_n X J_n^\top) \\
		&= \sum_{k \in A_{n-1,n}} \alpha_k [K_{\sigma_k}](X)
	\end{align*}
	as desired.

	It remains to prove the claim about $\LC_{\prec_\Lambda}(f_{i,j}(X,\Lambda))$.
	We proceed by induction on $(i,j)$ in the order $(1,2) \prec (1,3) \prec \cdots \prec (1,n) \prec (2,3) \prec \cdots \prec (n-1, n)$.
	Let $(i',j')$ be the predecessor of $(i,j)$ in the $\prec$ order.
	If $(i,j) = (1,2)$, we set $f_{i',j'}(X,\Lambda) = f(X)$ and $A_{i',j'} = [s]$.
	Let
	\[
		H_i^j \coloneqq \max_{k \in A_{i',j'}} h_i^j(\sub{i'}{j'} \circ \cdots \circ \sub{1}{2}(S_k))
	\]
	and
	\[
		A_{i,j} = \set{k \in A_{i',j'} : h_i^j(\sub{i'}{j'} \circ \cdots \circ \sub{1}{2}(S_k)) = H_i^j}.
	\]
	The set $A_{i,j}$ is necessarily nonempty, as $H_i^j$ is obtained by maximizing over a finite nonempty set.
	By induction, there is some $\vec{e} \in \naturals^{2n \times 2n}$ such that
	\[
		f_{i',j'}(X,\Lambda) = \Lambda^{\vec{e}} \sum_{k \in A_{i',j'}} \alpha_k [\sub{i'}{j'} \circ \cdots \circ \sub{1}{2}(S_k)](X) + g(X,\Lambda),
	\]
	where $g(X,\Lambda) \in \F[X][\Lambda]$ is a polynomial supported on monomials that are smaller than $\Lambda^{\vec{e}}$ in the $\prec_{\Lambda}$ order.
	Because $f_{i',j'}$ only depends on $\lambda_{1,2},\ldots,\lambda_{i',j'}$, we know that $\Lambda^{\vec{e}}$ is a monomial consisting of only these variables.
	Applying the definition of $f_{i,j}$, we then have
	\begin{align*}
		f_{i,j}(X,\Lambda) &= f_{i',j'}(E_{i,j}(\lambda_{i,j}) X E_{i,j}(\lambda_{i,j})^\top,\Lambda) \\
		&= \Lambda^{\vec{e}} \sum_{k \in A_{i',j'}} \alpha_k [\sub{i'}{j'} \circ \cdots \circ \sub{1}{2}(S_k)](E_{i,j}(\lambda_{i,j}) X E_{i,j}(\lambda_{i,j})^\top) + g(E_{i,j}(\lambda_{i,j}) X E_{i,j}(\lambda_{i,j})^\top,\Lambda) \\
		&= \Lambda^{\vec{e}} \lambda_{i,j}^{H_i^j} \sum_{\alpha \in A_{i,j}} \alpha_k [\sub{i}{j} \circ \cdots \circ \sub{1}{2}(S_k)](X) + \Lambda^{\vec{e}} p(X,\lambda_{i,j}) + g(E_{i,j}(\lambda_{i,j}) X E_{i,j}(\lambda_{i,j})^\top, \Lambda),
	\end{align*}
	where $p(X,\lambda_{i,j}) \in \F[X][\Lambda]$ is a polynomial of degree at most $H_i^j - 1$ in $\lambda_{i,j}$.
	Because of this, every monomial of $\Lambda^{\vec{e}} p(X,\lambda_{i,j})$ is smaller than $\Lambda^{\vec{e}} \lambda_{i,j}^{H_i^j}$ in the $\prec_{\Lambda}$ order.
	The same holds true for $g(E_{i,j}(\lambda_{i,j}) X E_{i,j}(\lambda_{i,j})^\top,\Lambda)$, as the substitution $X \mapsto E_{i,j}(\lambda_{i,j}) X E_{i,j}(\lambda_{i,j})^\top$ only changes the $\lambda_{i,j}$-degree of a monomial in $g(X,\Lambda)$ and every monomial of $g(X,\Lambda)$ is already smaller than $\Lambda^{\vec{e}}$ in the $\prec_\Lambda$ order.
	This implies that
	\[
		\LC_{\prec_\Lambda}(f_{i,j}(X,\Lambda)) = \sum_{k \in A_{i,j}} \alpha_k [\sub{i}{j} \circ \cdots \circ \sub{1}{2}](X)
	\]
	as claimed.
\end{proof}

We now use the result of \autoref{lem:pfaffian row/col transform} to construct a change of variables that takes a nonzero $f \in \pfaffideal{2n}{2r}$ to $[K_\sigma](X) + O(\eps)$ for a partition $\sigma$ of width at least $2r$.
This is the analogue of \autoref{prop:reduction to single bideterminant} for Pfaffians.
The proof is similar to that of \autoref{prop:reduction to single bideterminant}: after applying \autoref{lem:pfaffian row/col transform}, we scale the rows and columns of $X$ by powers of a new variable $y$ to isolate a single standard monomial $[K_\sigma](X)$.

\begin{proposition} \label{prop:sparsify pfaffian}
	Let $f(X) \in \pfaffideal{2n}{2r}$ be nonzero.
	There is a collection of $4n^2$ linearly independent linear functions $\ell_{i,j}(X,\eps) \in \F(\eps)[X]$ indexed by $(i,j) \in [2n] \times [2n]$, an integer $q \in \integers$, a nonzero $\alpha \in \F$, and a partition $\sigma$ with $\sigma_1 \ge 2r$ such that
	\[
		f(\ell_{1,1}(X,\eps),\ldots,\ell_{2n,2n}(X,\eps)) = \eps^q \alpha [K_\sigma](X) + O(\eps^{q+1}).
	\]
\end{proposition}

\begin{proof}
	Let $M \in \F[\Lambda]^{2n \times 2n}$ be the matrix constructed in \autoref{lem:pfaffian row/col transform}.
	Let $f(X) = \sum_{k \in [s]} \alpha_k [S_k](X)$ be the expansion of $f$ as a sum of standard monomials.
	Then \autoref{lem:pfaffian row/col transform} implies
	\[
		\LC_{\prec_\Lambda}(f(M X M^\top)) = \sum_{k \in A} \alpha_k [K_{\sigma_k}](X),
	\]
	where $A \subseteq [s]$ is nonempty and $\sigma_k$ is the shape of the tableau $S_k$.
	From \autoref{cor:pfaff ideal width}, we know that $(\sigma_k)_1 \ge 2r$ for all $k \in A$.

	Let $d \coloneqq \deg(f(X))$.
	Let $y$ be a new indeterminate and let $D \in \F[y]^{2n \times 2n}$ be the diagonal matrix given by $D_{i,i} = (d+1)^i$.
	Observe that $(D X D^\top)_{[k],[k]} = D_{[k],[k]} X_{[k],[k]} D^{\top}_{[k],[k]}$.
	Using this and \autoref{lem:pfaff orthogonal transform}, we have
	\begin{align*}
		\Pf_k(D X D^\top) &= \Pf(D_{[k],[k]} X_{[k],[k]} D^\top_{[k],[k]}) \\
		&= \det(D_{[k],[k]}) \Pf_k(X) \\
		&= y^{\sum_{i=1}^k (d+1)^i} \Pf_k(X).
	\end{align*}
	It then follows that for a partition $\sigma$, we have
	\begin{align*}
		[K_\sigma](D X D^\top) &= \prod_{i=1}^{\hat{\sigma}_1} \Pf_{\sigma_i}(D X D^\top) \\
		&= \prod_{i=1}^{\hat{\sigma}_1} y^{\sum_{j=1}^{\sigma_i} (d+1)^j} \Pf_{\sigma_i}(X) \\
		&= y^{\sum_{i=1}^{\hat{\sigma}_1} \sum_{j=1}^{\sigma_i} (d+1)^j} [K_\sigma](X) \\
		&= y^{\sum_{i=1}^{\sigma_1} \hat{\sigma}_i (d+1)^i} [K_\sigma](X).
	\end{align*}

	Suppose $\sigma$ and $\tau$ are distinct partitions with $\max(\hat{\sigma}_1,\hat{\tau}_1) \le d$.
	Then we can interpret $\deg_y([K_\sigma](D X D^\top))$ and $\deg_y([K_\tau](D X D^\top))$ as numbers in base $d+1$.
	Because these numbers differ in at least one place value, we have $\deg_y([K_\sigma](D X D^\top)) \neq \deg_y([K_\tau](D X D^\top))$.
	In particular, if $\sigma$ and $\tau$ are distinct shapes of tableaux appearing in the support of $f(X)$, then by our choice of $d$ we have $\max(\hat{\sigma}_1,\hat{\tau}_1) \le d$, so $\deg_y([K_\sigma](D X D^\top)) \neq \deg_y([K_\tau](D X D^\top))$.

	Consider the polynomial $f(M D X D^\top M^\top)$.
	The preceding discussion implies
	\begin{align*}
		\LC_y(\LC_{\prec_\Lambda}(f(M D X D^\top M^\top))) &= \LC_y\del{\sum_{k \in A} \alpha_k [K_{\sigma_k}](D X D^\top)} \\
		&= \LC_y \del{\sum_{k \in A} y^{\sum_{i=1}^{(\hat{\sigma_k})_1} (\hat{\sigma_k})_i\, (d+1)^i} [K_{\sigma_k}](X)} \\
		&= \alpha_{k} [K_{\sigma_{k}}](X)
	\end{align*}
	for some fixed $k \in A$.

	By taking leading coefficients in the ring $\F[X][\Lambda,y]$ with respect to the lexicographic order that sets $\Lambda \succ y$, we then have
	\[
		\LC(f(M D X D^\top M^\top)) = \alpha_{k} [K_{\sigma_{k}}](X).
	\]
	Invoking \autoref{lem:approximate lc} yields a map $\varphi : \Lambda \cup \set{y} \to \set{\eps^i : i \in \integers, i \neq 0}$ that, when extended to a homomorphism $\varphi : \F[X,\Lambda,y] \to \F(\eps)[X]$, gives us
	\[
		\varphi(f(M D X D^\top M^\top)) = \eps^q \alpha_k [K_{\sigma_k}](X) + O(\eps^{q+1})
	\]
	for some $q \in\integers$.
	Finally, the transformation $X \mapsto \varphi(MD) X \varphi(D^\top M^\top)$ is linear and invertible, since $\det(\varphi(M)) = \pm 1$ and $\det(\varphi(D)) = \eps^m$ for some nonzero $m \in \integers$.
\end{proof}

\subsection{Projecting to the Pfaffian}

The previous subsection yields a change of variables that takes any nonzero $f \in \pfaffideal{2n}{2r}$ to $[K_\sigma](X) + O(\eps)$ for some partition $\sigma$ of width at least $2r$.
As in the case of the determinant, we now want to find a projection of $X$ that takes $[K_\sigma](X)$ to $\Pf_m(X)$ for $m$ as large as possible.
Na\"ively, we would like to combine \autoref{lem:abp to det} with \autoref{lem:pfaff to det} to achieve this.
This nearly works, but suffers from the drawback that for a matrix $A$, the Pfaffians of the leading principal submatrices of 
\[
	\begin{pmatrix}
		0 & A \\ -A^\top & 0
	\end{pmatrix}
\]
do not correspond to minors of the leading principal submatrices of $A$.
However, we can amend this by suitably permuting the rows and columns of the above matrix to obtain a new matrix whose leading principal sub-Pfaffians do correspond to minors of leading principal submatrices of $A$.

\begin{lemma} \label{lem:subpfaff to subdet} 
	Let $A$ be an $n \times n$ matrix.
	Then there is a $2n \times 2n$ skew-symmetric matrix $M$ such that for every $k \in [n]$, we have $\Pf(M_{[2k],[2k]}) = \pm \det(A_{[k],[k]})$.
\end{lemma}

\begin{proof}
	Let $\sigma \in S_{2n}$ be the permutation sending $(1,2,\ldots,2n)$ to $(1,n+1,2,n+2,\ldots,n,2n)$.
	Let 
	\[
		B = \begin{pmatrix}
			0 & A \\
			-A^\top & 0 
			\end{pmatrix}
	\]
	and let $C$ be the permutation matrix corresponding to $\sigma$, i.e., $c_{i,j} = 1$ if and only if $j = \sigma(i)$.
	Then $M \coloneqq C B C^\top$ is the matrix whose $i$\ts{th} row (respectively $j$\ts{th} column) is row $\sigma(i)$ (respectively column $\sigma(j)$) of $B$.
	We claim that for all $k \in [n]$, we have $\Pf(M_{[2k],[2k]}) = \pm \det(A_{[k],[k]})$.

	To see this, let $k \in [n]$ be arbitrary.
	Let $\tau \in S_{2k}$ be the permutation sending $(1,k+1,2,k+2,\ldots,k,2k)$ to $(1,2,3,\ldots,2k)$ and let $D$ be the corresponding permutation matrix.
	We will show that $\Pf(D M_{[2k],[2k]} D^\top) = \pm \det(A_{[k],[k]})$.
	By \autoref{lem:pfaff orthogonal transform}, this implies $\Pf(M_{[2k],[2k]}) = \pm \det(A_{[k],[k]})$, so $M$ behaves as desired.

	It remains to show that $\Pf(D M_{[2k],[2k]} D^\top) = \pm \det(A_{[k],[k]})$.
	Note that for $i \in [2k]$, we have 
	\[
		\sigma(\tau(i)) = \begin{cases}
			i & \text{if } i \le k \\
			i - k + n & \text{if } i > k.
		\end{cases}
	\]
	For $i, j \in [2k]$, we have, by definition, 
	\begin{align*}
		(D M_{[2k],[2k]} D^\top)_{i,j} &= (M_{[2k],[2k]})_{\tau(i),\tau(j)} \\
		&= M_{\tau(i),\tau(j)} \\
		&= (CBC^\top)_{\tau(i),\tau(j)} \\
		&= B_{\sigma(\tau(i)), \sigma(\tau(j))} \\
		&= \begin{cases}
			0 & \text{if } i \le k \text{ and } j \le k \\
			A_{i,j} & \text{if } i \le k \text{ and } j > k \\
			-A_{j,i} & \text{if } i > k \text{ and } j \le k \\
			0 & \text{if } i > k \text{ and } j > k.
		\end{cases}
	\end{align*}
	Thus, the matrix $D M_{[2k],[2k]} D^\top$ is the $2k \times 2k$ matrix given by
	\[
		D M_{[2k],[2k]} D^\top = \begin{pmatrix}
			0 & A_{[k],[k]} \\
			-A^\top_{[k],[k]} & 0
		\end{pmatrix}.
	\]
	It follows from \autoref{lem:pfaff to det} that 
	\[
		\Pf(D M_{[2k],[2k]} D^\top) = (-1)^{\binom{k}{2}} \det(A_{[k],[k]})
	\]
	as needed.
\end{proof}

We are now ready to conclude our main result for Pfaffian ideals, an analogue of \autoref{thm:proj to small abp} for Pfaffians.
The proof is similar to the proof of \autoref{thm:proj to small abp}, but augments the use of \autoref{lem:abp to det} with \autoref{lem:subpfaff to subdet}.

\begin{theorem} \label{thm:pfaff to small abp}
	Let $X$ be a $2n \times 2n$ generic skew-symmetric matrix.
	Let $f(X)$ be a nonzero polynomial in the ideal generated by the Pfaffians of the principal $2r \times 2r$ submatrices of $X$.
	Let $h(X, \eps) \in \F \llb \eps \rrb [X]$ be any polynomial such that $h(X,\eps) = f(X) + O(\eps)$.
	Let $g(\vec{y}) \in \F[\vec{y}]$ be a polynomial in the border of layered algebraic branching programs with at most $r$ vertices.
	Then there is a depth-three $h$-oracle circuit $\Phi$ defined over $\F(\eps)$ such that the following hold.
	\begin{enumerate}
		\item
			$\Phi$ has $nm$ addition gates at the bottom layer, a single $h$-oracle gate in the middle layer, and a single addition gate at the top layer.
		\item
			If $\ch(\F) = 0$, then $\Phi$ computes $g(\vec{y}) + O(\eps)$.
		\item
			If $\ch(\F) = p > 0$, then $\Phi$ computes $g(\vec{y})^{p^k} + O(\eps)$ for some $k \in \naturals$.
	\end{enumerate}
\end{theorem}

\begin{proof}
	Using \autoref{lem:exact to approx oracle}, we only need to consider the case where the oracle gates compute $f$ exactly.
	By assumption, there is a polynomial $\tilde{g}(\vec{y},\eps) \in \F[\eps][\vec{y}]$ such that $\tilde{g}(\vec{y},\eps) = g(\vec{y}) + O(\eps)$ and $\tilde{g}(\vec{y},\eps)$ can be computed by a layered algebraic branching program on at most $r$ vertices.
	Further, \autoref{lem:abp to hom abp} yields a homogeneous polynomial $\hat{g}(\vec{y},\eps,z) \in \F[\eps][\vec{y},z]$ computable by a layered algebraic branching program on at most $r$ vertices such that $\hat{g}(\vec{y},\eps,1) = \tilde{g}(\vec{y},\eps)$.
	In what follows, we work with $\hat{g}(\vec{y},\eps,z)$.

	Applying \autoref{prop:sparsify pfaffian} to $f(X)$ gives us linear functions $\ell_{i,j}(X,\eps) \in \F(\eps)[X]$, an integer $q \in \integers$, and a nonzero $\alpha \in \F$ such that
	\[
		f(\ell_{1,1}(X,\eps),\ldots,\ell_{2n,2n}(X,\eps)) = \eps^q \alpha [K_\sigma](X) + O(\eps^{q+1})
	\]
	for some partition $\sigma$ with $\sigma_1\ge 2r$.

	Because $\hat{g}(\vec{y},\eps,z)$ can be computed by a layered ABP on at most $r$ vertices, we can obtain a layered ABP with exactly $r$ vertices that computes $\hat{g}(\vec{y},\eps,z)$ by adding dummy vertices if necessary.
	Let $A(\vec{y},z) \in \F[\eps][\vec{y},z]^{r \times r}$ be the matrix obtained by applying \autoref{lem:abp to det} to $\hat{g}(\vec{y},\eps,z)$.
	Extend $A(\vec{y},z)$ to an $n \times n$ matrix by adding ones along the diagonal and zeroes elsewhere.

	Let $M(\vec{y},z)$ be the $2n \times 2n$ matrix obtained by applying \autoref{lem:subpfaff to subdet} to $A(\vec{y},z)$.
	Let $\varphi : X \to \F[\eps][\vec{y},z]$ be the substitution given by $\varphi(X) = M(\vec{y},z)$.
	Under this substitution, we have
	\begin{align*}
		& f(\ell_{1,1}(\varphi(X),\eps), \ldots, \ell_{2n,2n}(\varphi(X),\eps)) \\
		&\qquad = \eps^q \alpha [K_\sigma](\varphi(X)) + O(\eps^{q+1}) \\
		&\qquad = \eps^q \alpha \prod_{i=1}^{\hat{\sigma}_1} \Pf_{\sigma_i}(\varphi(X)_{[\sigma_i],[\sigma_i]}) + O(\eps^{q+1}) \\
		&\qquad= \pm \eps^q \alpha \prod_{i=1}^{\hat{\sigma}_1} \det_{\sigma_i/2}(A(\vec{y},z)_{[\sigma_i/2],[\sigma_i/2]}) + O(\eps^{q+1}) \\
		&\qquad= \pm \eps^q \alpha \prod_{i : \sigma_i \ge 2r} \det_{\sigma_i/2}(A(\vec{y},z)_{[\sigma_i/2],[\sigma_i/2]}) \cdot \prod_{i : \sigma_i < 2r} \det_{\sigma_i/2}(A(\vec{y},z)_{[\sigma_i/2],[\sigma_i/2]}) + O(\eps^{q+1}) \\
		&\qquad= \pm \eps^q \alpha \prod_{i : \sigma_i \ge 2r} (1 + \hat{g}(\vec{y},\eps,z)) + O(\eps^{q+1}).
	\end{align*}
	Let $h(\vec{y},\eps,z) \coloneqq f(\ell_{1,1}(\varphi(X),\eps),\ldots,\ell_{2n,2n}(\varphi(X),\eps))$ and let $t \coloneqq \abs{\set{i : \sigma_i \ge 2r}}$.
	In this notation, the above establishes that $h(\vec{y},\eps,z) = \pm \eps^q \alpha (1 + \hat{g}(\vec{y},\eps,z))^t + O(\eps^{q+1})$.

	Suppose $\ch(\F) = 0$.
	Let $\delta$ be a new indeterminate.
	By performing the substitutions $y_i \mapsto \delta y_i$ and $z \mapsto \delta$, we obtain
	\begin{align*}
		h(\delta y_1, \ldots, \delta y_m, \eps, \delta) &= \pm \eps^q \alpha \del{1 + \hat{g}(\delta \cdot \vec{y},\eps,\delta)}^t + O(\eps^{q+1}) \\
		&= \pm \eps^q \alpha \del{1 + \delta^{\deg(\hat{g})} \hat{g}(\vec{y},\eps,1)}^t + O(\eps^{q+1}) \\
		&= \pm \eps^q \alpha \del{1 + \delta^{\deg(\hat{g})} g(\vec{y}) + O(\eps)}^t + O(\eps^{q+1}) \\
		&= \pm \eps^q \alpha \sum_{i=0}^t \binom{t}{i} \delta^{t \cdot \deg(\hat{g})} g(\vec{y})^t + O(\eps^{q+1}) \\
		&= \pm \eps^q \alpha \pm \eps^q \delta^{\deg(\hat{g})} \alpha t g(\vec{y}) + O(\eps^q \delta^{2 \deg(\hat{g})}) + O(\eps^{q+1}).
	\end{align*}
	Setting
	\begin{align*}
		\eps &\mapsto \eps^N \\
		\delta &\mapsto \eps
	\end{align*}
	for $N$ sufficiently large yields
	\[
		h(\eps y_1,\ldots, \eps y_m, \eps^N, \eps) = \pm \eps^{qN} \alpha \pm \eps^{qN + \deg(\hat{g})} \alpha t g(\vec{y}) + O(\eps^{qN + \deg(\hat{g}) + 1}).
	\]
	The claimed $f$-oracle circuit is then given by
	\[
		\Phi(\vec{y}) \coloneqq \frac{h(\eps y_1,\ldots,\eps y_m, \eps^N, \eps) \mp \eps^{qN} \alpha}{\pm \eps^{qN + \deg(\hat{g})} \alpha t} = g(\vec{y}) + O(\eps).
	\]

	In the case that $\ch(\F) = p > 0$, we need to modify the above argument in the case that $p$ divides $t$.
	Let $k \in \naturals$ be such that $p^k$ is the largest power of $p$ that divides $t$.
	Write $t = p^k b$.
	We then have
	\[
		h(\delta y_1, \ldots, \delta y_m, \eps, \delta) = \pm \eps^q \alpha \pm \eps^q \delta^{p^k \deg(\hat{g})} \alpha b g(\vec{y})^{p^k} + O(\eps^q \delta^{2 p^k \deg(\hat{g})}) + O(\eps^{q+1}).
	\]
	Again, for sufficiently large $N$, we can construct an $f$-oracle circuit that approximately computes $g$ via
	\[
		\Phi(\vec{y}) \coloneqq \frac{h(\eps y_1, \ldots, \eps y_m, \eps^N, \eps) \mp \eps^{qN} \alpha}{\pm \eps^{qN + p^k \deg(\hat{g})} \alpha b} = g(\vec{y})^{p^k} + O(\eps). \qedhere
	\]
\end{proof}

Since the Pfaffian can be computed efficiently by algebraic branching programs, we immediately obtain the following corollary of \autoref{thm:pfaff to small abp}.

\begin{corollary} \label{cor:proj to pfaff}
	Let $f(X) \in \pfaffideal{2n}{2r}$ be a nonzero polynomial, let $h(X,\eps) \in \F\llb \eps \rrb [X]$ be any polynomial such that $h(X,\eps) = f(X) + O(\eps)$, and let $t \le O(r^{1/3})$.
	Then there is a depth-three $h$-oracle circuit $\Phi$ defined over $\F(\eps)$ with the following properties.
	\begin{enumerate}
		\item
			The bottom layer of $\Phi$ consists of $4n^2$ addition gates, the middle layer has a single $h$-oracle gate, and the top layer has a single addition gate.
		\item
			If $\ch(\F) = 0$, then $\Phi$ computes $\Pf_t(X) + O(\eps)$.
		\item
			If $\ch(\F) = p > 0$, then $\Phi$ computes $\Pf_t(X)^{p^k} + O(\eps)$ for some $k \in \naturals$.
	\end{enumerate}
\end{corollary}

\begin{proof}
	\textcite[Theorem 12]{MSV04} constructed a layered algebraic branching program of size $O(n^3)$ that computes the $2n \times 2n$ Pfaffian.
	Combining this with \autoref{thm:pfaff to small abp} completes the proof.
\end{proof}

\section{Partial Derivatives in Determinantal Ideals} \label{sec:partials}

We now proceed to our applications of \autoref{thm:proj to small abp}.
Our first such application is the determination of the minimum possible value of $\dim(\partial_{< \infty}(f))$ for a nonzero $f \in \detideal{n}{m}{r}$.
The dimension of the space of partial derivatives (and variants thereof) has been successfully used as a complexity measure in proving algebraic circuit lower bounds.
Though \autoref{thm:proj to small abp} gives us a tool to prove circuit lower bounds for any nonzero polynomial $f(X) \in \detideal{n}{m}{r}$, there may be instances where the $f$-oracle circuit is too costly to implement.
For example, if $f$ is computed by a homogeneous or read-once circuit, these properties are not inherited by the oracle circuit.
In such cases, it may be useful to have direct estimates for $\dim(\partial_{< \infty}(f))$.

For notational convenience, let
\[
	\dim(\partial_{< \infty}(\detideal{n}{m}{r})) \coloneqq \min_{f \in \detideal{n}{m}{r} \setminus \set{0}} \dim(\partial_{< \infty}(f)).
\]
Since $\det_r(X) \in \detideal{n}{m}{r}$ and $\dim(\partial_{< \infty}(\det_r)) = \binom{2r}{r}$, we clearly have $\dim(\partial_{<\infty}(\detideal{n}{m}{r})) \le \binom{2r}{r}$.
Combining \autoref{cor:proj to det} with \autoref{lem:pd linear transformation} establishes the existence of a universal constant $c > 0$ such that $\dim(\partial_{< \infty}(\detideal{n}{m}{r})) \ge \dim(\partial_{<\infty}(\det_{c r^{1/3}})) = \binom{2c r^{1/3}}{c r^{1/3}}$.

Alternatively, one can use the observation of \textcite[Lemma 6.4]{FSTW16} that the set of rank-$r$ matrices contains all $r$-sparse vectors in $\F^{n \times m}$. 
This implies that rank-$r$ matrices are a hitting set for all polynomials that have a monomial supported on at most $r$ variables.
For any $f \in \detideal{n}{m}{r+1}$, it follows by definition that $f(X)$ vanishes on matrices of rank $r$.
This implies that the leading monomial of $f(X)$ is supported on at least $r+1$ variables.
From here, it is straightforward to conclude that there are at least $2^{r+1}$ distinct leading monomials among the partial derivatives of $f(X)$, which implies the stronger lower bound $\dim(\partial_{<\infty}(\detideal{n}{m}{r+1})) \ge 2^{r+1}$.

In this section, we will show that the na\"ive upper bound on $\dim(\partial_{<\infty}(\detideal{n}{m}{r}))$ is tight.
That is, 
\[
	\dim(\partial_{< \infty}(\detideal{n}{m}{r})) = \dim(\partial_{< \infty}(\det_r)) = \binom{2r}{r}.
\]
If one interprets $\dim(\partial_{< \infty}(f))$ as a measure of the complexity of $f$, then this says the $r \times r$ determinant $\det_r(X)$ is of minimal complexity in $\detideal{n}{m}{r}$.

We will show that $\dim(\partial_{<\infty}((S|T))) \ge \binom{2r}{r}$ for any bideterminant $(S|T) \in \detideal{n}{m}{r}$ and then extend this to all nonzero polynomials in $\detideal{n}{m}{r}$ using \autoref{prop:reduction to single bideterminant}.
We start by considering partial derivatives with respect to a single variable.
Recall that the operator $\frac{\partial^d}{\partial x_{i,j}^d}$ refers to the order-$d$ Hasse derivative with respect to $x_{i,j}$.
In the lemma below, we abuse notation and allow a bitableau to have rows whose lengths are not necessarily nonincreasing.

\begin{lemma} \label{lem:bideterminant single partial}
	Let $X$ be an $n \times m$ matrix of variables and let $(S|T)(X)$ be a nonzero bideterminant of shape $\sigma$.
	Let $(i,j) \in [n] \times [m]$ and let $d \coloneqq \ideg_{x_{i,j}}(S|T)(X)$.
	Then
	\[
		\frac{\partial^d}{\partial x_{i,j}^d}(S|T)(X) = \pm (S'|T')(X),
	\]
	where $(S',T')$ is the bitableau whose $k$\ts{th} row $(S'(k,\bullet), T'(k, \bullet))$ is given by
	\[
		(S'(k, \bullet), T'(k, \bullet)) = \begin{cases}
			(S(k, \bullet) \setminus \set{i}, T(k, \bullet) \setminus \set{j}) & \text{if } (i, j) \in S(k, \bullet) \times T(k, \bullet) \\
			(S(k, \bullet), T(k, \bullet)) & \text{if } (i, j) \notin S(k, \bullet) \times T(k, \bullet).
		\end{cases}
	\]
\end{lemma}

\begin{proof}
	By definition, we have
	\[
		(S|T) = \prod_{k=1}^{\hat{\sigma}_1} (S(k,\bullet) | T(k,\bullet)).
	\]
	Let $A \subseteq [\hat{\sigma}_1]$ be the set of indices given by
	\[
		A \coloneqq \set{k : \text{$i \in S(k,\bullet)$ and $j \in T(k,\bullet)$}}.
	\]
	For $k \in [\hat{\sigma}_1]$, we have
	\[
		\ideg_{x_{i,j}}(S(k,\bullet) | T(k, \bullet)) = \begin{cases}
			1 & k \in A, \\
			0 & \text{otherwise.}
		\end{cases}
	\]
	If $\ideg_{x_{i,j}}(S(k,\bullet) | T(k,\bullet)) = 1$, then expanding the determinant $(S(k,\bullet)|T(k,\bullet))$ by minors gives us
	\[
		\frac{\partial}{\partial x_{i,j}} (S(k,\bullet)|T(k,\bullet)) = \pm (S'(k,\bullet) | T'(k,\bullet)),
	\]
	where $S'(k,\bullet)$ and $T'(k,\bullet)$ are the one-row tableaux obtained by removing $i$ from $S(k,\bullet)$ and $j$ from $T(k,\bullet)$, respectively.
	Note that for $\ell > \ideg_{x_{i,j}}(S(k,\bullet) | T(k,\bullet))$, we have
	\[
		\frac{\partial^\ell}{\partial x_{i,j}^\ell}(S(k,\bullet) | T(k,\bullet)) = 0.
	\]
	Using the product rule (\autoref{lem:product rule}), we then have
	\begin{align*}
		\del{\frac{\partial^d}{\partial x_{i,j}^d}}(S|T) &= \del{\frac{\partial^d}{\partial x_{i,j}^d}}\del{\prod_{k=1}^{\hat{\sigma}_1} (S(k,\bullet) | T(k,\bullet))} \\
		&= \sum_{d_1 + \cdots + d_{\hat{\sigma}_1} = d} \ \prod_{k=1}^{\hat{\sigma}_1} \frac{\partial^{d_k}}{\partial x_{i,j}^{d_k}}(S(k,\bullet) | T(k,\bullet)) \\
		&= \prod_{k=1}^{\hat{\sigma}_1} (\pm(S'(k,\bullet)|T'(k,\bullet))) \\
		&= \pm (S'|T') \neq 0.\qedhere
	\end{align*}
\end{proof}

We now extend the preceding lemma to partial derivatives with respect to multiple variables.

\begin{lemma} \label{lem:bideterminant nonzero partials}
	Let $(S|T)$ be a nonzero bideterminant of shape $\sigma$.
	Let $R \subseteq S(1,\bullet)$ and $C \subseteq T(1,\bullet)$ be subsets of the entries in the first row of $S$ and $T$, respectively, such that $|R| = |C|$.
	Write $R = \set{r_1,\ldots,r_\ell}$ and $C = \set{c_1,\ldots,c_\ell}$.
	Then there are positive integers $\set{d_1,\ldots,d_\ell}$ such that
	\[
		\del{\prod_{i=1}^\ell \frac{\partial^{d_i}}{\partial x_{r_i,c_i}^{d_i}}}((S|T)) \neq 0.
	\]
	In the case $\ch(\F) = 0$, we may take $d_1 = \cdots = d_\ell = 1$.
\end{lemma}

\begin{proof}
	We prove this via induction on $\ell$.
	The case $\ell = 1$ follows from \autoref{lem:bideterminant single partial}.
	When $\ell \ge 2$, let $d_1 \coloneqq \ideg_{x_{r_1,c_1}}(S|T)$.
	\autoref{lem:bideterminant single partial} implies
	\[
		\frac{\partial^{d_1}}{\partial x_{r_1,c_1}^{d_1}}(S|T) = \pm (S'|T'),
	\]
	where $(S',T')$ is the bitableau obtained from $(S,T)$ as in the statement of \autoref{lem:bideterminant single partial}.
	Let $R' \coloneqq R \setminus \set{r_1}$ and $C' \coloneqq C \setminus \set{c_1}$.
	Since $S(1,\bullet) \subseteq S'(1,\bullet) \cup \set{r_1}$ and $R \subseteq S(1,\bullet)$, it follows that $R' \subseteq S'(1,\bullet)$.
	Similarly, we have $C' \subseteq T'(1,\bullet)$.
	By induction, there are positive integers $d_2,\ldots,d_\ell$ such that
	\[
		\del{\prod_{i=2}^\ell \frac{\partial^{d_i}}{\partial x_{r_i,c_i}^{d_i}}}((S'|T')) \neq 0.
	\]
	This implies
	\[
		\del{\prod_{i=2}^\ell \frac{\partial^{d_i}}{\partial x_{r_i,c_i}^{d_i}}}\del{\frac{\partial^{d_1}}{\partial x_{r_1,c_1}^{d_1}}((S|T))} \neq 0,
	\]
	so the fact that partial derivatives commute (\autoref{lem:pd commute}) yields
	\[
		\del{\prod_{i=1}^\ell \frac{\partial^{d_i}}{\partial x_{r_i,c_i}^{d_i}}}((S|T)) \neq 0.
	\]
	If $\ch(\F) = 0$, we also obtain
	\[
		\del{\prod_{i=1}^\ell \frac{\partial}{\partial x_{r_i,c_i}}}((S|T)) \neq 0. \qedhere
	\]
\end{proof}

We now use \autoref{lem:bideterminant nonzero partials} to lower bound the dimension of the space of partial derivatives of any bideterminant.

\begin{proposition} \label{prop:bideterminant partial dimension}
	Let $(S|T)$ be a nonzero bideterminant of width $r$.
	Then $\dim(\partial_{< \infty}((S|T))) \ge \binom{2r}{r}$.
	If $\ch(\F) = 0$, then we also have $\dim(\partial_{\le d}((S|T))) \ge \sum_{i=0}^d \binom{r}{i}^2$.
\end{proposition}

\begin{proof}
	Recall that because $(S|T)$ is of width $r$, we have $|S(1,\bullet)| = |T(1,\bullet)| = r$.
	For sets $R \subseteq S(1,\bullet)$ and $C \subseteq T(1,\bullet)$ with $|R| = |C|$, let $R = \set{r_1,\ldots,r_\ell}$ and $C = \set{c_1,\ldots,c_\ell}$ and define
	\[
		\frac{\partial}{\partial x_{R,C}} \coloneqq \prod_{i=1}^{\ell} \frac{\partial^{d_i}}{\partial x_{r_i,c_i}^{d_i}},
	\]
	where $d_1,\ldots,d_\ell$ are obtained by applying \autoref{lem:bideterminant nonzero partials} to $(S|T)$, $R$, and $C$.
	We will show that 
	\[
		D \coloneqq \Set{\frac{\partial}{\partial x_{R,C}}((S|T)) : R \subseteq S(1,\bullet), C \subseteq T(1,\bullet), |R| = |C|}
	\]
	is a set of linearly independent partial derivatives of $(S|T)$.
	From this, it follows immediately that
	\[
		\dim(\partial_{< \infty}((S|T))) \ge |D| = \sum_{i=0}^{r} \binom{r}{i}^2 = \binom{2 r}{r}
	\]
	and, in the case $\ch(\F) = 0$,
	\[
		\dim(\partial_{\le d}((S|T))) \ge \sum_{i=0}^d \binom{r}{i}^2.
	\]

	It remains to show that the elements of $D$ are linearly independent.
	From \autoref{lem:bideterminant nonzero partials}, we know that $\frac{\partial}{\partial x_{R,C}}((S|T)) \neq 0$.
	It follows from \autoref{lem:multideg partial derivative} that
	\[
		\multideg\del{\frac{\partial}{\partial x_{R,C}} ((S|T))} = \multideg((S|T)) - \del{\sum_{i=1}^\ell d_i \vec{e}_{r_i}} \oplus \del{\sum_{i=1}^\ell d_i \vec{e}_{c_i}}.
	\]
	Let $R' \subseteq S(1,\bullet)$ and $C' \subseteq T(1,\bullet)$ be such that $(R,C) \neq (R',C')$.
	From the above, we have
	\begin{align*}
		&\multideg\del{\frac{\partial}{\partial x_{R,C}}((S|T))} - \multideg\del{\frac{\partial}{\partial x_{R',C'}}((S|T))} \\
		&\quad = \del{\sum_{i=1}^\ell d_i' \vec{e}_{r_i'} - \sum_{i=1}^\ell d_i \vec{e}_{r_i}} \oplus \del{\sum_{i=1}^\ell d_i' \vec{e}_{c_i'} - \sum_{i=1}^\ell d_i \vec{e}_{c_i}}.
	\end{align*}
	Suppose without loss of generality that $R \neq R'$ and that $r_1 \in R \setminus R'$.
	Then the $r_1$ coordinate of $\sum_{i=1}^\ell d_i \vec{e}_{r_i} - \sum_{i=1}^\ell d_i' \vec{e}_{r_i'}$ is nonzero, so 
	\[
		\multideg\del{\frac{\partial}{\partial x_{R,C}}((S|T))} \neq \multideg\del{\frac{\partial}{\partial x_{R',C'}}((S|T))}.
	\]
	The argument when $C \neq C'$ is analogous.
	Thus, the elements of $D$ are nonzero, multihomogeneous, and of distinct multidegree.
	Polynomials of differing multidegree are linearly independent, so this immediately implies that the elements of $D$ are linearly independent as desired.
\end{proof}

We now use \autoref{prop:reduction to single bideterminant} to extend \autoref{prop:bideterminant partial dimension} to all nonzero polynomials in $I_r$.

\begin{theorem} \label{thm:pd lb}
	For every nonzero $f \in \detideal{n}{m}{r}$, we have $\dim(\partial_{< \infty}(f)) \ge \binom{2r}{r} = \dim(\partial_{< \infty}(\det_r))$.
	In the case $\ch(\F) = 0$, we also have $\dim(\partial_{\le d}(f)) \ge \sum_{i=0}^d \binom{r}{i}^2$.
\end{theorem}

\begin{proof}
	Apply \autoref{prop:reduction to single bideterminant} to $f$ to obtain linear functions $\ell_{1,1}(X,\eps),\ldots,\ell_{n,m}(X,\eps) \in \F(\eps)[X]$ such that
	\[
		g(X, \eps) \coloneqq \frac{1}{\eps^q} f(\ell_{1,1}(X,\eps),\ldots,\ell_{n,m}(X,\eps)) = \alpha (K_\sigma | K_\sigma)(X) + O(\eps)
	\]
	for some $q \in \integers$, a nonzero $\alpha \in \F$, and a partition $\sigma$ with $\sigma_1 \ge r$.

	Note that $g(X, 0) = \alpha (K_\sigma | K_\sigma)(X)$.
	By \autoref{prop:bideterminant partial dimension}, we have $\dim(\partial_{<\infty}(g(X,0))) \ge \binom{2\sigma_1}{\sigma_1} \ge \binom{2r}{r}$.
	This implies that $\dim(\partial_{< \infty}(g(X,\eps)) \ge \binom{2r}{r}$.
	Since the change of variables $x_{i,j} \mapsto \ell_{i,j}(X,\eps)$ is an invertible linear transformation over $\F(\eps)$, \autoref{lem:pd linear transformation} implies
	\[
		\dim(\partial_{< \infty}(f)) = \dim(\partial_{< \infty}(g(X,\eps)) \ge \binom{2r}{r}.
	\]
	When $\ch(\F) = 0$, \autoref{prop:bideterminant partial dimension} also yields $\dim(\partial_{\le d}(g(X,0))) \ge \sum_{i=0}^d \binom{2 \sigma_1}{i}^2 \ge \sum_{i=0}^d \binom{r}{i}^2$.
	As above, this extends to a lower bound on $\dim(\partial_{\le d}(g(X,\eps)))$, so using \autoref{lem:pd linear transformation} we get
	\[
		\dim(\partial_{\le d}(f)) = \dim(\partial_{\le d}(g(X,\eps)) \ge \sum_{i=0}^d \binom{r}{i}^2. \qedhere
	\]
\end{proof}

\begin{remark}
	The hypothesis $\ch(\F) = 0$ in the second part of \autoref{thm:pd lb} cannot be avoided in general.
	If $\ch(\F) = p > 0$ and $f \in \detideal{n}{m}{r} \setminus \set{0}$, then $\frac{\partial}{\partial x_i}(f^p) = 0$ for all $i$, so $\dim(\partial_{\le 1}(f^p)) = 1 < 1 + r^2$.
\end{remark}

\section{Hardness Versus Randomness I: Low-Depth Circuits} \label{sec:low depth}

A recent breakthrough of \textcite{LST21a} obtained super-polynomial lower bounds for low-depth algebraic circuits.
Combining their result with the hardness-randomness result of \textcite{CKS19a} yields a deterministic algorithm for identity testing of low-depth algebraic circuits.
Specifically, for every fixed $\eps > 0$, they construct an explicit hitting set generator with seed length $O(n^\eps)$ and degree $O(\log n / \log \log n)$ that hits polynomial-size $o(\log \log \log n)$-depth circuits.

In this section, we give an improved construction of a hitting set generator for low-depth circuits.
For every $k \in \naturals$, we construct a generator with seed length $n^{1/2^k + o(1)}$ and degree $2^k$ that hits polynomial-size $o(\log \log \log n)$-depth circuits.
It follows from \autoref{lem:seed length lb} that the tradeoff between the seed length and degree of our generator is optimal up to the $n^{o(1)}$ factor in the seed length.
Our generator is also computable by a circuit of product-depth $k$ and size $n^{1+o(1)}$ or a circuit of larger product-depth but size $n \log^{O(1)} n$.
As remarked in the introduction, existing techniques in algebraic hardness-randomness produce generators that cannot be computed by circuits smaller than those they hit.
Additionally, our generator hits the closure of small low-depth circuits.
We note that the generator of \textcite{CKS19a}, when instantiated with a polynomial hard for the border of low-depth circuits, can also be shown to hit the closure of low-depth circuits.

Our result can be interpreted as a hardness-randomness framework for low-depth circuits in an aggressive setting of parameters.
In order to instantiate our generator, we need lower bounds on the size of low-depth circuits that compute the determinant, which itself can be computed by small algebraic branching programs.
In contrast, typical hardness-randomness results only require lower bounds for a family of polynomials whose coefficients can be computed explicitly, but the polynomials themselves need not be efficiently computable.
In return for these strong lower bound assumptions, we obtain a generator with parameters that improve on known constructions and are near-optimal in the regime of $n^{\Theta(1)}$ seed length.

\subsection{Making \cite{LST21a} Robust} \label{subsec:robust lst}

In this subsection, we establish that the lower bound of \textcite{LST21a} extends to the border of low-depth circuits.
This essentially follows from the fact that they use a rank-based measure to prove their lower bound.
The extension of lower bounds based on rank measures to the setting of border complexity is a standard observation in algebraic circuit complexity, but we make this explicit for the sake of completeness.
Throughout this subsection, we assume familiarity with the notation and definitions of \cite{LST21a}.

The proof of \textcite{LST21a} proceeds in two steps.
They first establish a lower bound against low-depth set-multilinear circuits.
They then show that a low-depth circuit computing a low-degree set-multilinear polynomial can be made set-multilinear without increasing the depth or size too much.
Combined, this establishes a lower bound against general low-depth circuits.

We first observe that the lower bound against set-multilinear circuits is robust.
To prove their lower bound, they construct from a given polynomial $f(\vec{x})$ a matrix $M_f$ such that $M_f$ has small rank if $f$ can be computed by a small set-multilinear circuit of low depth.

\begin{lemma}[{\cite[Claim 16]{LST21a}}] \label{lem:relrk ub}
	Let $k \ge 10d$ and let $w$ be any word of length $d$ such that the entries of $w$ are $\floor{\alpha k}$ and $-k$ where $\alpha = 1/\sqrt{2}$.
	Then for any $\Delta \ge 1$, any set-multilinear formula $\Phi$ of product-depth $\Delta$ and size $s$ satisfies
	\[
		\relrk_w(\Phi) \le s \cdot 2^{-\frac{k d^{1/(2^\Delta-1)}}{20}}.
	\]
\end{lemma}

Next, they show that $M_f$ has large rank when $f$ corresponds to the iterated matrix multiplication polynomial $\IMM_{n,d}(\vec{x})$.
To do this, they show that a set-multilinear projection of $\IMM_{n,d}(\vec{x})$ has large rank.
This projection behaves nicely in the setting of border complexity, as we describe below.

\begin{lemma}[cf.~{\cite[Lemma 8]{LST21a}}] \label{lem:proj to P_w}
	Let $w \in A^d$ be any word which is $b$-unbiased.
	If there is a set-multilinear circuit computing $\IMM_{2^b,d}(\vec{x}) + O(\eps)$ of size $s$ and product-depth $\Delta$, then there is also a set-multilinear circuit of size $s$ and product-depth $\Delta$ computing $P_w + O(\eps)$ for a polynomial $P_w \in \F_{\text{sm}}[\overline{X}(w)]$ such that $\relrk_w(P_w + O(\eps)) \ge 2^{-b/2}$.
\end{lemma}

\begin{proof}
	\cite[Lemma 8]{LST21a} establishes that such a polynomial $P_w$ can be obtained as a set-multilinear projection of $\IMM_{2^b,d}$.
	Since a nonzero projection of any polynomial in $\eps\F[\eps][\vec{x}]$ remains in $\eps \F[\eps][\vec{x}]$, the same projection takes $\IMM_{2^b,d} + O(\eps)$ to $P_w + O(\eps)$.
	It is clear that such a projection does not increase the size or product-depth of a circuit.
	Further, since this projection is set-multilinear, the set-multilinearity of the circuit is preserved.

	It remains to show that $\relrk_w(P_w + O(\eps)) \ge 2^{-b/2}$.
	\textcite{LST21a} show that $\relrk_w(P_w) \ge 2^{-b/2}$.
	Since $P_w$ can be obtained as a projection of $P_w + O(\eps)$ by setting $\eps = 0$, the lower bound on relative rank extends to $P_w + O(\eps)$ for any error term $O(\eps)$.
\end{proof}

Given the preceding lemmas, we now establish lower bounds on the size of low-depth set-multilinear circuits computing $\IMM_{n,d}(\vec{x})$ in the setting of border complexity.
The proof is analogous to that of \cite[Lemma 15]{LST21a}.

\begin{lemma}[cf.~{\cite[Lemma 15]{LST21a}}] \label{lem:robust sm imm lb}
	Let $n, d, \Delta \in \naturals \setminus\set{0}$ such that $n \ge 4^{10d + 1}$.
	Any set-multilinear circuit $\Phi$ of product-depth $\Delta$ that computes $\IMM_{n,d}(\vec{x}) + O(\eps)$ must have size
	\[
		n^{\Omega\del{\frac{d^{1/(2^\Delta - 1)}}{\Delta}}}.
	\]
\end{lemma}

Given this lower bound, we now implement the second step of \cite{LST21a} by lifting this lower bound to general low-depth circuits.
Let $f(\vec{x})$ be a set-multilinear polynomial.
\textcite{LST21a} lift their lower bound from set-multilinear circuits to general circuits by giving a non-trivial simulation of low-depth circuits by set-multilinear circuits.
Our goal is to perform this same lifting in the border setting: given a low-depth circuit computing $f(\vec{x}) + O(\eps)$, we want to find a low-depth set-multilinear circuit that also computes $f(\vec{x}) + O(\eps)$.

There is a subtle issue in that the error term $O(\eps)$ may not correspond to a set-multilinear polynomial, so we cannot immediately conclude the existence of a low-depth set-multilinear circuit computing $f(\vec{x}) + O(\eps)$.
However, if we allow the error term to change, such a transformation is possible.
Given a low-depth circuit computing $f(\vec{x}) + O(\eps)$, the set-multilinearization procedure of \textcite{LST21a} in fact yields a small, low-depth circuit computing the set-multilinear part of $f(\vec{x}) + O(\eps)$.
This only modifies the error term, which is permissible in our setting.

\begin{lemma}[cf.~{\cite[Proposition 9]{LST21a}}] \label{lem:set-multilinearization}
	Let $s$, $N$, and $d$, be growing parameters with $s \ge Nd$ and let $\Delta \in \naturals$.
	Assume that $\ch(\F) = 0$ or $\ch(\F) > d$.
	If $\Phi$ is a circuit of size at most $s$ and product-depth at most $\Delta$ computing $P + O(\eps)$ for a set-multilinear polynomial $P$ over the sets of variables $(X_1,\ldots,X_d)$ (with $|X_i| \le N$), then there is a set-multilinear circuit $\tilde{\Phi}$ of size $d^{O(d)} \poly(s)$ and product-depth at most $2 \Delta$ computing $P + O(\eps)$.
\end{lemma}

Using \autoref{lem:set-multilinearization}, we now lift \autoref{lem:robust sm imm lb} to a lower bound against low-depth circuits without the set-multilinear restriction.
The proof is identical to that of \cite[Corollary 4]{LST21a}.

\begin{corollary}[cf.~{\cite[Corollary 4]{LST21a}}] \label{cor:LST border lb}
	Let $d \le (\log n)/100$ and suppose either $\ch(\F) = 0$ or $\ch(\F) > d$.
	Any algebraic circuit of product-depth $\Delta$ which computes $\IMM_{n,d}(\vec{x}) + O(\eps)$ must have size at least $n^{d^{\exp(-O(\Delta))}}$.
\end{corollary}

\subsection{Constructing a Hitting Set Generator}

We now use the lower bound of \autoref{cor:LST border lb} to design hitting set generators for the closure of small low-depth circuits.
Of course, a generator with improved parameters can be constructed if one assumes an even stronger lower bound on the size of low-depth circuits needed to compute $\IMM_{n,d}(\vec{x}) + O(\eps)$.
For ease of exposition, we directly instantiate our generator with the lower bound of \autoref{cor:LST border lb}.

The generator of \autoref{cons:matrix generator} will act as a basic building block in our construction.
In order to make use of this generator, we need to extend \autoref{cor:LST border lb} to a lower bound for any non-zero polynomial in the ideal $\detideal{n}{m}{r}$.
This essentially follows by combining \autoref{cor:LST border lb} with \autoref{cor:proj to imm}.

\begin{lemma} \label{lem:const depth det ideal lb}
	There is a universal constant $c_{\ref{lem:const depth det ideal lb}} > 0$ such that the following holds.
	Let $f(X) \in \detideal{n}{m}{r}$ be a nonzero polynomial.
	Assume that either $\ch(\F) = 0$ or $\ch(\F) > \deg(f)$.
	Then any circuit of product-depth $\Delta$ which computes $f(X) + O(\eps)$ must be of size 
	\[
		r^{(\log r)^{\exp(-c_{\ref{lem:const depth det ideal lb}} \Delta)}}.
	\]
\end{lemma}

\begin{proof}
	Without loss of generality, we assume that there is no $n' < n$ such that $f \in \detideal{n'}{m}{r}$ and that there is no $m' < m$ such that $f \in \detideal{n}{m'}{r}$.
	If there is such an $n'$ or $m'$, we may zero out the $n$\ts{th} row (respectively $m$\ts{th} column) of $X$ without affecting the polynomial $f$.
	In particular, we may assume that $f$ depends on at least one variable in each row and column of $X$, so $f$ depends on at least $\max(n,m)$ variables.
	This implies that any circuit computing $f + O(\eps)$ must have size at least $s \ge \max(n,m)$.

	Let $\Phi$ be a circuit of size $s$ and product-depth $\Delta$ that computes $f(X) + O(\eps)$.
	Let $d \coloneqq (\log r)/1000$ and $w \coloneqq r / \log r$.
	Using \autoref{cor:proj to imm}, we obtain a circuit $\Psi(\vec{y})$ of size $s + O(n^2 m^2)$ and product-depth $\Delta$ that computes
	\[
		\Psi(\vec{y}) = \begin{cases}
			\IMM_{w,d}(\vec{y}) + O(\eps) & \text{if } \ch(\F) = 0 \\
			\IMM_{w,d}(\vec{y})^{p^k} + O(\eps) & \text{if } \ch(\F) = p > 0.
		\end{cases}
	\]
	In the case $\ch(\F) = p > 0$, the fact that $\Psi$ is obtained from $\Phi$ by adding a layer of addition gates above and below $\Phi$ implies
	\[
		\deg(f) \ge \deg(\IMM_{w,d}^{p^k}) = d p^k.
	\]
	By assumption, we have $p > \deg(f)$, so $k = 0$.
	That is, we have the equality
	\[
		\Psi(\vec{y}) = \IMM_{w,d}(\vec{y}) + O(\eps)
	\]
	both when $\ch(\F) = 0$ or when $\ch(\F) > \deg(f)$.
	When $r$ is sufficiently large, we have 
	\[
		d = (\log r)/1000 \le (\log r - \log \log r)/100 = (\log w)/100.
	\]
	\autoref{cor:LST border lb} then implies 
	\[
		s + O(n^2 m^2) \ge w^{d^{\exp(-O(\Delta))}} = r^{(\log r)^{\exp(-O(\Delta))}}.
	\]
	Since $n^2 m^2 \le O(s^4)$, we conclude the desired lower bound on $s$.
\end{proof}

Having established border complexity lower bounds against low-depth circuits for all nonzero polynomials in $\detideal{n}{m}{r}$, we now turn to polynomial identity testing.
Using \autoref{lem:const depth det ideal lb}, we show that matrices of low rank are a hitting set for the closure of low-depth circuits.

\begin{lemma} \label{lem:small depth matrix generator}
	Let $\F$ be a field of characteristic zero or characteristic larger than $s^\Delta$.
	Let $\mathcal{G}_{n,m,r}(Y,Z)$ be the generator defined in \autoref{cons:matrix generator}.
	There is a universal constant $c_{\ref{lem:small depth matrix generator}} > 0$ such that for $r = 2^{(\log s)^{1 - \exp(-c_{\ref{lem:small depth matrix generator}} \Delta)}}$, the map $\mathcal{G}_{\sqrt{n},\sqrt{n},r-1}(Y,Z)$ is a hitting set generator for the closure of $n$-variate circuits of size $s$ and product-depth $\Delta$.
\end{lemma}

\begin{proof}
	Let $c_{\ref{lem:const depth det ideal lb}}$ be the constant from \autoref{lem:const depth det ideal lb} and choose $k > 0$ large enough so that $\exp(-c_{\ref{lem:const depth det ideal lb}}) < (1/2)^{\frac{1}{k-1}}$.
	Suppose for the sake of contradiction that the statement of the lemma fails for $c_{\ref{lem:small depth matrix generator}} = k c_{\ref{lem:const depth det ideal lb}}$.
	Then there is a circuit $\Phi$ of size $s$ and product-depth $\Delta$ which computes $f(X) + O(\eps)$ for some nonzero $f(X)$ such that $f(\mathcal{G}_{\sqrt{n},\sqrt{n},r-1}(Y,Z)) = 0$.
	By \autoref{lem:vanish on matrix generator}, we have $f(X) \in \detideal{\sqrt{n}}{\sqrt{n}}{r}$.
	Since $f$ is computed by a circuit of size $s$ and product-depth $\Delta$, we have $\deg(f) \le s^\Delta$, so either $\ch(\F) = 0$ or $\ch(\F) > \deg(f)$.
	The lower bound of \autoref{lem:const depth det ideal lb} implies
	\[
		s \ge 2^{(\log r)^{1 + \exp(-c_{\ref{lem:const depth det ideal lb}} \Delta)}} 
		= 2^{(\log s)^{(1 - \exp(-c_{\ref{lem:small depth matrix generator}} \Delta))(1 + \exp(-c_{\ref{lem:const depth det ideal lb}} \Delta)}}.
	\]
	We claim that
	\[
		(1 - \exp(-c_{\ref{lem:small depth matrix generator}} \Delta))(1 + \exp(-c_{\ref{lem:const depth det ideal lb}} \Delta)) > 1.
	\]
	This would imply $s > s$, a contradiction, which would in turn prove that $\mathcal{G}_{\sqrt{n},\sqrt{n},r-1}(Y,Z)$ is a hitting set generator for the closure.

	To prove this inequality, we first note that it suffices to prove the equivalent
	\[
		\exp(-c_{\ref{lem:const depth det ideal lb}} \Delta) > \exp(-c_{\ref{lem:small depth matrix generator}} \Delta) + \exp(-(c_{\ref{lem:const depth det ideal lb}} + c_{\ref{lem:small depth matrix generator}}) \Delta).
	\]
	By our choice of $c_{\ref{lem:small depth matrix generator}}$ and the fact that $\Delta \ge 1$, we have
	\begin{align*}
		\exp(-c_{\ref{lem:small depth matrix generator}} \Delta) + \exp(-(c_{\ref{lem:const depth det ideal lb}} + c_{\ref{lem:small depth matrix generator}}) \Delta) &< 2 \exp(-c_{\ref{lem:small depth matrix generator}} \Delta) \\
		&= 2\exp(-k c_{\ref{lem:const depth det ideal lb}} \Delta) \\
		&\le 2 \exp(- c_{\ref{lem:const depth det ideal lb}} \Delta) \exp(-(k-1)c_{\ref{lem:const depth det ideal lb}}) \\
		&< \exp(-c_{\ref{lem:const depth det ideal lb}} \Delta),
	\end{align*}
	where the last step follows from our choice of $k$ so that $\exp(-c_{\ref{lem:const depth det ideal lb}}) \le 2^{\frac{1}{k-1}}$.
	This establishes the claimed inequality and completes the proof of the lemma.
\end{proof}

\autoref{lem:small depth matrix generator} constructs a hitting set generator for polynomial-size low-depth circuits with seed length $n^{1/2+o(1)}$ and degree 2.
By \autoref{lem:seed length lb}, this seed length is near-optimal for a degree-two generator.
To obtain hitting set generators with better seed length, we recursively apply the generator of \autoref{lem:small depth matrix generator}.

\begin{theorem} \label{thm:small depth hsg}
	Let $\F$ be a field of characteristic zero.
	For every fixed $k \in \naturals$, there is an explicit hitting set generator $\mathcal{G}_k$ for the closure of $n$-variate, size-$s$, product-depth $\Delta \le o(\log \log \log n)$ circuits such that $\mathcal{G}_k$ has the following properties.
	\begin{enumerate}
		\item
			$\mathcal{G}_k$ has seed length $n^{1/2^k} s^{o(1)}$.
		\item
			$\deg(\mathcal{G}_k) = 2^k$.
		\item
			$\mathcal{G}_k$ can be computed by a circuit of product-depth $k$ and size $ns^{o(1)}$.
			Moreover, each product gate in this circuit has fan-in 2.
		\item
			If $s \le n^{O(1)}$, then $\mathcal{G}_k$ can also be computed by a circuit of size $n \log^{O(1)} n$.
	\end{enumerate}
\end{theorem}

\begin{proof}
	We proceed via induction on $k$.
	In the case $k = 1$, \autoref{lem:small depth matrix generator} establishes the claimed bounds on the seed length and degree of the generator.
	We can compute $\mathcal{G}_1(Y,Z)$ with a circuit of product-depth 1 by na\"{i}vely computing the matrix product $YZ$.
	To improve the circuit complexity of $\mathcal{G}_1$ from $n^{1 + o(1)}$ to $n \log^{O(1)} n$ when $s$ is small enough, we use fast rectangular matrix multiplication.
	Observe that when $s \le n^{O(1)}$, the output of the generator $\mathcal{G}_1(Y,Z)$ is the product of an $\sqrt{n} \times n^{o(1)}$ matrix and an $n^{o(1)} \times \sqrt{n}$ matrix.
	Such a product can be computed in time $n \log^{O(1)} n$ arithmetic operations using Coppersmith's algorithm \cite{Coppersmith82} for rectangular matrix multiplication.
	(For an explanation of why this is the case, see \textcite[Appendix C]{Williams14}.)

	When $k \ge 2$, let $\mathcal{G}_{k-1}(\vec{w})$ be the generator given by induction and let $\Phi$ be a nonzero circuit of size $s$ and product-depth $\Delta$ over $\F(\eps)$.
	By induction, we have that $\mathcal{G}_{k-1}(\vec{w})$ hits $\Phi$ even when $\eps = 0$, so $\Phi(\mathcal{G}_{k-1}(\vec{w})) \neq 0$ and $\Phi(\mathcal{G}_{k-1}(\vec{w})) \notin \eps \F[\eps][\vec{w}]$.
	Further, the composition $\Phi(\mathcal{G}_{k-1}(\vec{w}))$ can be computed by a circuit of product-depth $\Delta + k - 1$ and size $s + n s^{o(1)} \le s^{1 + o(1)}$.

	Let $n_{k-1}$ and $d_{k-1}$ be the seed length and degree, respectively, of $\mathcal{G}_{k-1}(\vec{w})$.
	Arrange the variables $\vec{w}$ into a $\sqrt{n_{k-1}} \times \sqrt{n_{k-1}}$ matrix and let
	\[
		\mathcal{G}_k(Y,Z) \coloneqq \mathcal{G}_{k-1}(\mathcal{G}_{\sqrt{n_{k-1}},\sqrt{n_{k-1}},r_k}(Y,Z)),
	\]
	where
	\[
		r_k = 2^{\log (s^{1 + o(1)})^{1 - \exp(-c_{\ref{lem:small depth matrix generator}}(\Delta+k-1))}}.
	\]
	\autoref{lem:small depth matrix generator} implies that $\mathcal{G}_{\sqrt{n_{k-1}},\sqrt{n_{k-1}},r_k}(Y,Z)$ hits $\Phi(\mathcal{G}_{k-1}(\vec{w}))$ even when $\eps = 0$.
	Equivalently, the composition $\mathcal{G}_k(Y,Z)$ hits $\Phi$ even when $\eps = 0$.
	We now analyze the parameters of $\mathcal{G}_k(Y,Z)$.
	\begin{description}
		\item [Seed length] 
			By definition, the seed length $n_k$ of $\mathcal{G}_k(Y,Z)$ is bounded by
			\[
				n_k \le 2 \sqrt{n_{k-1}} r_k.
			\]
			It follows from induction that $n_{k-1} \le n^{1/2^{k-1}} s^{o(1)}$, so we bound the above as
			\[
				n_k \le 2 n^{1/2^k} s^{o(1)} r_k.
			\]
			We now bound $r_k$.
			As $k$ is fixed and $\Delta \le o(\log \log \log n) \le o(\log \log \log s)$, we have $k + \Delta \le o(\log \log \log s)$.
			This implies
			\[
				\exp(-c_{\ref{lem:small depth matrix generator}}(\Delta + k - 1)) \ge \frac{1}{\exp(o(\log \log \log s))} \ge \omega\del{\frac{1}{\log \log s}}.
			\]
			From this, we obtain
			\[
				(\log(s^{1 + o(1)}))^{1 - \exp(-c_{\ref{lem:small depth matrix generator}}(\Delta+k-1))} \le (\log(s^{1+o(1)}))^{1 - \omega(\frac{1}{\log\log s})} \le \frac{\log s^{1+o(1)}}{\omega(1)} \le o(\log s).
			\]
			By definition, we have
			\[
				r_k = 2^{\log (s^{1 + o(1)})^{1 - \exp(-c_{\ref{lem:small depth matrix generator}}(\Delta+k-1))}} \le 2^{o(\log s)} \le s^{o(1)}.
			\]
			Thus $n_k \le n^{1/2^k} s^{o(1)}$.
		\item [Degree]
			Clearly, we have $\deg(\mathcal{G}_k) = 2 \deg(\mathcal{G}_{k-1})$.
			By induction, $\deg(\mathcal{G}_{k-1}) = 2^{k-1}$, so $\deg(\mathcal{G}_k) = 2^k$.
		\item [Circuit size]
			We can compute $\mathcal{G}_{\sqrt{n_{k-1}},\sqrt{n_{k-1}},r_k}(Y,Z)$ with a circuit of product-depth 1 and size $O(n_{k-1} r_k)$.
			Using induction to bound $n_{k-1}$ and the analysis of the seed length to bound $r_k$, we have $O(n_{k-1} r_k) \le n s^{o(1)}$.
			By induction, we can compute $\mathcal{G}_{k-1}(\vec{w})$ with a circuit of product-depth $k-1$ and size $n s^{o(1)}$.
			Composing these circuits yields a circuit computing $\mathcal{G}_k(Y,Z)$ of product-depth $k$ and size $n s^{o(1)}$.

			If we have the additional assumption that $s \le n^{O(1)}$, then by induction we can compute $\mathcal{G}_{k-1}(\vec{w})$ with a circuit of size $n \log^{O(1)} n$.
			From the analysis of the seed length, we have $r_k \le s^{o(1)} \le n^{o(1)}$.
			Because $k \ge 2$, we have $n_{k-1} \le n^{1/2} s^{o(1)} \le n^{1/2 + o(1)}$.
			By definition, the generator $\mathcal{G}_{\sqrt{n_{k-1}},\sqrt{n_{k-1}},r_k}(Y,Z)$ can be computed by a circuit of size $2 n_{k-1} r_k \le n^{1/2 + o(1)}$.
			Composing this with the circuit computing $\mathcal{G}_{k-1}(\vec{w})$ yields a circuit of size $n \log^{O(1)} n$ that computes the generator $\mathcal{G}_k(Y,Z)$.
			\qedhere
	\end{description}
\end{proof}

\begin{remark}
	While \autoref{lem:small depth matrix generator} holds over fields of sufficiently large positive characteristic, this is not true of \autoref{thm:small depth hsg}.
	This occurs because in our construction of the generator $\mathcal{G}_k$, we apply \autoref{lem:small depth matrix generator} to a polynomial of degree $s^\Delta 2^k$.
	Doing so requires $\ch(\F) > s^\Delta 2^k$ for all $k$, which is not possible for fields of non-zero characteristic.
	Of course, for any fixed $k$, the generator $\mathcal{G}_k$ can be constructed over fields of sufficiently large characteristic.
\end{remark}

\section{Hardness Versus Randomness II: Formulas} \label{sec:formulas}

One can mimic the results of \autoref{sec:low depth} in the setting of algebraic formulas.
While we still lack strong lower bounds for formulas, it seems reasonable to conjecture that neither iterated matrix multiplication nor the determinant can be computed by polynomial-size algebraic formulas.
If we strengthen this assumption to a lower bound against border formula complexity, then we can obtain hitting set generators for the closure of small formulas just as in \autoref{lem:small depth matrix generator} and \autoref{thm:small depth hsg}.
In this section, we describe this construction.

We start by constructing a generator whose correctness is conditional on the hardness of bideterminants for border formulas.
By \autoref{thm:proj to small abp}, such lower bounds are implied by lower bounds on the border formula size of any family of polynomials computable by small ABPs, including iterated matrix multiplication and the determinant.
Phrasing our results in terms of the border formula complexity of bideterminants allows us to derive hardness-to-randomness results for homogeneous formulas as well as general formulas.

First, we show that lower bounds for bideterminants imply the generator $\mathcal{G}_{n,m,r}$ of \autoref{cons:matrix generator}, with appropriate parameters, hits the closure of small formulas.
Recall that for a partition $\sigma$, the $i$\ts{th} row of the tableau $K_\sigma$ consists of $(1, 2, \ldots, \sigma_i)$.

\begin{lemma} \label{lem:formula matrix generator}
	Let $\F$ be an arbitrary field.
	Let $t : \naturals \to \naturals$ be a function such that for every partition $\sigma$, the border formula complexity of $(K_\sigma | K_\sigma)(X)$ is bounded from below by $t(\sigma_1)$.
	Let $\mathcal{G}_{n,m,r}(Y,Z)$ be the generator defined in \autoref{cons:matrix generator}.
	Then $\mathcal{G}_{\sqrt{n},\sqrt{n},t^{-1}(2sn)-1}(Y,Z)$ is a hitting set generator for the closure of $n$-variate formulas of size $s$.

	If $t : \naturals \to \naturals$ instead lower bounds the size of homogeneous formulas computing $(K_\sigma | K_\sigma)(X) + O(\eps)$, then $\mathcal{G}_{\sqrt{n},\sqrt{n},t^{-1}(2sn)-1}$ hits the closure of $n$-variate size-$s$ homogeneous formulas.
\end{lemma}

\begin{proof}
	We first consider non-homogeneous formulas.
	Let $r \coloneqq t^{-1}(2sn)$.
	Suppose for the sake of contradiction that $\mathcal{G}_{\sqrt{n},\sqrt{n},r-1}(Y,Z)$ is a not a hitting set generator for the closure of size-$s$ formulas.
	Then there is some nonzero polynomial $f(X)$ such that $f(\mathcal{G}_{\sqrt{n},\sqrt{n},r-1}(Y,Z)) = 0$ and $f(X) + O(\eps)$ can be computed by a formula of size $s$.
	Since $f(\mathcal{G}_{\sqrt{n},\sqrt{n},r-1}) = 0$, \autoref{lem:vanish on matrix generator} implies that $f \in \detideal{\sqrt{n}}{\sqrt{n}}{r}$.
	By \autoref{prop:reduction to single bideterminant}, there are linear forms $\ell_{1,1}(X,\eps),\ldots,\ell_{\sqrt{n},\sqrt{n}}(X,\eps)$, some nonzero $\alpha \in \F$, an integer $q$, and a partition $\sigma$ with $\sigma_1 \ge r$ such that
	\[
		\frac{1}{\alpha \eps^q} f(\ell_{1,1}(X,\eps),\ldots,\ell_{\sqrt{n},\sqrt{n}}(X,\eps)) = (K_\sigma | K_\sigma)(X) + O(\eps).
	\]
	This yields a formula of size $sn$ that computes $(K_\sigma | K_\sigma)(X) + O(\eps)$.
	This contradicts the assumption that any such formula must be of size at least $t(\sigma_1) \ge t(r) \ge 2sn$.
	Thus $\mathcal{G}_{\sqrt{n},\sqrt{n},r-1}(Y,Z)$ is a hitting set generator for the closure of $n$-variate size-$s$ formulas.

	The homogeneous case is analogous.
	The only difference is that if $f(X)$ is computed by a size-$s$ homogeneous formula, we need to establish that $\frac{1}{\alpha} f(\ell_{1,1}(X,\eps), \ldots, \ell_{\sqrt{n},\sqrt{n}}(X,\eps))$ is computable by a homogeneous formula of size $sn$.
	This follows immediately from the fact that the $\ell_{i,j}(X,\eps) \in \F(\eps)[X]$ are homogeneous linear polynomials in $X$.
\end{proof}

Assuming super-polynomial lower bounds on the border formula complexity of bideterminants, we can recursively apply the generator of \autoref{lem:formula matrix generator} to obtain generators with smaller seed length.
This is analogous to the derivation of \autoref{thm:small depth hsg} from \autoref{lem:small depth matrix generator}.
The only difference is in the analysis, as we now have to compute the generator using formulas, not low-depth circuits.

\begin{proposition} \label{prop:formula recursion superpoly lb}
	Let $\F$ be an arbitrary field.
	Let $t : \naturals \to \naturals$ be a function such that for every partition $\sigma$, the border formula complexity of $(K_\sigma | K_\sigma)(X)$ is bounded from below by $t(\sigma_1)$.
	Assume $t(r) \ge r^{\omega(1)}$.
	Then for every fixed $k \in \naturals$, there is an explicit hitting set generator $\mathcal{G}_k$ for the closure of $n$-variate size-$s$ (homogeneous) formulas with the following properties.
	\begin{enumerate}
		\item
			$\mathcal{G}_k$ has seed length $n^{1/2^k}s^{o(1)}$.
		\item
			$\deg(\mathcal{G}_k) = 2^k$.
		\item
			$\mathcal{G}_k$ can be computed by a homogeneous formula of size $n s^{o(1)}$.
		\item
			If $s \le n^{O(1)}$, then $\mathcal{G}_k$ can also be computed by a circuit of size $n \log^{O(1)} n$.
	\end{enumerate}
\end{proposition}

\begin{proof}
	We use induction on $k$. 
	The claimed bounds on the seed length and degree in the case $k=1$ follow immediately from \autoref{lem:formula matrix generator}.
	To compute the generator $\mathcal{G}_1(Y,Z)$ using a homogeneous formula, we directly write the matrix product $YZ$ as a homogeneous formula.
	If $s \le n^{O(1)}$, then as in the case of low-depth circuits, the generator $\mathcal{G}_1(Y,Z)$ outputs a $\sqrt{n} \times n^{o(1)} \times \sqrt{n}$ matrix product.
	This product can be computed by a circuit of size $n \log^{O(1)} n$ using Coppersmith's algorithm \cite{Coppersmith82} for fast rectangular matrix multiplication.

	When $k \ge 2$, let $\mathcal{G}_{k-1}(\vec{w})$ be the generator given by induction and let $\Phi$ be a nonzero (homogeneous) formula of size $s$.
	By induction, $\mathcal{G}_{k-1}$ hits $\Phi$ even when $\eps = 0$, so $\Phi(\mathcal{G}_{k-1}(\vec{w})) \neq 0$ and $\Phi(\mathcal{G}_{k-1}(\vec{w})) \notin \eps \F[\eps][\vec{w}]$.
	Furthermore, the composition $\Phi(\mathcal{G}_{k-1}(\vec{w}))$ can be computed by a (homogeneous) formula of size $n s^{1 + o(1)}$.

	Let $n_{k-1}$ and $d_{k-1}$ be the seed length and degree, respectively, of $\mathcal{G}_{k-1}$.
	Arrange the variables of $\vec{w}$ into a $\sqrt{n_{k-1}} \times \sqrt{n_{k-1}}$ matrix and let
	\[
		\mathcal{G}_k(Y,Z) \coloneqq \mathcal{G}_{k-1}(\mathcal{G}_{\sqrt{n_{k-1}}, \sqrt{n_{k-1}}, r_k}(Y,Z)),
	\]
	where $r_k \coloneqq t^{-1}(n s^{1 + o(1)})$ and $\mathcal{G}_{n,m,r}(Y,Z)$ is the generator of \autoref{cons:matrix generator}.
	By \autoref{lem:formula matrix generator}, the generator $\mathcal{G}_{\sqrt{n_{k-1}}, \sqrt{n_{k-1}}, r_k}(Y,Z)$ hits the composition $\Phi(\mathcal{G}_{k-1}(\vec{w}))$ even when $\eps = 0$.
	Equivalently, $\mathcal{G}_k(Y,Z)$ hits $\Phi$, even when $\eps = 0$.
	We now analyze the parameters of $\mathcal{G}_k$.

	\begin{description}
		\item [Seed length]
			By construction, $\mathcal{G}_k$ has seed length $2 \sqrt{n_{k-1}} r_k$.
			It follows from induction that $n_{k-1} \le n^{1/2^{k-1}} s^{o(1)}$.
			By assumption, we have 
			\[
				r_k = t^{-1}(n s^{1+o(1)}) \le (ns)^{o(1)} \le s^{o(1)}.
			\]
			This lets us bound the seed length of $\mathcal{G}_k$ by
			\[
				2 \sqrt{n_{k-1}} r_k \le n^{1/2^k} s^{o(1)}
			\]
			as claimed.

		\item [Degree]
			Clearly $\deg(\mathcal{G}_k) = 2 \deg(\mathcal{G}_{k-1})$.
			By induction, we have $\deg(\mathcal{G}_{k-1}) = 2^{k-1}$, so $\deg(\mathcal{G}_k) = 2^k$.
	
		\item [Formula size]
			Each coordinate of $\mathcal{G}_{\sqrt{n_{k-1}},\sqrt{n_{k-1}},r_k}(Y,Z)$ can be computed by a homogeneous formula of size $2 r_k \le s^{o(1)}$.
			By induction, the generator $\mathcal{G}_{k-1}$ can be computed by a homogeneous formula of size $n s^{o(1)}$.
			Composing these formulas gives a homogeneous formula of size $n s^{o(1)}$ that computes $\mathcal{G}_k$.

			If we additionally have $s \le n^{O(1)}$, then each coordinate of $\mathcal{G}_{\sqrt{n_{k-1}},\sqrt{n_{k-1}},r_k}(Y,Z)$ can be computed by a formula of size $n^{o(1)}$.
			Using the fact that $k \ge 2$, this generator has $n_{k-1} \le n^{1/2 + o(1)}$ outputs, so we can compute the generator using a circuit of size $n^{1/2 + o(1)}$.
			By induction, the generator $\mathcal{G}_{k-1}$ can be computed by a circuit of size $n \log^{O(1)} n$.
			Composing these yields a circuit of size $n \log^{O(1)} n$ that computes the generator $\mathcal{G}_k$.
		\qedhere
	\end{description}
\end{proof}

We now relax the hardness assumption of \autoref{prop:formula recursion superpoly lb} using \autoref{thm:proj to small abp}.
This allows us to construct hitting set generators for the closure of small formulas using lower bounds on the border formula complexity of any family of polynomials that can be computed efficiently by algebraic branching programs, including the determinant and iterated matrix multiplication.

\begin{theorem} \label{thm:formula hardness-randomness}
	Let $\F$ be a field of characteristic zero.
	Let $\set{f_n(\vec{x}) : n \in \naturals}$ be a family of $n^{\Theta(1)}$-variate polynomials such that (1) $f_n(\vec{x})$ is computable by algebraic branching programs of size $n^{\Theta(1)}$, and (2) the border formula complexity of $f_n(\vec{x})$ is bounded from below by $n^{\omega(1)}$.
	Then the conclusion of \autoref{prop:formula recursion superpoly lb} holds for formulas; that is, for every fixed $k \in \naturals$, there is an explicit hitting set generator $\mathcal{G}_k$ for the closure of $n$-variate size-$s$ formulas with the following properties.
	\begin{enumerate}
		\item
			$\mathcal{G}_k$ has seed length $n^{1/2^k}s^{o(1)}$.
		\item
			$\deg(\mathcal{G}_k) = 2^k$.
		\item
			$\mathcal{G}_k$ can be computed by a homogeneous formula of size $n s^{o(1)}$.
		\item
			If $s \le n^{O(1)}$, then $\mathcal{G}_k$ can also be computed by a circuit of size $n \log^{O(1)} n$.
	\end{enumerate}
\end{theorem}

\begin{proof}
	Because the determinant is $\VBP$-complete, the assumed lower bound on the border formula complexity of $f_n$ implies that the border formula complexity of $\det_n(X)$ is bounded from below by $n^{\omega(1)}$.
	Thus, it suffices to extend this to a lower bound on the border formula complexity of bideterminants as in the hypothesis of \autoref{prop:formula recursion superpoly lb}.

	Let $X$ be an $n \times m$ generic matrix and let $\sigma$ be a partition. 
	Let $\Phi$ be a formula of size $s$ which computes $(K_\sigma | K_\sigma)(X) + O(\eps)$.
	By using \autoref{cor:proj to det} and converting the resulting circuit into a formula, we obtain a formula of size $O(s n^2 m^2) \le O(s^3)$ which computes $\det_r(X) + O(\eps)$ for $r = \Theta(\sigma_1^{1/3})$.
	Such a formula must be of size $r^{\omega(1)}$.
	This implies $s \ge r^{\omega(1)}$, which in turn yields $s \ge \sigma_1^{\omega(1)}$.
	Hence the hypothesis of \autoref{prop:formula recursion superpoly lb} holds, so we obtain the claimed family of generators.
\end{proof}

\section{Lower Bounds for the Ideal Proof System} \label{sec:ips}

Our final application of \autoref{thm:proj to small abp} is to proof complexity.
We construct an unsatisfiable system of equations $\mathcal{F}$ such that no IPS refutation of $\mathcal{F}$ can be computed by a low-depth circuit of polynomial size.
We also show that if the border formula complexity of the determinant is super-polynomial, then polynomial-size formulas cannot refute $\mathcal{F}$.

In general, one cannot immediately transfer circuit lower bounds to proof complexity lower bounds.
The difficulty in proving lower bounds on the size of IPS refutations lies in the fact that for a given system of equations $\mathcal{F}$, there are many possible refutations of $\mathcal{F}$ and we must prove a lower bound for each of them.
The set of IPS refutations of $\mathcal{F}$ in fact has useful algebraic structure (see \cite[Section 6]{GP18}), but to the best of our knowledge this has not been used successfully in proving IPS lower bounds.

\textcite{FSTW16} developed machinery to derive IPS lower bounds from stronger notions of circuit lower bounds.
Specifically, they showed that circuit lower bounds can be lifted to IPS lower bounds if one can prove circuit lower bounds on either (a) circuits that compute a polynomial $f(\vec{x})$ as a function over the boolean hypercube or (b) circuits that compute any multiple of $f(\vec{x})$.
Using this approach, they proved $\mathcal{C}$-IPS lower bounds for various restricted circuit classes $\mathcal{C}$.

Recent work by \textcite{ST21a} constructed a family of CNF formulas that require IPS refutations of super-polynomial size if and only if $\VP \neq \VNP$.
To the best of our knowledge, this is the first instance where an algebraic circuit lower bound (without further assumptions) is known to imply a lower bound for IPS.
Their result requires the underlying field to be finite; in contrast, we work with fields of characteristic zero, which are necessarily infinite.

Recall that \autoref{thm:proj to small abp} extends circuit lower bounds for $\det(X)$ to circuit lower bounds for the ideal $\detideal{n}{m}{r}$. 
Since $\detideal{n}{m}{r}$ is closed under multiplication by arbitrary polynomials, it is natural to follow the strategy of \cite{FSTW16} and attempt to lift the lower bound for $\detideal{n}{m}{r}$ to an IPS lower bound.
To do this, we need a system of polynomials $f, g_1, \ldots, g_k$ that satisfies the hypothesis of \autoref{lem:fstw} with the additional property that $f \in \detideal{n}{m}{r}$, where $r$ is not too small compared to $n$ and $m$.
Fortunately, such a system is easy to construct.
Recall that for two matrices $A, B \in \F^{n \times m}$, their Hadamard product $A \odot B$ is given by $(A \odot B)_{i,j} \coloneqq a_{i,j} b_{i,j}$.
Let $X$ and $Y$ be two $n \times n$ matrices of variables and $I_n$ be the $n \times n$ identity matrix.
Consider the system
\begin{align*}
	\det_n(X) &= 0 \\
	XY - I_n &= 0 \\
	X \odot X - X &= 0 \\
	Y \odot Y - Y &= 0.
\end{align*}
This system is unsatisfiable, since $\det_n(X) = 0$ implies that $X$ is not invertible, while $XY - I_n = 0$ implies that $X$ is invertible.
However, removing the equation $\det_n(X) = 0$ results in a satisfiable system as witnessed by $X = Y = I_n$.
Thus, this system satisfies the hypotheses of \autoref{lem:fstw} and is a natural candidate for IPS lower bounds.
Note the equations $X \odot X - X = 0$ and $Y \odot Y - Y = 0$ can be removed without affecting the hardness of this system.
These equations enforce boolean constraints on the variables $x_{i,j}$ and $y_{i,j}$, which is the typical setting of proof complexity.

The use of \autoref{lem:fstw} in the preceding sketch only relies on hardness of the ideal generated by $\det_n(X)$, while \autoref{thm:proj to small abp} allows us to lift lower bounds to the larger ideal $\detideal{n}{n}{r}$.
Using a natural generalization of \autoref{lem:fstw}, we can lift a lower bound on the circuit complexity $\detideal{n}{n}{r}$ to an IPS lower bound.
Let $X$ and $Y$ be $n \times n$ matrices of variables and let $r \in \naturals$ such that $r \le n$.
Consider the system of equations given by
\begin{align*}
	\det_r(X_{S,T}) &= 0 \quad\quad \forall S, T \in \binom{[n]}{r} \\
	XY - I_n &= 0 \\
	X \odot X - X &= 0 \\
	Y \odot Y - Y &= 0,
\end{align*}
where $\binom{[n]}{r}$ is the set of size-$r$ subsets of $[n]$.
As in the previous example (which corresponds to the special case of $r = n$), this system of equations is unsatisfiable: the first collection of equations implies that $\rank(X) < r \le n$, while the equation $XY - I_n = 0$ implies that $\rank(X) = n$.
One can show that refuting this system is as hard as computing a nonzero element of $\detideal{n}{n}{r}$, which, in light of \autoref{cor:proj to det}, is as hard as computing the $\Theta(r^{1/3}) \times \Theta(r^{1/3})$ determinant.

As discussed in the introduction, the preceding lower bound is somewhat unsatisfying.
For constant-depth IPS, we would conclude a lower bound of $r^{(\log r)^{\Omega(1)}}$ on the size of a refutation.
However, this system consists of $\binom{n}{r}^2 + 3 n^2$ equations, which is much larger than our lower bound when $r \le (1 - \Omega(1)) n$.
We will use the rank condenser of \textcite{FS12} (\autoref{lem:rank condenser construction}) to give a more succinct encoding of the contradiction ``$\rank(X) < r$ and $\rank(X) = n$.''

Suppose we want to check if a matrix $M \in \F^{n \times n}$ has rank at least $r$ or less than $r$.
Instead of computing all $r \times r$ minors of $M$, we first apply a rank condenser $\mathcal{E}$ to $M$ on the left and the right, resulting in a set of $|\mathcal{E}|$ matrices of size $r \times r$.
By appropriately setting the parameters of the rank condenser, we are guaranteed that one of these $r \times r$ matrices has full rank if and only if $\rank(M) \ge r$.
Thus, to check if $\rank(M) \ge r$, it suffices to compute $\det_r(E M E^\top)$ for all $E \in \mathcal{E}$.

More formally, let $\mathcal{E}$ be a weak $(r,r(n-r))$-lossless rank condenser that satisfies $|\mathcal{E}| = 2 r (n-r) + 1$.
(Observe that such a condenser is given by \autoref{lem:rank condenser construction}.)
Let $X$ and $Y$ be $n \times n$ matrices of variables and consider the system of equations given by
\begin{align*}
	\det_r(E X E^\top) &= 0 \quad \quad \forall E \in \mathcal{E} \\
	XY - I_n &= 0 \\
	X \odot X - X &= 0 \\
	Y \odot Y - Y &= 0,
\end{align*}
This reduces the number of equations in this system from $\binom{n}{r}^2 + 3n^2$ to $3n^2 + 2r(n-r)+1$, which is polynomial in $n$ and $r$ for any choice of $r$.

We now begin by stating and proving a generalization of \autoref{lem:fstw} that will be useful for our IPS lower bounds.
We note that the proof of \autoref{lem:general fstw} below is essentially the same as the proof of \autoref{lem:fstw}.

\begin{lemma} \label{lem:general fstw}
	Let $f_1(\vec{x}),\ldots,f_m(\vec{x}), g_1(\vec{x}), \ldots, g_k(\vec{x}) \in \F[\vec{x}]$ be an unsatisfiable system of equations where $g_1(\vec{x}),\ldots,g_k(\vec{x})$ is satisfiable.
	Let $C \in \F[\vec{x},y,\vec{z}]$ be an IPS refutation of $f_1,\ldots,f_m,g_1,\ldots,g_k$.
	Then $1 - C(\vec{x},\vec{0},g_1(\vec{x}),\ldots,g_k(\vec{x}))$ is a nonzero element of the ideal $\abr{f_1(\vec{x}),\ldots,f_m(\vec{x})} \subseteq \F[\vec{x}]$.
\end{lemma}

\begin{proof}
	Let $C(\vec{x},y_1,\ldots,y_m, z_1,\ldots,z_k)$ be an IPS refutation of $f_1,\ldots,f_m,g_1,\ldots,g_k$ and let $h(\vec{x}) \coloneqq 1 - C(\vec{x},\vec{0},g_1(\vec{x}),\ldots,g_k(\vec{x}))$.
	To see that $h(\vec{x})$ is a nonzero polynomial, let $\vec{\alpha}\in\F^n$ be a point satisfying the equations $g_1(\vec{x}) = \cdots = g_k(\vec{x}) = 0$.
	Recall that because $C(\vec{x},\vec{y},\vec{z})$ is an IPS refutation, we have $C(\vec{x},\vec{0},\vec{0}) = 0$.
	Evaluating at $\vec{x} = \vec{\alpha}$, we have
	\begin{align*}
		h(\vec{\alpha}) &= 1 - C(\vec{\alpha},\vec{0},g_1(\vec{\alpha}),\ldots,g_k(\vec{\alpha})) \\
		&= 1 - C(\vec{x},\vec{0},\vec{0}) \\
		&= 1.
	\end{align*}
	Thus, the polynomial $h(\vec{x})$ is nonzero at $\vec{\alpha}$, so $h(\vec{x})$ is a nonzero polynomial.

	To show that $h(\vec{x})$ is an element of the ideal $\abr{f_1(\vec{x}),\ldots,f_m(\vec{x})}$, we first expand $C(\vec{x},\vec{y},\vec{z})$ as a polynomial in $\vec{y}$ to obtain
	\[
		C(\vec{x},\vec{y},\vec{z}) = \sum_{\vec{e} \in \naturals^m} C_{\vec{e}}(\vec{x},\vec{z}) \vec{y}^{\vec{e}}.
	\]
	We then have
	\begin{align*}
		h(\vec{x}) &= 1 - C(\vec{x},\vec{0},g_1(\vec{x}),\ldots,g_k(\vec{x})) \\
		&= C(\vec{x},f_1(\vec{x}),\ldots,f_m(\vec{x}),g_1(\vec{x}),\ldots,g_k(\vec{x})) - C(\vec{x},\vec{0},g_1(\vec{x}),\ldots,g_k(\vec{x})) \\
		&= \sum_{\vec{e} \in \naturals^m} C_{\vec{e}}(\vec{x},g_1(\vec{x}),\ldots,g_k(\vec{x})) \prod_{i \in [m]}f_i(\vec{x})^{e_i} - C_{\vec{0}}(\vec{x},g_1(\vec{x}),\ldots,g_k(\vec{x})) \\
		&= \sum_{\substack{\vec{e} \in \naturals^m \\ \vec{e} \neq \vec{0}}} C_{\vec{e}}(\vec{x},g_1(\vec{x}),\ldots,g_k(\vec{x})) \prod_{i \in [m]} f_i(\vec{x})^{e_i}.
	\end{align*}
	For each $\vec{e} \in \naturals^m$ such that $\vec{e} \neq \vec{0}$, it is clear that the term $C_{\vec{e}}(\vec{x},g_1(\vec{x}),\ldots,g_k(\vec{x})) \prod_{i \in [m]} f_i(\vec{x})^{e_i}$ is an element of the ideal $\abr{f_1(\vec{x}),\ldots,f_m(\vec{x})}$.
	Because ideals are closed under addition, we see that $h(\vec{x})$ is also an element of the ideal $\abr{f_1(\vec{x}),\ldots,f_m(\vec{x})}$.
\end{proof}

We now show that a weak lossless rank condenser provides a small collection of polynomial equations that encode the statement ``$\rank(X) < r$.''

\begin{lemma} \label{lem:rank condenser determinant}
	Let $\F$ be a field and let $M \in \F^{n \times n}$.
	Let $r \le n$ and let $\mathcal{E} \subseteq \F^{r \times n}$ be a weak $(r, L)$-lossless rank condenser such that $|\mathcal{E}| \ge 2 L + 1$.
	Then $\rank(M) < r$ if and only if $\det_r(E M E^\top) = 0$ for all $E \in \mathcal{E}$.
\end{lemma}

\begin{proof}
	One direction is straightforward: if $\rank(M) < r$, then $\rank(E M E^\top) < r$ for all $E \in \mathcal{E}$.
	In this case, it is immediate that $\det_r(E M E^\top) = 0$ for all $E \in \mathcal{E}$.

	In the other direction, suppose that $\rank(M) \ge r$.
	Let $t \coloneqq \rank(M)$.
	We can write $M$ as $M = AB^T$, where $A, B \in \F^{n \times t}$ are matrices of rank $t$.
	Let $S \subseteq [t]$ be a subset of size $r$ such that $\rank(A_{[n],S}) = r$.
	For a matrix $E \in \F^{r \times n}$, we have $\rank(E A_{[n],S}) \le \rank(E A)$, as the columns of $E A_{[n],S}$ are a subset of the columns of $EA$.
	This implies
	\[
		\Abs{\Set{E \in \mathcal{E} : \rank(EA) < r}} \le \Abs{\Set{E \in \mathcal{E} : \rank(E A_{[n],S}) < r}} \le L,
	\]
	where the second inequality follows from the fact that $\mathcal{E}$ is a weak $(r,L)$-lossless rank condenser.

	A symmetric argument implies
	\[
		\Abs{\Set{E \in \mathcal{E} : \rank(EB) < r}} \le L,
	\]
	so we have
	\[
		\Abs{\Set{E \in \mathcal{E} : \rank(EA) < r \text{ or } \rank(EB) < r}} \le 2 L.
	\]
	Because $|\mathcal{E}| \ge 2 L + 1$, there is some $E \in \mathcal{E}$ such that $\rank(EA) = r$ and $\rank(EB) = r$.
	This implies that 
	\[
		\rank(E M E^\top) = \rank(E A B^\top E^\top) = r.
	\]
	from which it follows that $\det_r(E M E^\top) \neq 0$ as desired.
\end{proof}

Using \autoref{lem:general fstw}, we proceed to lift lower bounds for $\detideal{n}{n}{r}$ to IPS lower bounds in the setting of low-depth circuits.

\begin{theorem} \label{thm:ips vac0 lb}
	Let $X$ and $Y$ be $n \times n$ matrices of variables and let $r \le n$.
	Let $\mathcal{E} \subseteq \F^{r \times n}$ be a weak $(r, r(n-r))$-lossless rank condenser satisfying $|\mathcal{E}| \ge 2 r (n-r) + 1$.
	Assume that 
	\begin{enumerate}
		\item
			if $\ch(\F) = 0$, any product-depth $\Delta$ circuit which computes $\det_n(X) + O(\eps)$ must be of size at least $t(n,\Delta)$ for some function $t : \naturals \times \naturals \to \naturals$; and
		\item
			if $\ch(\F) = p > 0$, any product-depth $\Delta$ circuit which computes $\det_n(X)^{p^k} + O(\eps)$ for any $k \in \naturals$ must be of size at least $t(n,\Delta)$ for some function $t : \naturals \times \naturals \to \naturals$.
	\end{enumerate}
	Let $C(X,Y,\vec{z},W,U,V)$ be an IPS refutation of 
	\begin{align*}
		\det_r(E X E^\top) &= 0 \quad \quad \forall E \in \mathcal{E}\\
		XY - I_n &= 0 \\
		X \odot X - X &= 0 \\
		Y \odot Y - Y &= 0,
	\end{align*}
	where $\vec{z}$ is a collection of placeholder variables corresponding to the first set of equations and $W$, $U$, and $V$ are matrices of placeholder variables that correspond to the second, third, and fourth sets of equations, respectively.
	Then any product-depth $\Delta$ circuit that computes $C(X,Y,\vec{z},W,U,V) + O(\eps)$ must be of size $t(\Omega(r^{1/3}),\Delta+1) - O(n^4)$.
\end{theorem}

\begin{proof}
	Suppose $C+O(\eps)$ can be computed by a circuit of size $s$ and product-depth $\Delta$.
	The system above is unsatisfiable: if $X$ satisfies $\det_r(E X E^\top) = 0$ for all $E \in \mathcal{E}$, then \autoref{lem:rank condenser determinant} implies $\rank(X) < r \le n$; on the other hand, if $X$ satisfies $XY - I_n = 0$, then clearly $\rank(X) = n$.
	Observe that if we omit the equations $\det_r(E X E^\top) = 0$ for all $E \in \mathcal{E}$, then this system becomes satisfiable (take $X = Y = I_n$).
	\autoref{lem:general fstw} implies that
	\[
		f(X,Y) \coloneqq 1 - C(X, Y, 0, XY - I_n, X \odot X - X, Y \odot Y - Y)
	\]
	is a nonzero element of the ideal generated by the polynomials $\det_r(E X E^\top)$ for $E \in \mathcal{E}$.

	Expanding $\det_r(E X E^\top)$ using the Cauchy--Binet formula, we obtain
	\[
		\det_r(E X E^\top) = \sum_{S, T \in \binom{[n]}{r}} \det_r(E_{[r], S}) \det_r(X_{S, T}) \det_r(E^\top_{T,[r]}).
	\]
	Thus, each polynomial $\det_r(E X E^\top)$ lies in the ideal generated by the $r \times r$ minors of $X$.
	Arranging $X$ and $Y$ into an $n \times 2n$ matrix, we have $f(X,Y) \in \detideal{n}{2n}{r}$.
	The coordinates of $XY - I_n$, $X \odot X - X$, and $Y \odot Y - Y$ can be computed by a multi-output circuit of size $O(n^3)$ and product-depth 1.
	This yields a circuit of size $s + O(n^3)$ and product-depth $\Delta + 1$ that computes $f(X,Y) + O(\eps)$.
	Using \autoref{cor:proj to det}, we obtain a product-depth $\Delta + 1$ circuit $\Phi$ of size $s + O(n^4)$ such that
	\begin{enumerate}
		\item
			if $\ch(\F) = 0$, then $\Phi$ computes $\det_{\Theta(r^{1/3})}(X) + O(\eps)$; and
		\item
			if $\ch(\F) = p > 0$, then $\Phi$ computes $\det_{\Theta(r^{1/3})}(X)^{p^k} + O(\eps)$ for some $k \in \naturals$.
	\end{enumerate}
	In both cases, we must have $s + O(n^4) \ge t(\Omega(r^{1/3}),\Delta + 1)$, which completes the proof.
\end{proof}

Since \autoref{cor:LST border lb} establishes unconditional lower bounds on the size of low-depth circuits that border compute elements of $\detideal{n}{2n}{r}$ over fields of characteristic zero, we obtain corresponding lower bounds for low-depth IPS.

\begin{corollary} \label{cor:ips vac0 lb}
	Let $\F$ be a field of characteristic zero.
	Let $\mathcal{E} \subseteq \F^{r \times n}$ be the weak $(r,r(n-r))$-lossless rank condenser of size $|\mathcal{E}| = 2r(n-r) + 1$ given by \autoref{lem:rank condenser construction}.
	Let $X$ and $Y$ be $n \times n$ matrices of variables and let $r \le n$.
	Let $C(X,Y,\vec{z},W,U,V)$ be an IPS refutation of 
	\begin{align*}
		\det_r(E X E^\top) &= 0 \quad\quad \forall E \in \mathcal{E} \\
		XY - I_n &= 0 \\
		X \odot X - X &= 0 \\
		Y \odot Y - Y &= 0,
	\end{align*}
	where $\vec{z}$ is a collection of placeholder variables corresponding to the first set of equations and $W$, $U$, and $V$ are matrices of placeholder variables that correspond to the second, third, and fourth sets of equations, respectively.
	Then any product-depth $\Delta$ circuit that computes $C(X,Y,\vec{z},W,U,V) + O(\eps)$ must be of size $r^{(\log r)^{\exp(-O(\Delta))}} - O(n^4)$.
\end{corollary}

\begin{proof}
	This follows immediately from \autoref{thm:ips vac0 lb} and \autoref{cor:LST border lb}.
\end{proof}

\begin{remark} \label{rem:ac0p frege}
	Over fields of characteristic $p > 0$, bounded-depth IPS can efficiently simulate $\AC^0[p]$-Frege \cite[Theorem 3.5]{GP18}.
	Proving super-polynomial lower bounds on the length of $\AC^0[p]$-Frege proofs is a longstanding open problem in proof complexity.
	\autoref{cor:ips vac0 lb} can be seen as a step towards resolving this problem.
	In order to obtain $\AC^0[p]$-Frege lower bounds, two obstacles must be overcome.
	First, one must extend the lower bound of \textcite{LST21a} to hold over fields of small characteristic and to hold for $p$\ts{th} powers of the determinant.
	Second, it is necessary to prove an IPS lower bound for a system of equations that arises from the encoding of a CNF formula.
	Our system is not the encoding of a CNF; the IPS lower bounds of \textcite{FSTW16} also suffer from this drawback. 
\end{remark}

We can also carry out the reasoning of \autoref{thm:ips vac0 lb} with formulas instead of low-depth circuits.
The resulting formula-IPS lower bound is conditional, as we currently lack good lower bounds on the formula size of any explicit polynomial, let alone the determinant.

\begin{theorem} \label{thm:ips vf lb}
	Let $X$ and $Y$ be $n \times n$ matrices of variables and let $r \le n$.
	Let $\mathcal{E} \subseteq \F^{r \times n}$ be the weak $(r,r(n-r))$-lossless rank condenser of size $|\mathcal{E}| = 2r(n-r) + 1$ given by \autoref{lem:rank condenser construction}.
	Assume that 
	\begin{enumerate}
		\item
			if $\ch(\F) = 0$, any formula which computes $\det_n(X) + O(\eps)$ must be of size at least $t(n)$ for some function $t : \naturals \to \naturals$; and
		\item
			if $\ch(\F) = p > 0$, any formula which computes $\det_n(X)^{p^k} + O(\eps)$ for any $k \in \naturals$ must be of size at least $t(n)$ for some function $t : \naturals \to \naturals$.
	\end{enumerate}
	Let $C(X,Y,\vec{z},W,U,V)$ be an IPS refutation of 
	\begin{align*}
		\det_r(E X E^\top) &= 0 \quad\quad \forall E \in \mathcal{E} \\
		XY - I_n &= 0 \\
		X \odot X - X &= 0 \\
		Y \odot Y - Y &= 0,
	\end{align*}
	where $\vec{z}$ is a collection of placeholder variables corresponding to the first set of equations and $W$, $U$, and $V$ are matrices of placeholder variables that correspond to the second, third, and fourth sets of equations, respectively.
	Then any formula that computes $C(X,Y,\vec{z},W,U,V) + O(\eps)$ must be of size $\Omega\del{\frac{t(\Omega(r^{1/3}))}{n^3}}$.
\end{theorem}

\begin{proof}
	Suppose $C + O(\eps)$ can be computed by a formula of size $s$.
	As in the proof of \autoref{thm:ips vac0 lb}, we deduce from \autoref{lem:general fstw} that
	\[
		f(X,Y) \coloneqq 1 - C(X, Y, 0, XY - I_n, X \odot X - X, Y \odot Y - Y)
	\]
	is a nonzero element of the ideal $\detideal{n}{2n}{r} \subseteq \F[X,Y]$, where we view $X \cup Y$ as an $n \times 2n$ matrix.
	The coordinates of $XY - I_n$, $X \odot X - X$, and $Y \odot Y - Y$ can each be computed by a formula of size $O(n)$.
	This yields a formula of size $O(sn)$ that computes $f(X,Y) + O(\eps)$.
	From \autoref{cor:proj to det}, we obtain a formula $\Phi$ of size $O(sn^3)$ such that
	\begin{enumerate}
		\item
			if $\ch(\F) = 0$, then $\Phi$ computes $\det_{\Theta(r^{1/3})}(X) + O(\eps)$; and
		\item
			if $\ch(\F) = p > 0$, then $\Phi$ computes $\det_{\Theta(r^{1/3})}(X)^{p^k} + O(\eps)$ for some $k \in \naturals$.
	\end{enumerate}
	By assumption, we must have $O(sn^3) \ge t(\Omega(r^{1/3}))$, which implies the desired lower bound on $s$.
\end{proof}

The previous results show that in the setting of border complexity, the task of computing the $\Theta(r^{1/3}) \times \Theta(r^{1/3})$ determinant can be reduced to computing any IPS refutation of the system $\mathcal{F} \coloneqq \Set{\det_r(E X E^\top) = 0 : E \in \mathcal{E}} \cup \set{XY - I_n = 0}$ where $\mathcal{E}$ is a weak lossless rank condenser with appropriate parameters and size.
When $r = n$, we complement this by constructing a depth-three $\det_n$-oracle circuit with $O(n^2)$ wires that computes a refutation of this system.
This shows that the complexity of refuting $\set{\det(X) = 0, XY - I_n = 0}$ is equivalent to the complexity of approximately computing the determinant, up to polynomial factors.

\begin{proposition} \label{prop:ips ub}
	Let $\F$ be any field and let $X$ and $Y$ be $n \times n$ matrices of variables.
	Let $\mathcal{F}$ be the system of equations given by
	\begin{align*}
		\det(X) &= 0 \\
		XY - I_n &= 0.
	\end{align*}
	Then the following hold.
	\begin{enumerate}
		\item
			There is a depth-three $\det_n$-oracle circuit with $O(n^2)$ wires that computes an IPS refutation of $\mathcal{F}$.
		\item
			There is a depth-three $(\det_n + O(\eps))$-oracle circuit with $O(n^2)$ wires that approximately computes an IPS refutation of $\mathcal{F}$.
	\end{enumerate}
\end{proposition}

\begin{proof}
	Item (2) follows immediately from (1) using \autoref{lem:exact to approx oracle}, so it suffices to prove (1).
	Let $Z$ be an $n \times n$ matrix of variables and let $w$ be an additional variable.
	We claim that
	\[
		C(X,Y,w,Z) \coloneqq 1 - \det_n(Z + I_n) + w \cdot \det_n(Y)
	\]
	is an IPS refutation of $\mathcal{F}$.
	It is clear from the expression above that $C(X,Y,w,Z)$ can be computed by a depth-three $\det_n$-oracle circuit with $O(n^2)$ wires.
	To see that $C(X,Y,w,Z)$ is a valid IPS refutation, observe that we have
	\[
		C(X,Y,0,0) = 1 - \det_n(I_n) + 0  = 0
	\]
	and
	\begin{align*}
		C(X, Y, \det_n(X), XY - I_n) &= 1 - \det_n(XY - I_n + I_n) + \det_n(X) \det_n(Y) \\
		&= 1 - \det_n(XY) + \det_n(XY) \\
		&= 1.
	\end{align*}
	Thus $C(X,Y,w,Z)$ is an IPS refutation of $\mathcal{F}$.
\end{proof}

We end with a brief discussion on the hard instance used in this section.

\begin{remark}
	\textcite[Example A.6]{GP18} showed that a short IPS refutation of $\set{\det_n(X) = 0, XY - I_n = 0}$ can be used to construct a short IPS proof of the inversion principle $XY = I_n \implies YX = I_n$.
	The inversion principle is one of the ``hard matrix identities'' of \textcite{SC04}, which are four tautologies proposed as candidates for separating the Frege and Extended Frege proof systems.
	Unfortunately, our methods are not able to prove lower bounds, conditional or otherwise, on the size of IPS proofs of the hard matrix identities.
\end{remark}

\section*{Acknowledgments}

We thank Tuomas Hakoniemi and Iddo Tzameret for useful conversations that led to the proof of \autoref{lem:general fstw} and for allowing us to include \autoref{lem:general fstw} in this work.

\printbibliography

\end{document}